\documentclass{sig-alternate}
\usepackage{algorithm,algpseudocode,cite,amsmath,amssymb,amsfonts,dsfont,multirow,multicol,enumitem,epsfig,url,array,makecell}
\usepackage{times}
\usepackage{hyperref}
\usepackage{tikz}
\newcommand*\circled[1]{\tikz[baseline=(char.base)]{\node[shape=circle,draw,inner sep=0.3pt] (char) {#1};}}

\newtheorem{definition} {Definition}

\newtheorem{theorem} {Theorem}
\newtheorem{lemma} {Lemma}

\newtheorem{assumption} {Assumption}

\def\header{\vspace{2mm} \noindent}

\def\figcapup{\vspace{-0mm}}
\def\figcapdown{\vspace{-0mm}}
\def\tblcapup{\vspace{-0mm}}
\def\tblcapdown{\vspace{2mm}}
\def\tbldown{\vspace{-0mm}}
\def\algocapup{\vspace{-3mm}}
\def\algocapdown{\vspace{-2mm}}

\def\r{\alpha}
\def\H{H^*}
\def\G{G^*}
\def\a{\lambda}

\begin{document}
\begin{sloppy}

\title{Shortest Path and Distance Queries on Road Networks: Towards Bridging Theory and Practice}

\numberofauthors{1}
\author{
Andy Diwen Zhu$^\dag$ $\:\:$ Hui Ma$^\dag$ $\:\:$ Xiaokui Xiao$^\dag$ $\:\:$ Siqiang Luo$^\S$ $\:\:$ Youze Tang$^\dag$ $\:\:$ Shuigeng Zhou$^\S$ \\
\and
\alignauthor
\affaddr{$^\dag$Nanyang Technological University $\qquad \qquad \qquad \qquad \qquad \qquad$ $^\S$Fudan University} \\
\affaddr{$\qquad$ Singapore $\qquad \qquad \qquad \qquad \qquad \qquad \qquad \qquad \qquad \qquad$ China} \\
\email{$^\dag$\{dwzhu, mahui, xkxiao, yztang\}@ntu.edu.sg $\qquad$ $^\S$\{sqluo, sgzhou\}@fudan.edu.cn  $\qquad \qquad$}
% \alignauthor
% \affaddr{$^2$Nanyang Technological University} \\
% \affaddr{Singapore} \\
% \email{xkxiao@ntu.edu.sg}
% \affaddr{$^3$University of Utah} \\
% \affaddr{Salt Lake City, Utah, USA} \\
% \email{$\{$lifeifei$\}$@cs.utah.edu}
}

\maketitle
\begin{abstract}
Given two locations $s$ and $t$ in a road network, a distance query returns the minimum network distance from $s$ to $t$, while a shortest path query computes the actual route that achieves the minimum distance. These two types of queries find important applications in practice, and a plethora of solutions have been proposed in past few decades. The existing solutions, however, are optimized for either practical or asymptotic performance, but not both. In particular, the techniques with enhanced practical efficiency are mostly heuristic-based, and they offer unattractive worst-case guarantees in terms of space and time. On the other hand, the methods that are worst-case efficient often entail prohibitive preprocessing or space overheads, which render them inapplicable for the large road networks (with millions of nodes) commonly used in modern map applications.

This paper presents {\em Arterial Hierarchy (AH)}, an index structure that narrows the gap between theory and practice in answering shortest path and distance queries on road networks. On the theoretical side, we show that, under a realistic assumption, AH answers any distance query in $\tilde{O}(\log \r)$ time, where $\r = d_{max}/d_{min}$, and $d_{max}$ (resp.\ $d_{min}$) is the largest (resp.\ smallest) $L_\infty$ distance between any two nodes in the road network. In addition, any shortest path query can be answered in $\tilde{O}(k + \log \r)$ time, where $k$ is the number of nodes on the shortest path. On the practical side, we experimentally evaluate AH on a large set of real road networks with up to twenty million nodes, and we demonstrate that (i) AH outperforms the state of the art in terms of query time, and (ii) its space and pre-computation overheads are moderate.
\end{abstract}

%In addition, AH requires only $O(h n)$ space, and it incurs $O(h n^2)$ pre-computation cost.

%============================================
\section{Introduction} \label{sec:intro}

Given two locations $s$ and $t$ in a road network, a distance query returns the network distance from $s$ to $t$, while a shortest path query computes the actual route that achieves the minimum distance. These two types of queries find important applications in map, navigation, and location-based services. To illustrate, consider that a user of a map service is looking for a nearby Italian restaurant for dinner. In response to the user's query, the service provider can first retrieve the list of Italian restaurants in the region close to the user's current location $u$. After that, the network distance from $u$ to each restaurant is computed (using a distance query), and those distances are returned to the user along with the list of restaurants. Then, if the user chooses a preferred restaurant $r$ from the list, the service provider can employ a shortest path query to provide the user with driving directions from $u$ to $r$.

The classic solution for shortest path and distance queries is Dijkstra's algorithm \cite{d59}. It traverses the road network nodes in ascending order of their distances from $s$; once it reaches $t$ during the traversal, it can compute the distance from $s$ to $t$ and can retrieve the shortest path based on the information recorded before $t$ is visited. With proper data structures, Dijkstra's algorithm runs in $O(n \log n + m)$ time for any shortest path or distance query, where $n$ (resp.\ $m$) is the number of nodes (resp.\ edges) in the road network. Albeit simple and elegant, Dijkstra's algorithm is inefficient for sizable road networks, as it requires traversing all network nodes that are closer to $s$ than $t$, which incurs a significant overhead when $s$ and $t$ are far part.

%p71

A plethora of techniques \cite{ssa08,ssa09,ss10,gss08,bfm06,tsp11,rt10,gh05,gkw06,kmw10,hkm06,dss09,fr06,afg10,bfs07,jp02,jhr98,gkr04,ms12} have been proposed to improve over Dijkstra's algorithm in terms of either practical efficiency or asymptotic bounds. The existing methods that focus on practical performance mostly rely on heuristics, and hence, their asymptotic bounds are unattractive in general. For instance, the best heuristic approach by Geisberger et al.\ \cite{gss08} answers shortest path or distance queries in at most a few milliseconds even on road networks with millions of nodes, but its space and time complexities are both $O(n^2)$, i.e., its asymptotic performance is even worse than that of Dijkstra's algorithm. On the other hand, the solutions that offer favorable query time complexities often entail prohibitive preprocessing cost or space overhead, rendering them only applicable for small datasets. For example, the state-of-the-art approaches by Samet et al.\ \cite{ssa08} and Abraham et al.\ \cite{afg10} provide superior bounds on query time, but they require pre-computing the shortest path between {\em any} pair of nodes, which is impractical for the large road networks commonly used in modern map applications.

\header {\bf Contributions.} This paper presents {\em Arterial Hierarchy (AH)}, an index structure that narrows the gap between theory and practice in answering shortest path and distance queries on road networks. On the theoretical side, we show that, under a realistic assumption, AH answers any distance query in $\tilde{O}(\log \r)$ time, where $\r = d_{max}/d_{min}$, and $d_{max}$ (resp.\ $d_{min}$) is the largest (resp.\ smallest) $L_\infty$ distance between any two nodes in the road network. In addition, any shortest path query can be answered in $\tilde{O}(k + \log \r)$ time, where $k$ is the number of nodes on the shortest path.
%In addition, AH requires only $O(h n)$ space, and it incurs $O(h n^{1.5})$ pre-computation cost.
On the practical side, we experimentally evaluate AH on a large set of real road networks with up to twenty million nodes, and we demonstrate that (i) AH outperforms the state of the art in terms of query time, and (ii) its space and pre-computation overheads are moderate.

In a nutshell, AH organizes the nodes in the road network into a hierarchy, based on which it pre-computes auxiliary information to facilitate query processing. For instance, given the road network $G$ in Figure~\ref{fig:intro-original}, AH constructs a three-level hierarchy $H$ (illustrated in Figure~\ref{fig:intro-hierarchy}), where each level consists of a disjoint subset of the nodes in $G$. Note that $H$ contains all edges in $G$, as well as two auxiliary edges, $\langle v_9, v_{10} \rangle$ and $\langle v_{10}, v_{11} \rangle$, each of which has a length that equals the distance between the two nodes that it connects. These auxiliary edges are referred as {\em shortcuts}, and they can be exploited to considerably reduce the numbers of nodes and edges that need to be traversed during a shortest path or distance query.

For example, given a distance query between $v_1$ and $v_{10}$ (in $G$), AH would perform two alternating traversals (in $H$) starting from $v_1$ and $v_{10}$, respectively, and {\em it would always avoid traveling from a higher-level node to a lower-level node}. In particular, the traversal starting from $v_1$ can only reach two nodes, $v_{10}$ and $v_{11}$, since (i) $v_{11}$ is the only node adjacent to $v_1$, and (ii) from $v_{11}$, AH would only traverse to $v_{10}$ (since $v_{10}$ is the only neighbor of $v_{11}$ that is not at a lower level than $v_{11}$). Similarly, the traversal starting from $v_{10}$ would only reach $v_{11}$. Once the two traversals terminate, the distance between $v_1$ and $v_{10}$ is calculated by summing up the weights of $\langle v_1, v_{11} \rangle$ and $\langle v_{11}, v_{10} \rangle$.

In general, AH answers any distance query with two traversals of the node hierarchy, such that each traversal only moves up from low-level nodes to high-level nodes, but not vice versa. We show that, for real road networks, the node hierarchy contains $O(\log \r)$ levels, where $\r = d_{max}/d_{min}$.
%, and $d_{max}$ ($d_{min}$) is the largest (smallest) $L_\infty$ distance between any two nodes in the road network.
Furthermore, each traversal performed by AH visits only a constant number of nodes and edges in any level of the hierarchy. As a consequence, the total number of nodes and edges visited by AH is $O(\log \r)$, which results in an $\tilde{O}(\log \r)$ time complexity for any distance query. In addition, once the distance between two nodes $s$ and $t$ is computed, AH can derive the actual shortest path from $s$ to $t$ in $O(k)$ time, where $k$ is the number of nodes on the shortest path.

%The aforementioned time complexities of AH rely on two assumptions on road networks (to be clarified in Section~\ref{sec:def}). We provide detailed discussion on our assumptions, and we demonstrate their applicability on practical road networks with extensive experiments on a large collection of real datasets. These experimental findings not only form a basis for our theoretical claims but also shed light on the characteristics of real road networks, which paves the path for future research on shortest path and distance queries.

The aforementioned time complexities of AH rely on an assumption on road networks (to be clarified in Section~\ref{sec:def}). We provide detailed discussion on the assumption, and we demonstrate its applicability on practical road networks with extensive experiments on a large collection of real datasets. These experimental findings not only form a basis for our theoretical claims but also shed light on the characteristics of real road networks, which paves the path for future research on shortest path and distance queries.

%The remainder of the paper is organized as follows. [To be added]

\begin{figure}[t]
%\begin{small}
\begin{minipage}[t]{1.63 in}
\centering
\includegraphics[width = 1.62in]{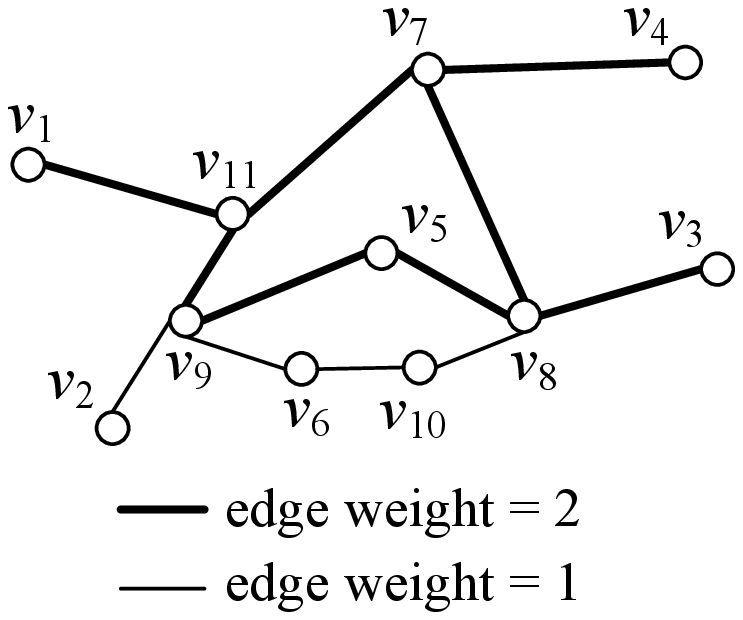}
\vspace{-3.52mm}
\figcapup \caption{Road network $G$ $\quad$ (with bidirectional edges).} \figcapdown \label{fig:intro-original}
\end{minipage}
\begin{minipage}[t]{1.85 in}
\centering
%\hspace{3mm}
\includegraphics[width = 1.70in]{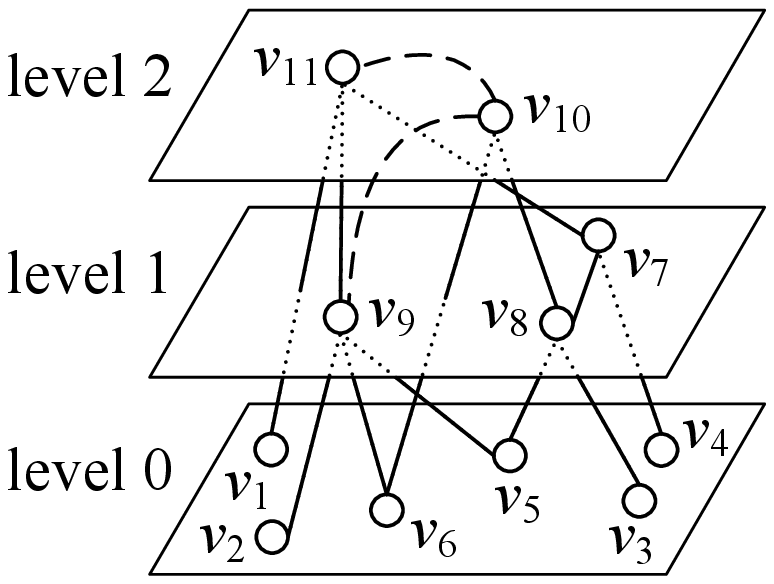}
\figcapup \caption{Node hierarchy $H$. $\quad$} \figcapdown \label{fig:intro-hierarchy}
\end{minipage}
%\end{small}
\end{figure}

\section{Problem and Assumptions}\label{sec:def}

Let $G$ be a road network. We assume that $G$ is a directed, degree-bounded, and connected graph with a node set $V$ and an edge set $E$, such that (i) $|V| = n$, (ii) each node in $V$ locates in a two-dimensional space, and (iii) each edge $e \in E$ is associated with a positive weight $l(e)$. Without loss of generality, we consider that $l(e)$ equals the length of  $e$.
%, and we assign to each edge in $E$ a unique ID.
For any path $P$ in $G$, we define its length $l(P)$ as the total length of the edges in $P$.

%For simplicity, in all of our figures, we use an undirected line between two nodes to represent a bidirectional edge connecting the nodes.

We study two types of queries on $G$, namely, shortest path queries and distance queries. Given an ordered pair of nodes $(s, t) \in V \times V$, a shortest path query asks for a sequence of edges $e_1, e_2, \ldots, e_k$ that form a path from $s$ to $t$, such that $\sum_{i = 1}^k l(e_i)$ is minimized. On the other hand, a distance query from $s$ to $t$ asks only for the value of $\sum_{i = 1}^k l(e_i)$ instead of the actual shortest path. For convenience, we define the {\em distance} from $s$ to $t$ as $dist(s, t) = \sum_{i = 1}^k l(e_i)$.

Our solution for shortest path and distance queries is developed based on an observation on the properties of real road networks, as explained in the following.

\header {\bf Observation.} Assume that we impose a square grid $R$ on $G$. Let $B$ be a region containing $4 \times 4$ cells in the grid. We define the left-most (resp.\ right-most) column of cells in $B$ as the {\em west strip} (resp.\ {\em east strip}) of $B$, and we refer to the vertical line that evenly divides $B$ as the {\em vertical bisector} of $B$. We also define $B$'s {\em north strip}, {\em south strip}, and {\em horizontal bisector} in a similar manner. For example, Figure~\ref{fig:def-grid} illustrates (i) a square grid imposed on the road network in Figure~\ref{fig:intro-original}, (ii) a region $B$ covering $4 \times 4$ grid cells, and (iii) the strips and bisectors of $B$.

We observe that, in practice, the shortest paths between the west and east strips of $B$ can often be {\em covered} by a small set $S_{we}$ of road network edges intersecting $B$'s vertical bisector. That is, given any two points in $B$'s west and east strips, respectively, the shortest path between the two points should pass through at least one edge in $S_{we}$. For instance, suppose that $B$ covers the area of a state. In that case, any shortest path $P$ between the west and east strips of $B$ corresponds to a route that connects the west and east ends of the state. Intuitively, $P$ would have to pass through some major intra-state highways. Therefore, if $S_{we}$ contains the road network edges on the intra-state highways that intersect $B$'s vertical bisector, then $S_{we}$ should cover any aforementioned shortest path $P$. Furthermore, the cardinality of $S_{we}$ should be small, as there should exist only a handful of major highways in the state that go across the vertical bisector. Similar statements can be made even when $B$ corresponds to a larger region (e.g., a continent) or a smaller one (e.g., a city). In addition, we also observe that all shortest paths between the north and south strips of $B$ can be covered by a few edges intersecting $B$'s horizontal bisector.

The above observations are similar in spirit to those made in previous work \cite{afg10,bfm06,ssa09}, which all illustrate that there exists a small set of {\em important} road network edges or nodes that cover all shortest paths connecting distant regions (see Section~\ref{sec:related} for a survey of related work). In what follows, we will formalize our observations and provide empirical evidence, so as to form a basis for further discussions in Sections \ref{sec:fc} and \ref{sec:ah}.

\header {\bf Formalization.} Given a region $B$ of $4\times4$ grid cells, we say that a road network path $P$ is a {\em local path} in $B$, if at most one edge in $P$ intersects the boundary of $B$. For instance, in Figure~\ref{fig:def-grid}, the paths $\langle v_9, v_5, v_8 \rangle$ and $\langle v_{11}, v_7, v_4 \rangle$ are both local paths in $B$. A local path in $B$ is the {\em shortest}, if it is shorter than any other local path in $B$ with the same endpoints. For simplicity, we assume that no two local paths in $B$ share the same endpoints and have the same length -- This assumption can be enforced by adding a small perturbation to each edge in $G$, as shown in Appendix~\ref{sec:perturb}. %(Interested readers are referred to our online and anonymous technical report \cite{AhReport} for a thorough treatment of this issue.)

We are interested in the local shortest paths between opposite strips of $B$, and a set of edges on $B$'s bisectors that cover all such paths, as defined in the following.
\begin{definition}[Spanning Paths \& Arterial Edges] \label{def:def-ae}
A local shortest path $P$ in $B$ is a {\bf spanning path} of $B$, if (i) the two endpoints of $P$ are on different sides of a bisector of $B$ (denoted as $l_b$), and (ii) neither of the endpoints is contained in a grid cell adjacent to the bisector $l_b$. Any edge on $P$ that intersects $l_b$ is an {\bf arterial edge} of $B$.
\end{definition}
By Definition~\ref{def:def-ae}, the path $P = \langle v_9, v_6, v_{10}, v_8\rangle$ in Figure~\ref{fig:def-grid} is a spanning path of $B$, since (i) $P$ is a local shortest path of $B$, (ii) $v_9$ and $v_8$ are on different sides of $B$'s vertical bisector, and (iii) neither $v_9$ nor $v_8$ is in a grid cell adjacent to the bisector. Accordingly, the edge $\langle v_6, v_{10}\rangle$ is an arterial edge of $B$, as it is the only edge in $P$ that intersects $B$'s vertical bisector. Likewise, $\langle v_{11}, v_7, v_4\rangle$ is also a spanning path of $B$, and $\langle v_{11}, v_7\rangle$ is an arterial edge of $B$.

As explained previously, the number of arterial edges in a $(4{\times}4)$-cell region $B$ tends to be small in practice, since there usually exist only a few major connections between opposite strips of $B$. We formalize this observation as follows.

\begin{assumption}[Arterial Dimension] \label{assu:def-ae}
For any square grid on $G$ and any region $B$ with $4 \times 4$ grid cells, the number of arterial edges of $B$ is at most a constant $\lambda$, referred to as the {\bf arterial dimension} of $G$.
\end{assumption}

To demonstrate the applicability of Assumption~\ref{assu:def-ae}, we conduct an experiment on eight real datasets that represent various parts of the road network in the United States (see Section~\ref{sec:exp} for details). The weight of each edge in the data equals the time required to travel between the two endpoints of the edge. On each dataset, we impose a $2^r \times 2^r$ square grid ($r \in [3, 17]$), and compute the number of arterial edges for each ($4{\times}4$)-cell region (ignoring the regions that are empty).
%do not contain any road network node).
After that, we compute the maximum number of arterial edges for a region, as well as the mean, $90\%$ quantile, and $99\%$ quantile. Figure~\ref{fig:def-ae-num} plots the results as functions of the grid resolution $r$. Regardless of the grid resolution and the dataset size, the maximum number of arterial edges for a ($4 {\times} 4$)-cell region is at most $97$, and is below $60$ in most cases. Furthermore, the $90\%$ and $99\%$ quantiles are at most $60$, while the mean is never above $22$. This indicates that practical road networks have fairly small arterial dimensions. In Sections \ref{sec:fc} and \ref{sec:ah}, we will exploit this fact to construct efficient indices for shortest path and distance queries.

%property of road networks

\begin{figure*}[!t]
\begin{small}
\begin{tabular}{cccc}
\multicolumn{4}{c}{\hspace{-4mm} \includegraphics[height=3.6mm]{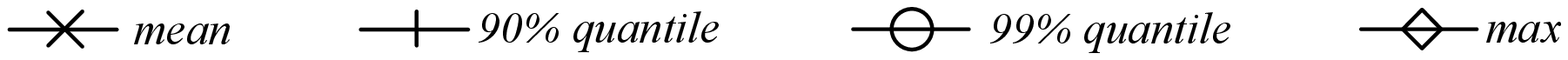}} \\ \vspace{-2.5mm} \\
\hspace{-9.5mm}\includegraphics[height=38mm]{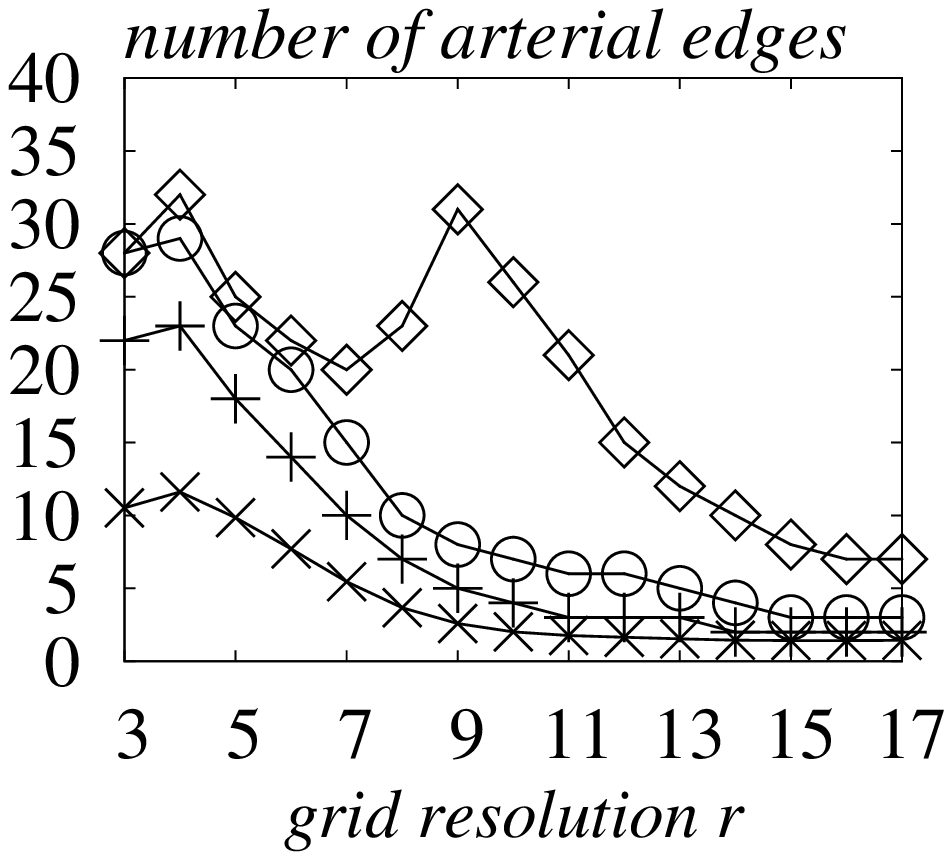}
&
\hspace{-12.5mm}\includegraphics[height=38mm]{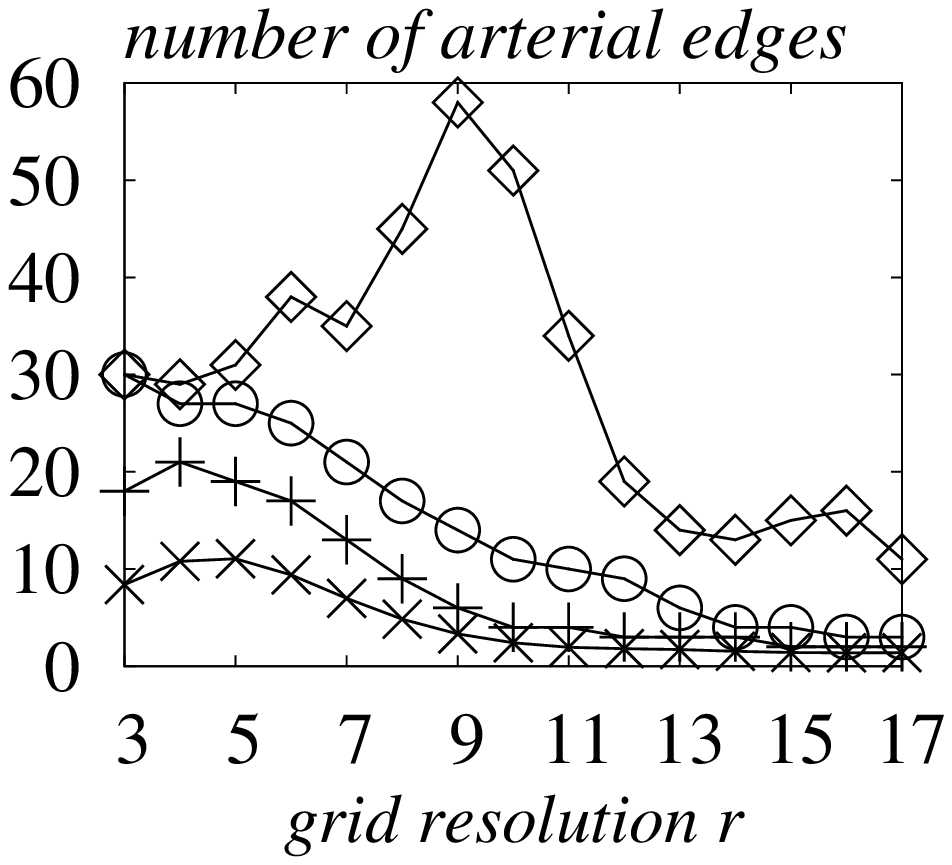}
&
\hspace{-12.5mm}\includegraphics[height=38mm]{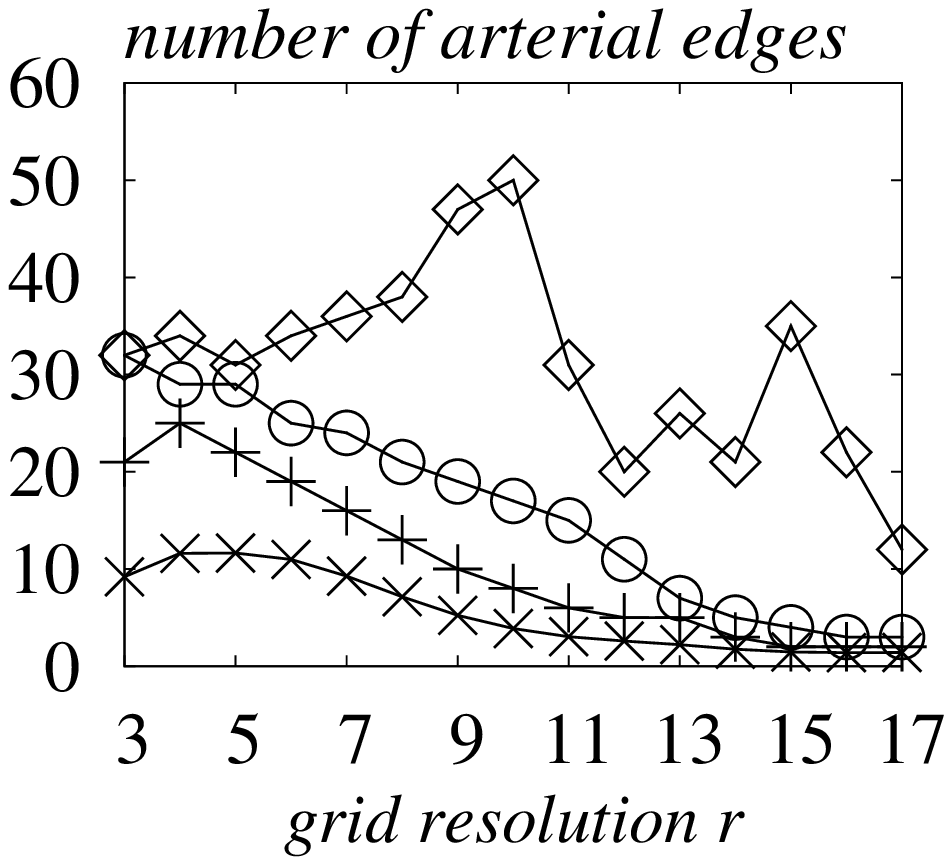}
&
\hspace{-12.5mm}\includegraphics[height=38mm]{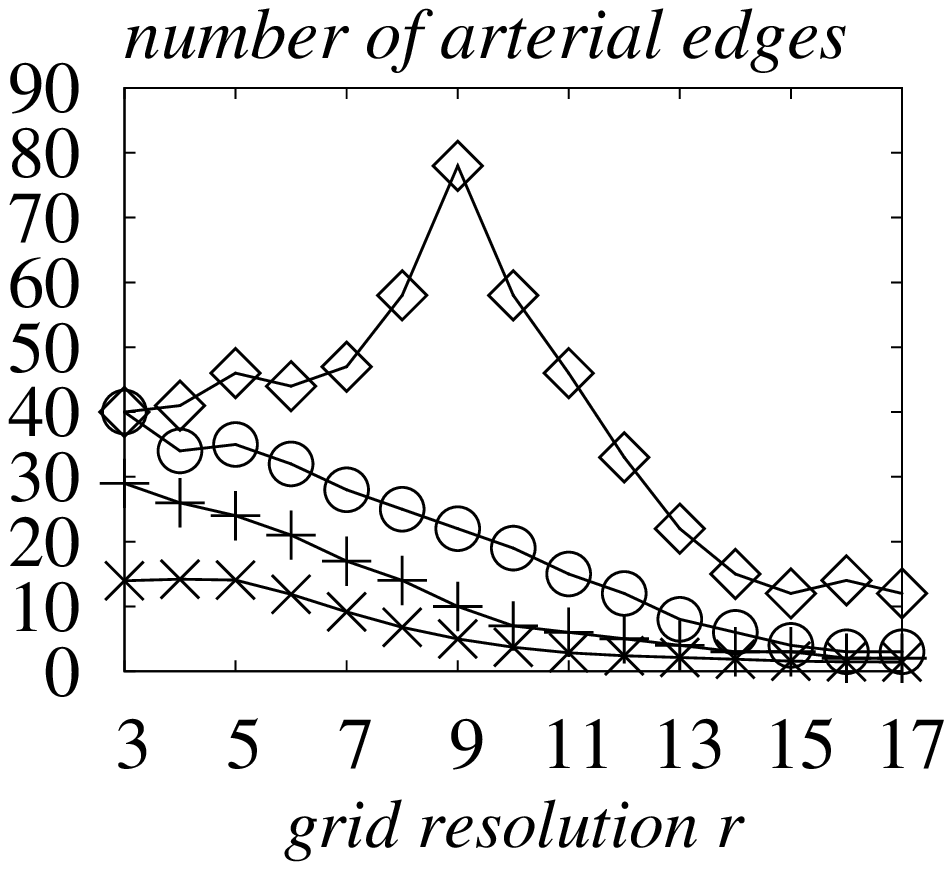} \\

\hspace{-5.5mm}(a) ME ($n = 187{,}315$) & \hspace{-8mm}(b) CO ($n = 435{,}666$) & \hspace{-8mm}(c) FL ($n = 1{,}070{,}376$) & \hspace{-8mm}(d) CA ($n = 1{,}890{,}815$) \\ \vspace{-0.5mm}\\

\hspace{-9.5mm}\includegraphics[height=38mm]{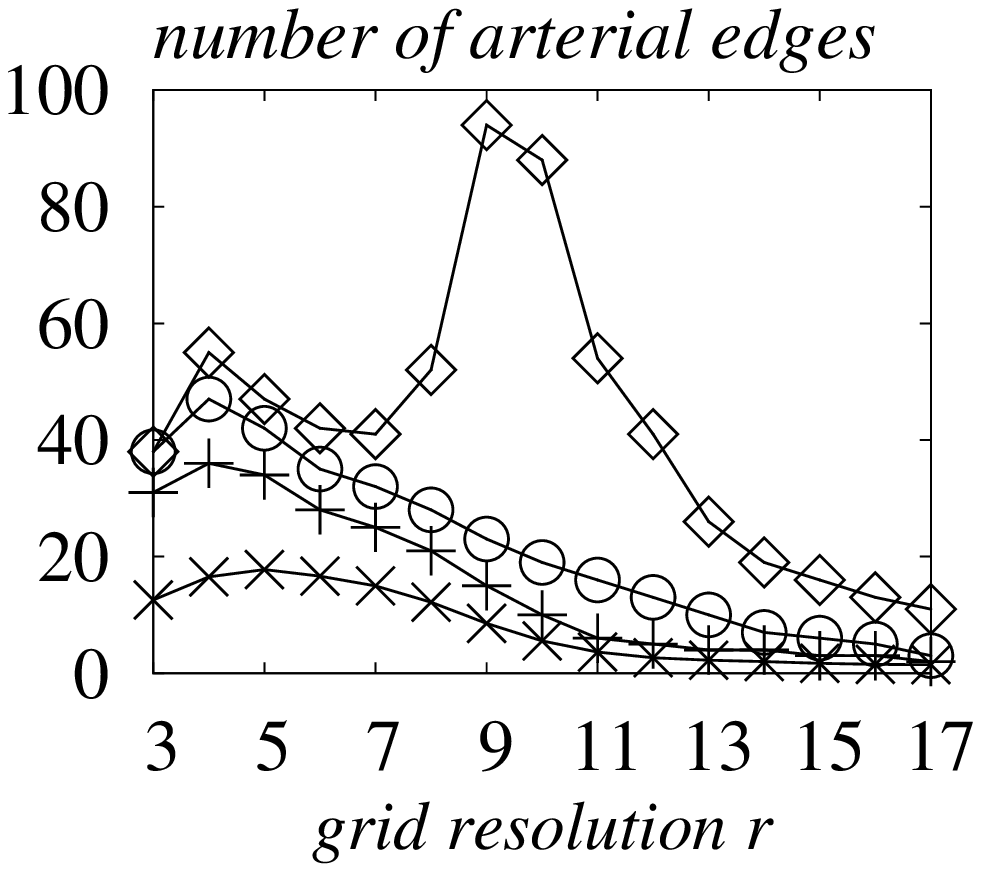}
&
\hspace{-12.5mm}\includegraphics[height=38mm]{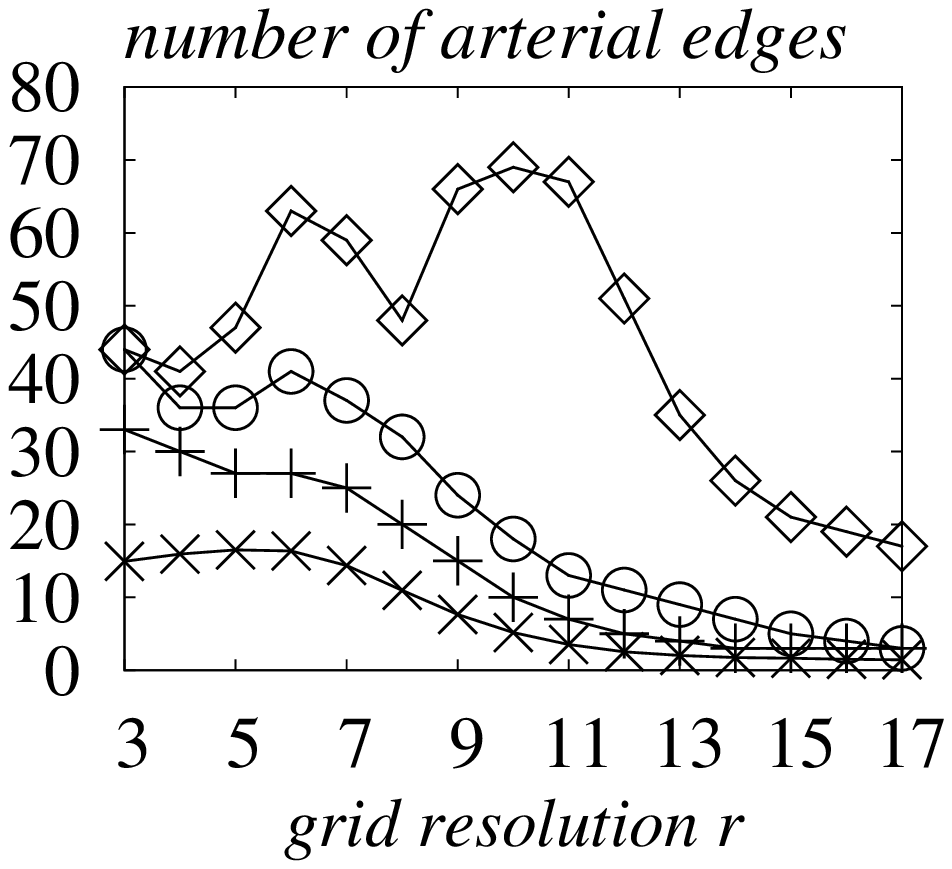}
&
\hspace{-12.5mm}\includegraphics[height=38mm]{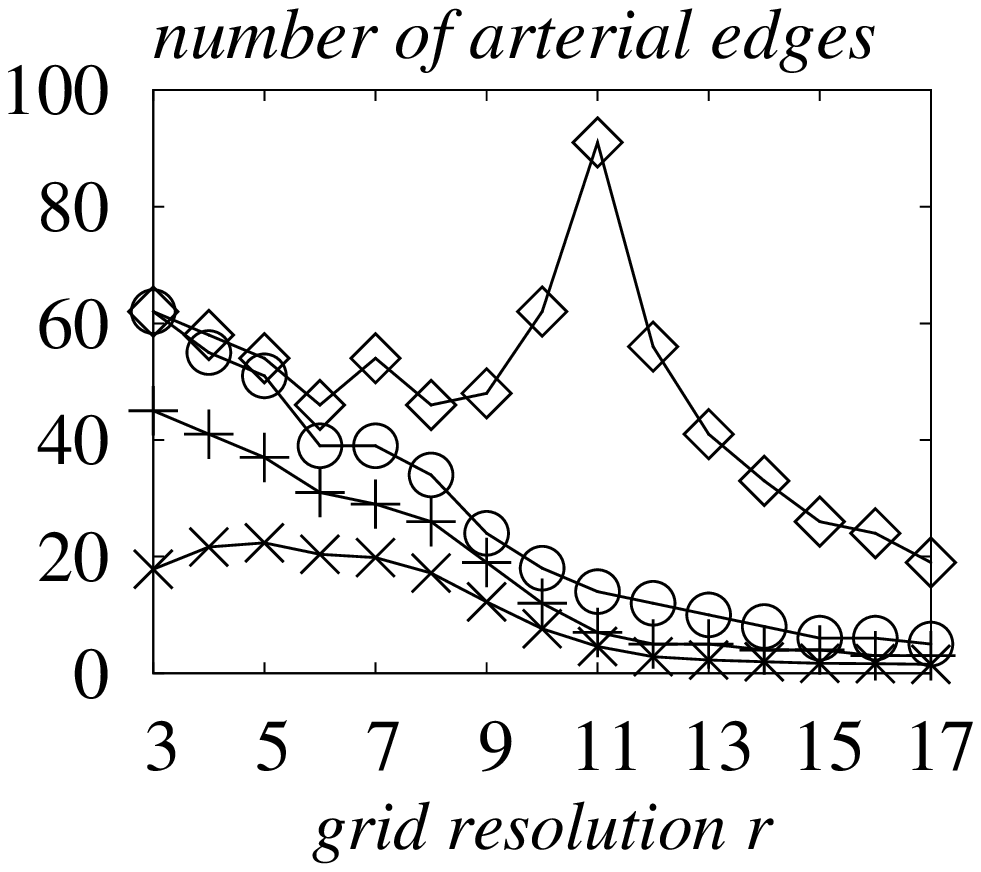}
&
\hspace{-12.5mm}\includegraphics[height=38mm]{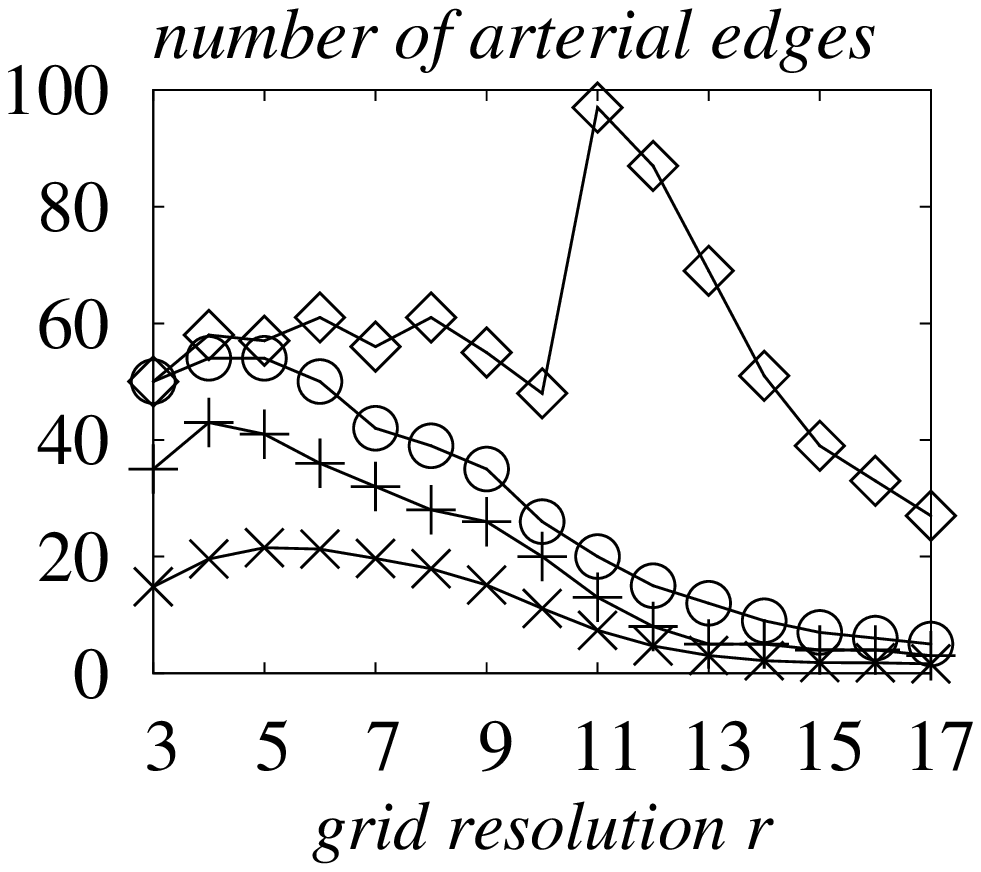} \\

\hspace{-5.5mm}(e) E-US ($n = 3{,}598{,}623$) & \hspace{-8mm}(f) W-US ($n = 6{,}262{,}104$) & \hspace{-8mm}(g) C-US ($n = 14{,}081{,}816$) & \hspace{-8mm}(h) US ($n = 23{,}947{,}347$)
\end{tabular}
\end{small}
\figcapup \caption{Arterial dimensions of real road networks.} \figcapdown
\label{fig:def-ae-num}
\end{figure*}

\section{A First-Cut Solution}\label{sec:fc}

This section presents {\em FC} (\underline{f}irst-\underline{c}ut), an index structure designed for road networks with small arterial dimensions. FC is worst-case efficient for distance queries, and its space consumption is modest; nevertheless, FC is unsuitable for large road networks as it incurs significant pre-processing cost. The reasons that we introduce FC are (i) it is a conceptually simple method that demonstrates the key idea of our proposal, and (ii) with a few modifications and optimizations, FC can be turned into a scalable method that handles both distance and shortest path queries (see Section~\ref{sec:ah}).

\begin{figure}[t]
%\begin{small}
\centering
\includegraphics[width = 3.2in]{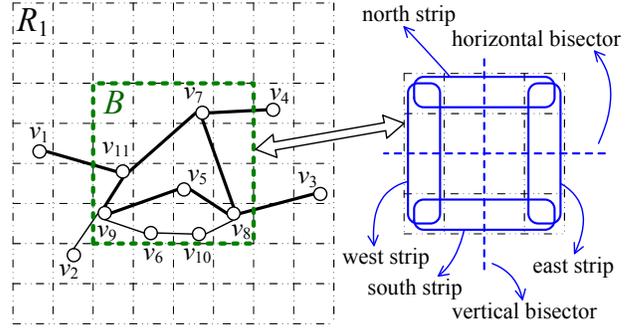}
\figcapup \caption{Strips and bisectors of a region with $4 \times 4$ cells.} \label{fig:def-grid} \figcapdown
\end{figure}

%In what follows, we will first present an overview of FC (Section~\ref{sec:fc-pre}), and then present the query processing algorithms (Section~\ref{sec:fc-query}), followed by a complexity analysis (Section~\ref{sec:fc-complexity}) and a correctness proof (Section~\ref{sec:fc-correctness}).

\subsection{Index Construction} \label{sec:fc-pre}

Given a road network $G$, FC first assigns a {\em level} to each node in $G$, such that nodes with higher levels tend to be more {\em important}. After that, FC organizes the nodes into a hierarchy based on their levels, and it adds auxiliary edges between various nodes to facilitate query processing. In the following, we will elaborate how the node levels are decided and how the auxiliary edges are created.

\header {\bf Deciding Node Levels.} First, FC imposes on $G$ a ($4 {\times} 4$)-cell square grid that tightly covers all nodes in $G$. After that, FC recursively splits each grid cell into $2 \times 2$ smaller cells, until each cell contains at most one node in $G$. This results in a sequence of square grids with increasing resolutions. Let $h$ be the number of grids thus constructed. We use $R_i$ to denote the grid with $2^{h+2 - i} \times 2^{h+2 - i}$ cells, i.e., $R_h$ is the ($4 {\times} 4$)-cell grid that FC first constructed, and $R_1$ is the grid with the highest resolution.

Let $d_{max}$ (resp.\ $d_{min}$) be the largest (resp.\ smallest) $L_\infty$ distance between any two nodes in $G$. It can be verified that $h \le \log_2 (d_{max}/d_{min}) - 1$. We note that $h$ is always a small number for practical road networks: Even if $d_{max}$ is as large as the length of the Equator ($\approx 4 \times 10^7$ meters) and $d_{min}$ is as small as $1$ meter, the value of $h$ is no more than $26$.

Given each $R_i$ ($i \in [1, h]$), FC computes the arterial edges in any ($4 {\times} 4$)-cell region in $R_i$. Let $A_i$ be the set of arterial edges obtained from $R_i$. For any edge in $A_i$, we define it as a {\em level-$i$ edge} if it does not appear in $A_{i+1}, \ldots, A_h$. If an edge does not appear in any $A_i$, then we refer to it as a {\em level-$0$ edge}. In other words, an edge has a higher level if it is an arterial edge for a larger region. Similarly, we also define the level of each node $v$ in $G$: we say that $v$ is a {\em level-$i$ node} if it is adjacent to some edge at level $i$ but not any edge at level $i+1, \ldots, h$. Intuitively, a higher-level node tends to be more important for shortest path and distance queries.

\header {\bf Creation of Shortcuts.} Once the node levels are decided, FC organizes the nodes in $G$ into a hierarchy $H$ of $h+1$ levels $L_0, L_1, \ldots, L_h$, such that all level-$i$ ($i \in [0, h]$) nodes are contained in $L_i$. For example, Figure~\ref{fig:intro-hierarchy} illustrates a $3$-level hierarchy of the nodes in Figure~\ref{fig:intro-original}.

The hierarchy $H$ retains all edges in $G$. In addition, FC inserts into $H$ some auxiliary edges, referred to as {\em shortcuts}. For any two nodes $v_s$ and $v_t$, FC creates a shortcut $c$ from $v_s$ to $v_t$, if the shortest path from $v_s$ to $v_t$ only passes through nodes whose levels are lower than both $v_s$'s and $v_t$'s. Furthermore, the length of $c$ equals to the distance from $v_s$ to $v_t$, i.e., $l(c) = dist(v_s, v_t)$. For instance, consider the nodes $v_6, v_8, v_9, v_{10}$ in Figure~\ref{fig:intro-hierarchy}, whose levels are $0$, $1$, $1$, and $2$, respectively. There is a shortcut from $v_9$ to $v_{10}$, since the shortest path from $v_9$ to $v_{10}$ only goes through $v_6$, and the level of $v_6$ is lower than those of $v_9$ and $v_{10}$. On the other hand, there is no shortcut from $v_8$ to $v_9$, since the shortest path from $v_8$ to $v_9$ passes through $v_{10}$, whose level is higher than both $v_8$'s and $v_9$'s.

The shortcuts inserted into $H$ enable us to avoid visiting unimportant nodes when processing distance queries. For example, given the shortcut $c$ from $v_9$ to $v_{10}$ in Figure~\ref{fig:intro-hierarchy}, we can determine that $dist(v_9, v_{10}) = l(c)$, without having to compute the actual shortest path from $v_9$ to $v_{10}$. In general, for any two nodes $v_s$ and $v_t$ in $H$, there exists a path from $v_s$ to $v_t$ that bypasses unimportant nodes with shortcuts, as will be explained in Section~\ref{sec:fc-query}. For simplicity, we will use the term ``edge'' to refer to either an original edge or a shortcut in $H$, unless otherwise specified.

\subsection{Query Processing} \label{sec:fc-query}

Consider a query $q$ that asks for the distance from a node $s$ in $G$ to another node $t$. Given the node hierarchy $H$, FC answers $q$ with two concurrent traversals of $H$ that start from $s$ and $t$, respectively. Each traversal is performed using a {\em constrained} version of Dijkstra's algorithm \cite{d59}, as explained in the following.

\header {\bf Traversal Algorithm.}  The traversal from $s$ maintains a hash table $T_s$ and a priority queue $Q_s$. The hash table $T_s$ maps each node $v$ in $G$ to a value $\kappa_s(v)$, which equals the length of the shortest path from $s$ to $t$ that has been found so far. Initially, we have $\kappa_s(s) = 0$ and $\kappa(v) = +\infty$ for any other node $v$.

Meanwhile, each entry in the priority queue $Q_s$ corresponds to a certain node $v'$ in $G$, and the key of the entry equals $\kappa_s(v')$. In the beginning of the traversal, $Q_s$ contains only one entry, which corresponds to $s$. Subsequently, FC iteratively extracts (from $Q_s$) the node $u$ with the smallest key. For each $u$ extracted, FC inspects every edge $\langle u, v\rangle$ in $H$ that starts from $u$, and it checks whether $v$ satisfies certain constraints. (We will clarify these constraints shortly). If $v$ violates any of the constraints, it would be ignored; otherwise, FC would further check whether $\kappa_s(u) + l(\langle u, v\rangle) < \kappa_s(v)$, i.e., whether the path from $s$ to $v$ via $u$ is shorter than all known paths from $s$ to $v$. If the inequality holds, then FC sets $\kappa_s(v) = \kappa_s(u) + l(\langle u, v\rangle)$ and inserts $v$ into $Q_s$ (if $v$ has not been inserted before).

The traversal from $t$ also maintains a hash table $T_t$ and a priority queue $Q_t$. It is performed in a manner similar to the traversal from $s$, with one notable difference: Whenever FC extracts a node $u$ from $T_t$, it only inspects the edges $\langle v, u\rangle$ that {\em points to} $u$. In other words, the traversal from $t$ focuses on paths that end at $t$.

FC conducts the above two traversals in a round-robin fashion, i.e., it extracts nodes from the two priority queues $Q_s$ and $Q_t$ in turns. To determine when the traversals can be terminated, FC maintains a variable $\theta$ that records the length of the shortest path from $s$ to $t$ that is seen so far. Initially, $\theta = +\infty$. After that, for each node $u$ extracted from either priority queue, FC retrieves its key $\kappa_s(u)$ in the hash table $T_s$, as well as its key $\kappa_t(u)$ in $T_t$. Recall that $\kappa_s(u)$ (resp.\ $\kappa_t(u)$) records the length of the shortest path from $s$ to $u$ (resp.\ from $u$ to $t$) found so far. Therefore, the shortest path from $s$ to $t$ should be no longer than $\kappa_s(u) + \kappa_t(u)$. Based on this, if $\kappa_s(u) + \kappa_t(u) < \theta$, then FC would update $\theta$ and set it to $\kappa_s(u) + \kappa_t(u)$.

Whenever $\theta$ is no more than the smallest key value in $Q_s$, we know that $dist(s, u) \ge \theta$ for any node $u$ remaining in $Q_s$, which indicates that $u$ cannot be on the shortest path from $s$ to $t$. In that case, FC would terminate the traversal from $s$. Similarly, the traversal from $t$ is stopped if $\theta$ is no more than any key values in $Q_t$. When both traversals are terminated, FC returns $\theta$ as the answer to the distance query.

\header {\bf Constraints on Node Traversals.} As mentioned above, whenever FC extracts a node $u$ from a priority queue (either $Q_s$ or $Q_t$), it inspects the neighbors of $u$, and it processes only those neighbors $v$ that satisfy certain constraints. Specifically, there are two constraints on $v$:
\begin{enumerate} %[itemsep = 1mm]
\item {\em Level Constraint:} $v$ should not be at a level lower than $u$'s.

\item {\em Proximity Constraint:} Let $i$ be the level of $v$ ($i \in [0, h-1]$). If $u$ is extracted from $Q_s$ (resp.\ $Q_t$), then $v$ and $s$ (resp.\ $t$) should be covered in the same ($3 {\times} 3$)-cell region in $R_{i+1}$. (Recall that $R_i$ is a square grid with $2^{h+2 - i} \times 2^{h+2 - i}$ cells.)
\end{enumerate}

Both of the above constraints are intended to improve the efficiency of FC. In particular, the level constraint helps FC bypass unimportant nodes during query processing. For example, consider that we use FC to compute the distance from $v_8$ to $v_{11}$ in Figure~\ref{fig:intro-hierarchy}. As explained previously, FC would invoke two traversals starting from $v_8$ and $v_{11}$, respectively. Since $v_{11}$ is at level $2$, the traversal from $v_{11}$ would only visit the neighbors of $v_{11}$ that are at levels no lower than $2$. As a consequence, $v_{10}$ is visited (since its level equals $2$), while $v_1$ and $v_9$ are bypassed. After that, the traversal would visit any neighbor of $v_{10}$ whose level is no lower than that of $v_{10}$. Since none of the neighbors of $v_{10}$ fulfills this requirement, the traversal terminates. Similarly, the traversal from $v_8$ would first visit two of $v_8$'s neighbors, $v_7$ and $v_{10}$, ignoring the remaining neighbor $v_3$, since $v_3$'s level is lower than that of $v_8$. After that, the traversal visits only $v_{11}$ and terminates, as all other remaining nodes violate the level constraint. In summary, the two traversals by FC visit only four nodes: $v_7$, $v_8$, $v_{10}$, and $v_{11}$. %; all other nodes in $H$ are omitted.

\begin{figure}[t]
%\begin{small}
\centering
\includegraphics[width = 2.8in]{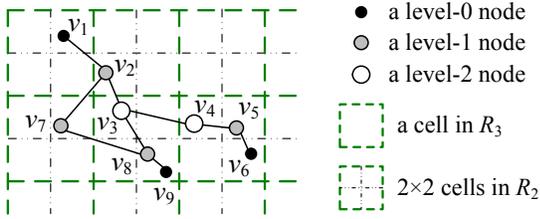}
\figcapup \caption{Illustration of the proximity constraint.} \label{fig:fc-proximity} \figcapdown
\end{figure}

%Meanwhile, the proximity constraint ensures that FC only searches a small number of grid cells in each level of the node hierarchy $H$. For example, suppose that we are given the node hierarchy in Figure~\ref{fig:fc-proximity}, and we use FC to compute the distance from $v_1$ to $v_6$. Among the two traversals invoked by FC, the one starting from $v_1$ would first visit $v_2$, and then $v_3$ and $v_7$. The node $v_7$ has two unvisited neighbors, $v_8$ and $v_{10}$, both of which are at the same level as $v_7$, i.e., they satisfy the level constraint. However, they would be ignored by FC, as they violate the proximity constraint. In particular, both $v_8$ and $v_{10}$ are level-$1$ nodes, but there does not exist any $(2 \times 2)$-cell region in $R_2$ that can cover $v_1$ and $v_8$ (resp.\ $v_{10}$) simultaneously. Note that ignoring $v_8$ and $v_{10}$ does not affect the correctness of the query result, since neither $v_8$ nor $v_{10}$ connects $v_1$ to $v_6$. In contrast, the node $v_4$, which is a neighbor of $v_3$, would be visited by FC as it satisfies both the level and proximity constraints. Specifically, the level of $v_4$ equals $2$, which is no less than that of $v_3$; furthermore, $v_4$ and $v_1$ is contained in the same ($4 {\times} 4$)-cell region in $R_3$.

Meanwhile, the proximity constraint ensures that FC only searches a small number of grid cells in each level of the node hierarchy $H$. For example, suppose that we are given the node hierarchy in Figure~\ref{fig:fc-proximity}, and we use FC to compute the distance from $v_1$ to $v_6$. Among the two traversals invoked by FC, the one starting from $v_1$ would first visit $v_2$, and then $v_3$ and $v_7$. The node $v_7$ has a neighbor $v_8$, which is at the same level as $v_7$, i.e., $v_8$ satisfies the level constraint. However, $v_8$ would still be ignored by FC, as it violates the proximity constraint. In particular, $v_8$ is a level-$1$ node, but there does not exist any ($3 {\times} 3$)-cell region in $R_2$ that can cover both $v_1$ and $v_8$. In contrast, the node $v_4$, which is a neighbor of $v_3$, would be visited by FC as it satisfies both the level and proximity constraints. Specifically, the level of $v_4$ equals $2$, which is no less than that of $v_3$; furthermore, $v_4$ and $v_1$ are contained in the same ($3 {\times} 3$)-cell region in $R_3$. Note that, although FC ignores $v_8$, the correctness of the query result is not affected, since $v_8$ is not on the shortest path from $v_1$ to $v_6$.

In general, the proximity constraint guarantees that in each level $i$ of the node hierarchy, FC only traverses the nodes contained in two ($5 {\times} 5$)-cell regions, which are centered at the source $s$ and destination $t$ of the query, respectively. In particular, the region centered at $s$ (resp.\ $t$) is the union of all ($3 {\times} 3$)-cell regions that cover $s$ (resp.\ $t$). This, when combined with the level constraint, ensures that FC is worst-case efficient in terms of query time, as will be shown in Section~\ref{sec:fc-complexity}.

%the proximity constraint prevents FC from visiting a node $v$ too far away from $s$ and $t$, unless $v$ is of high importance. Intuitively, this is consistent with the way that we drive between two distant locations $s$ and $t$ in practice: In the beginning, we may start from some small streets connected to $s$; but as we move further away from $s$, we turn into the main roads and no longer traverse the small streets, i.e., we focus on more important road network nodes as we travel away from $s$. After a while, we may further move on to the highways; but upon approaching $t$, we will exit the highways to get into some main roads, which then lead us to some small streets connected with $t$. In other words, we would traverse less important nodes only when they are close to $t$.

\subsection{Complexity Analysis} \label{sec:fc-complexity}

In this section, we will prove that FC takes $O(h n)$ space, and it answers any distance query in $O(h^2)$ time, where $h$ is the maximum level in the node hierarchy $H$, and $n$ is the number of nodes in the road network $G$. In addition, we will discuss the pre-computation time of FC.

\header {\bf Query Time.} As explained in Section~\ref{sec:fc-query}, FC answers any distance query by two traversals on the node hierarchy $H$, starting from the source $s$ and destination $t$ of the query, respectively. Due to the level and proximity constraints, each traversal of FC visits any level of $H$ at most once; in addition, for the $i$-th level ($i \in [0, h]$), each traversal only examines the nodes in a ($5 {\times} 5$)-cell region in the grid $R_{i+1}$. A natural question is: How many level-$i$ nodes are there in the ($5 {\times} 5$)-cell region? The following lemma provides an answer.
%\begin{lemma} \label{lmm:fc-node-num}
%Any ($\alpha {\times} \alpha$)-cell region in $R_i$ contains $O(\alpha^2 \a)$ level-$i$ nodes in $H$, where $\a$ is the arterial dimension of $G$.
%\end{lemma}
\begin{lemma} \label{lmm:fc-node-num}
Any ($\alpha {\times} \alpha$)-cell region in $R_i$ contains $O(\alpha^2 \a)$ level-$i$ nodes in $H$, where $\a$ is the arterial dimension of $G$.
\end{lemma}

To explain the rationale behind Lemma~\ref{lmm:fc-node-num}, recall that each level-$i$ node in $H$ is adjacent to an arterial edge in a ($4 {\times} 4$)-cell region in $R_i$. Furthermore, each ($4 {\times} 4$)-cell region in $R_i$ has at most $\a$ arterial edges. For any ($\alpha \times \alpha$)-cell region in $R_i$, it can overlap with $O(\alpha^2)$ regions of $4 \times 4$ cells, and hence, it contains $O(\alpha^2 \a)$ level-$i$ nodes.

Observe that any ($5 {\times} 5$)-cell region in $R_{i+1}$ corresponds to a ($10 \times 10$)-cell region in $R_i$. By Lemma~\ref{lmm:fc-node-num}, this region contains $O(\a)$ level-$i$ nodes in $H$. In other words, the number of level-$i$ nodes visited by each traversal of FC is $O(\a)$. Given that $H$ has $h+1$ levels, the total number of nodes traversed by FC is $O(h\a)$.

Next, we will show that each node in $H$ has $O(h\a)$ edges that satisfy the level constraint. (The edges that violate the constraint can be removed from $H$ beforehand, as they would never be traversed by FC for any query). Consider any level-$i$ node $u$, and any node $v$ whose level is at least $i$. By the way that $H$ is constructed, there is a shortcut connecting $u$ to $v$, if and only if the shortest path between $u$ and $v$ only goes through nodes at levels lower than $i$. Intuitively, this indicates that $u$ and $v$ should not be too far apart from each other; otherwise, the shortest path between $u$ and $v$ in $G$ would be a path that connects two distant locations, in which case the path might contain some highly important node at a level higher than $i$, due to which there would not be any shortcut between $u$ and $v$. More formally, we have the following lemma:

\begin{lemma} \label{lmm:fc-sliding}
Let $P$ be a shortest path in $G$, such that no ($3 {\times} 3$)-cell region in $R_i$ ($i \in [1, h]$) can cover all nodes in $P$ simultaneously. Then, $P$ must contain an arterial edge of some ($4 {\times} 4$)-cell region in $R_i$.
\end{lemma}

By Lemma~\ref{lmm:fc-sliding}, $u$ and $v$ must be covered in the same ($3 {\times} 3$)-cell region in $R_{i+1}$; otherwise, the shortest path between $u$ and $v$ must pass through a level-($i{+}1$) node, for which there cannot exist any shortcut between $u$ and $v$. This implies that $v$ must be in the ($5 {\times} 5$)-cell region in $R_{i+1}$ that is centered at $u$. By Lemma~\ref{lmm:fc-node-num}, this region contains $O(\a)$ level-$i$ nodes in $H$. With a similar analysis, it can be shown that the region also covers $O(\a)$ nodes at any level higher than $i$. Therefore, the total number of edges adjacent to $u$ is $O(h \a)$.

In summary, FC answers any distance query with two constrained Dijkstra search, each of which traverses $O(h\a)$ nodes and $O(h^2 \a^2)$ edges. As such, the time complexity of each traversal equals $O(h\a \log(h\a) + h^2 \a^2)$. Given that the arterial dimension $\a$ of the road network $G$ is a constant, the overall time complexity of FC is $O(h^2)$.

\header {\bf Space Complexity.} Recall that the node hierarchy contains $h+1$ levels, each of which contains $O(n)$ nodes. In addition, each node in $H$ has $O(h\a)$ edges. Therefore, the space consumption of FC is $O(h n)$ when $\a$ is constant.

\header {\bf Preprocessing Cost.} The pre-computation of FC consists of two steps: First, we identify the arterial edges in any ($4 {\times} 4$)-cell region in any grid $R_i$ ($i \in [1, h]$); After that, we decide the level of each node and we connect pairs of nodes with shortcuts. The identification of arterial edges requires computing the shortest paths in all ($4 {\times} 4$)-cell regions in all $R_i$, which incurs considerable overhead, especially when the granularity of the grid is low. Similarly, the construction of shortcuts is time consuming as it requires deriving a larger number of shortest paths (between nodes that are potentially far apart). Such significant preprocessing cost renders FC only applicable for small road networks. In Section~\ref{sec:ah}, we will address this issue and present a modified and scalable version of FC.

\subsection{Correctness Proof} \label{sec:fc-correctness}

Let $P = \langle v_1, v_2, \ldots, v_k\rangle$ be a shortest path in $G$. Let $P'$ be a path from $v_1$ to $v_k$ on the node hierarchy $H$, such that FC reports $l(P')$ as the distance from $v_1$ to $v_k$. We will prove the correctness of FC's result by showing that $l(P') = l(P)$. In particular, we will show that both $l(P') \ge l(P)$ and $l(P') \le l(P)$ hold.

\header {\bf Proving $\boldsymbol{l(P') \ge l(P)}$}. Recall that every shortcut on $H$ corresponds to a path in $G$. Therefore, if we replace each shortcut in $P'$ with the corresponding path, we can transform $P'$ into a path $P''$, such that (i) $P''$ does not contain any shortcut, (ii) $P''$ connects $v_1$ to $v_k$, and (iii) $l(P') = l(P'')$. On the other hand, we have $l(P'') \ge l(P)$, since $P$ is the shortest path from $v_1$ to $v_k$ in $G$. Therefore, $l(P') = l(P'') \ge l(P)$.

\header {\bf Proving $\boldsymbol{l(P') \le l(P)}$}. Assume for simplicity that $P$ contains a node $v_j$ ($j \in [1, k]$) whose level is higher than that of any other node on $P$. (Our analysis can be easily extended to the case when the highest-level node on $P$ is not unique.) Let $P_1 = \langle v_1, v_2, \ldots, v_j\rangle$ and $P_2 = \langle v_j, v_{j+1}, \ldots, v_k\rangle$. In the following, we will show that the node hierarchy $H$ contains a path $P'_1$ from $v_1$ to $v_j$ that has the same length with $P_1$. Furthermore, we will prove that the sequence of nodes on $P'_1$ satisfies both the level and proximity constraints, i.e., $P'_1$ can be identified by FC with a traversal starting from $v_1$. In a similar manner, it can be shown that $H$ contains a path $P'_2$ from $v_j$ to $v_k$, such that $l(P'_2) = l(P_2)$, and that $P'_2$ can be found by FC with a constrained Dijkstra search starting from $v_k$. This would lead to
\begin{displaymath}
l(P') \,\; \le \,\; l(P'_1) + l(P'_2) \,\; = \,\; l(P_1) + l(P_2) \,\; = \,\; l(P).
\end{displaymath}

Consider the path $P_1 = \langle v_1, v_2, \ldots, v_j\rangle$. Suppose that we remove from $P_1$ any node $v_i$ ($i \in [1, j]$) that has a smaller level than some node $v_a$ ($a < i$) preceding it. Let $S = \langle v'_1, v'_2, \ldots v'_b \rangle$ be the sequence of nodes remaining on $P_1$. We have $v'_1 = v_1$ (since no node precedes $v_1$), and $v'_b = v_j$ (since $v_j$ is the highest-level node in $P$). For instance, if $P_1$ contains six nodes $v_1$, $v_2$, $v_3$, $v_4$, $v_5$, $v_6$ at levels $1$, $0$, $2$, $2$, $1$, $3$, respectively, then $S = \langle v_1, v_3, v_4, v_6 \rangle$.

By the way that $S$ is constructed, for any $v'_i$ ($i \in [1, b{-}1]$), the shortest path from $v'_i$ to $v'_{i+1}$ contains only nodes whose levels are smaller than those of $v'_i$ and $v'_{i+1}$. As such, the node hierarchy $H$ would contain a shortcut from $v'_i$ to $v'_{i+1}$, and the length of the shortcut equals $dist(v'_i, v'_{i+1})$. This indicates that $H$ contains a path $P'_1 = \langle v'_1, v'_2, \ldots v'_b \rangle$ that connects $v_1$ to $v_k$, such that $l(P'_1) = l(P)$. Furthermore, $P'_1$ satisfies the level constraint, since the level of $v'_i$ ($i \in [1, b{-}1]$) is no larger than that of $v'_{i+1}$.

%Assume to the contrary that $P'_1$ does not satisfy the proximity constraint. Then, there should exist a node $v'_a$ ($a \in [1, b-1]$) on $P'_1$, such that (i) $v'_a$ is at level $i$ ($i \in [0, h]$), but (ii) no ($3\times3$)-cell region in $R_{i+1}$ covers $v'_1$ and $v'_{a+1}$ simultaneously. Then, $v'_1$ and $v'_a$ must locate in two distant cells in $R_{i+1}$. Intuitively, this indicates that the shortest path from $v'_1$ to $v'_a$ must run through a certain ($4 {\times} 4$)-cell region $B$ in $R_{i+1}$, in which case the path should contain an arterial edge of $B$. More formally, we have the following lemma:

%\begin{lemma} \label{lmm:fc-sliding}
%Let $P$ be a shortest path in $G$, such that no ($3 {\times} 3$)-cell region in $R_i$ ($i \in [1, h]$) can cover all nodes in $P$ simultaneously. Then, $P$ must contain an arterial edge of some ($4 {\times} 4$)-cell region in $R_i$.
%\end{lemma}
%\begin{proof}
%The proofs for all lemmas and theorems in this paper can be found in our online and anonymous technical report \cite{???}.
%\end{proof}

Assume to the contrary that $P'_1$ does not satisfy the proximity constraint. Then, there should exist a node $v'_a$ ($a \in [1, b{-}1]$) on $P'_1$, such that (i) $v'_a$ is at level $i$ ($i \in [0, h]$), but (ii) no ($3\times3$)-cell region in $R_{i+1}$ covers both $v'_1$ and $v'_{a+1}$. Then, by Lemma~\ref{lmm:fc-sliding}, the shortest path from $v'_1$ to $v'_a$ must contain an arterial edge $e$ of a ($4 {\times} 4$)-cell region in $R_{i+1}$, since none of the ($3 {\times} 3$)-cell regions in $R_{i+1}$ covers both $v'_1$ and $v'_a$. In that case, each endpoint of $e$ has a level at least $i+1$. In other words, on the shortest path from $v'_1$ to $v'_a$, there exists some node whose level is higher than that of $v'_a$ (recall that $v'_a$ is at level $i$). This contradicts the assumption that $v'_a$ has a level no lower than any node preceding it on $P_1$.

In summary, the node hierarchy $H$ contains a path $P'_1$ from $v_1$ to $v_j$, such that $P'_1$ has the same length with $P_1$ and satisfies both the level and proximity constraints. Therefore, FC can correctly identify the distance from $v_1$ to $v_j$ with a traversal starting from $v_1$. Similarly, we can show that FC can correctly compute the distance from $v_j$ to $v_k$ with a traversal starting from $v_k$. This proves the correctness of the query processing algorithm of FC.

\vspace{-1mm}
\section{Arterial Hierarchy} \label{sec:ah}

This section presents {\em Arterial Hierarchy (AH)}, a scalable indexing method built upon the FC approach introduced in Section~\ref{sec:fc}. Compared with FC, AH has the same space complexity, a similar time complexity for distance queries, but significantly smaller pre-computation cost. In addition, AH also supports shortest path queries in a worst-case efficient manner.

%In the following, we will XXXX

\subsection{Overview} \label{sec:ah-overview}

The main structure of AH is a node hierarchy $\H$ that resembles FC's node hierarchy $H$. In particular, both $\H$ and $H$ have $h+1$ levels, and both of their $i$-th levels ($i \in [1, h]$) are associated with a square grid $R_i$ of $2^{h+2 - i} \times 2^{h+2 - i}$ cells. However, AH and FC differ substantially in the ways that they decide node levels, construct shortcuts, and process queries.

\header
{\bf Differences in Node Levels.} To compute the level of each node, FC first imposes each $R_i$ on the road network $G$, and then computes the arterial edges in each ($4 {\times} 4$)-cell region in $R_i$, after which FC decides the node levels based on the arterial edges. As discussed in Section~\ref{sec:fc-complexity}, the derivation of arterial edges could incur significant overheads, since each ($4 {\times} 4$)-cell region in a coarse grid may cover a large number of nodes and edges in $G$.

In contrast, AH computes node levels with an incremental algorithm that substantially improves efficiency. Given $G$, it first imposes the grid $R_1$ on $G$. Based on $R_1$, it identifies a set of unimportant nodes in $G$, and it assigns them to level $0$ of the node hierarchy $\H$. Then, it removes a subset of the unimportant nodes from $G$, and constructs shortcuts between the remaining nodes. This results in a {\em reduced} graph $G_1$ that is considerably smaller than $G$. After that, AH recursively reduces $G_1$ into smaller graphs $G_2, G_3, \ldots G_h$, during which it assigns nodes to higher levels of $\H$. For the reduction from $G_{i}$ to $G_{i+1}$, AH needs to impose the grid $R_{i+1}$ on $G_{i}$ and compute the shortest paths in each ($4 {\times} 4$)-cell region. However, this computation is inexpensive since (i) $G_i$ has a much smaller size than $G$, and hence, (ii) each ($4 {\times} 4$)-cell region in $R_{i+1}$ contains only a small number of nodes and edges in $G_i$.

\header
{\bf Differences in Shortcuts.} FC creates only necessary shortcuts to ensure correct results for distance queries under the level and proximity constraints. In contrast, the shortcuts constructed by AH are not only for processing distance queries under the level and proximity constraints, but also for computing the actual shortest path between any two given nodes. Specifically, every shortcut $\langle v_a, v_c\rangle$ in AH's node hierarchy $\H$ is associated with a node $v_b$, such that (i) both $\langle v_a, v_b \rangle$ and $\langle v_b, v_c \rangle$ are edges in $\H$, and (ii) the length of $\langle v_a, v_c\rangle$ equals the lengths of $\langle v_a, v_b \rangle$ and $\langle v_b, v_c \rangle$ combined. In other words, $\langle v_a, v_c\rangle$ can be transformed into a two-hop shortest path $\langle v_a, v_b, v_c \rangle$. As such, given any path $P'$ in $\H$, we can transform $P'$ into a path in $G$, by recursively replacing each shortcut in $P'$ with its corresponding two-hop path.

For example, Figure~\ref{fig:ah-recon} illustrates a shortest path $\langle v_1, v_2, \ldots, v_6\rangle$ in $G$, as well as three shortcuts $\langle v_1, v_4\rangle$, $\langle v_2, v_4 \rangle$, and $\langle v_4, v_6\rangle$. The shortcut $\langle v_1, v_4\rangle$ is associated with the node $v_2$, since $v_1$ is directly connected with $v_2$ and $v_2$ is directly connected with $v_4$. Similarly, $\langle v_2, v_4 \rangle$ and $\langle v_4, v_6\rangle$ are associated with $v_3$ and $v_5$, respectively. Now suppose that, given a distance query from $v_1$ to $v_6$, AH identifies $P' = \langle v_1, v_4, v_6\rangle$ as the shortest path from $v_1$ to $v_6$ in $\H$. To derive the actual shortest path from $v_1$ to $v_6$ in $G$, AH first replaces the shortcut $\langle v_1, v_4 \rangle$ in $P'$ with a two-hop path $\langle v_1, v_2, v_4 \rangle$, since $\langle v_1, v_4 \rangle$ is associated with $v_2$. This transforms $P'$ into another path $\langle v_1, v_2, v_4, v_6\rangle$. After that, we can replace $\langle v_2, v_4 \rangle$ with $\langle v_2, v_3, v_4 \rangle$, and substitute $\langle v_4, v_6\rangle$ with $\langle v_4, v_5, v_6\rangle$. As such, we obtain the shortest path $\langle v_1, v_2, \ldots, v_6\rangle$ from $v_1$ to $v_6$ in $G$.

In general, given any shortest path query from a node $s$ to another node $t$, AH first computes the shortest path $P'$ from $s$ to $t$ in $\H$, and then it converts $P'$ into the corresponding path $P$ in the original road network. The conversion from $P'$ to $P$ takes only $O(k)$ time, where $k$ is the number of edges in $P$. This is because (i) for any shortcut in $\H$, we can identify its corresponding two-hop path in $O(1)$ time, and (ii) converting $P'$ to $P$ requires only $O(k)$ replacements of shortcuts.

\header
{\bf Differences in Query Processing.} Besides the aforementioned shortcuts (for reconstructing shortest paths), the node hierarchy $\H$ of AH also contains some extra shortcuts that can be leveraged for higher query efficiency. As a consequence, AH's query processing algorithm is slightly more sophisticated than FC's, as will be elaborated in Section~\ref{sec:ah-query}.

\subsection{Index Construction} \label{sec:ah-pre}

Similar to the case of FC, AH constructs its node hierarchy $\H$ in two steps: it first assigns each node in $G$ to a level in $\H$, and then it constructs shortcuts in $\H$ for query processing.

%for all nodes in the same level, it orders the nodes based on a measurement of their relative importance. Finally, it creates shortcuts between nodes based on both the node levels and orderings.
%
\header
{\bf Deciding Node Levels.}
Given the road network $G$, AH first imposes on $G$ the grid $R_1$, where each cell contains at most one node. After that, AH identifies all ($4 {\times} 4$)-cell regions in $R_1$ that cover at least one node in $G$. For each of the ($4 {\times} 4$)-cell region identified, AH computes the arterial edges of the region in $O(1)$ time, and it marks each endpoint of an arterial edge as a {\em level-$1$ core}. At the same time, AH assigns all unmarked nodes to level $0$ of the node hierarchy $\H$ since, intuitively, those nodes are less important than the level-$1$ cores. After that, if any ($4 {\times} 4$)-cell region $B$ contains a local shortest path $P$ from a level-$1$ core $u$ to another level-$1$ core $v$, such that $P$ only goes through unmarked nodes, then AH inserts into $G$ a shortcut $\langle u, v\rangle$ with the same length as $P$. We say that $\langle u, v\rangle$ is a shortcut generated from $B$, and we use $G_1$ to denote the modified version of $G$ with all shortcuts added.
%We also define the ID of $\langle u, v\rangle$ as the smallest ID of the edges on $P$.
Overall, the computation of level-$1$ cores and the construction of $G_1$ take only $O(n)$ time, since the number of non-empty ($4 {\times} 4$)-cell regions in $R_1$ is $O(n)$, and each of those regions contains $O(1)$ nodes and edges in $G$ (recall that $G$ is degree-bounded).

For example, given the road network $G$ and the grid $R_1$ in Figure~\ref{fig:def-grid}, assume that AH identifies $5$ level-$1$ cores: $v_7, v_8, v_9, v_{10}, v_{11}$. After adding shortcuts, $G$ is transformed into the graph $G_1$ in Figure~\ref{fig:ah-G0}. There exists a shortcut $\langle v_9$, $v_{10}\rangle$ in $G_0$ since (i) both $v_9$ and $v_{10}$ are level-$1$ cores, and (ii) in the ($4 {\times} 4$)-cell region $B$ illustrated in Figure~\ref{fig:def-grid}, the local shortest path between $v_9$ and $v_{10}$ goes only through $v_6$, which is unmarked. In general, the shortcuts in $G_1$ ensure that the level-$1$ cores form a connected graph even if we remove all unmarked nodes from $G_0$.

%\header
%{\bf Deciding Node Levels.}
Given $G_1$, AH selects a subset of the level-$1$ cores in $G_1$ that are deemed more important than the others. The selected nodes are marked as the {\em level-$2$ cores}, while the remaining level-$1$ cores are assigned to level $1$ of $\H$. After that, AH converts $G_1$ into a smaller graph $G_2$ that retains all level-$2$ cores. This procedure is applied in a recursive manner: In the $i$-th recursion ($i \in [1, h-1]$), AH picks level-($i{+}1$) cores from the level-$i$ cores in $G_{i}$, and then assigns the un-picked ones to level $i$ of $\H$, after which it transforms $G_{i}$ into a smaller graph $G_{i+1}$.

A natural question is: Given $G_{i}$, how should AH select the level-($i{+}1$) cores from the level-$i$ cores? One straightforward solution is to construct a subgraph of $G_{i}$ that contains only the level-$i$ cores, and then compute the arterial edges in the subgraph to identify the more important nodes as level-($i{+}1$) cores. For example, given $G_1$ in Figure~\ref{fig:ah-G0}, we can first construct a subgraph of $G_1$ that contains only the five level-$1$ cores (i.e., $v_7, v_8, v_9, v_{10}, v_{11}$) and the edges connecting them (i.e., the five edges on the loop $\langle v_7, v_8, v_{10}, v_9, v_{11}, v_7\rangle$). After that, we impose the grid $R_2$ on the subgraph, compute the arterial edges, and then mark the endpoints of the arterial edges as level-$2$ cores. While this approach is intuitive, we find that (i) the resulting node hierarchy does not guarantee query correctness under the level and proximity constraints, and (ii) without the level and proximity constraints, it is difficult to achieve favorable asymptotic bounds on query time. To address this issue, we adopt a more careful approach to choose the level-($i{+}1$) cores without affecting the applicability of the level and proximity constraints in query processing. Specifically, our approach utilizes the concept of {\em border nodes}:
\begin{definition}[Border Nodes] \label{def:ah-border}
Let $B$ be a ($4 {\times} 4$)-cell region in $R_i$ ($i \in [1, h]$). A node $v$ in $G$ is a {\bf border node} of $B$, if (i) $v$ is not contained in the $2 \times 2$ cells centered at $B$, and (ii) $v$ is an endpoint of an edge in $G$ that intersects the boundary of the east, west, south, or north strip of $B$.
\end{definition}

\begin{figure}[t]
%\begin{small}
\begin{minipage}[t]{1.68 in}
\centering
%\hspace{3mm}
\includegraphics[width = 1.64in]{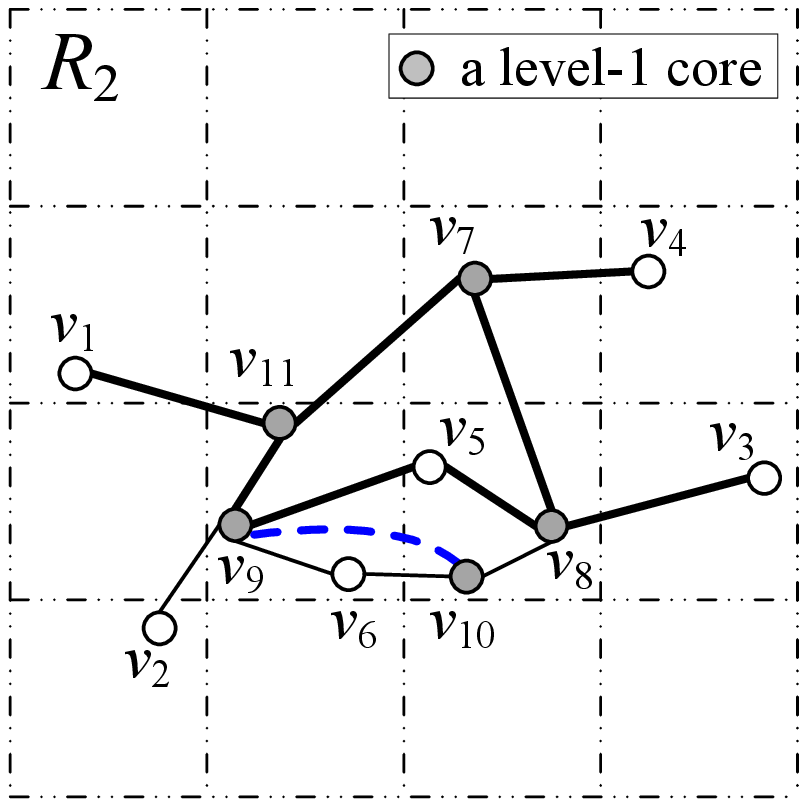}
\figcapup \caption{Graph $G_1$.} \figcapdown \label{fig:ah-G0}
\end{minipage}
%\hspace{0.02in}
\begin{minipage}[t]{1.68 in}
\centering
\includegraphics[width = 1.64in]{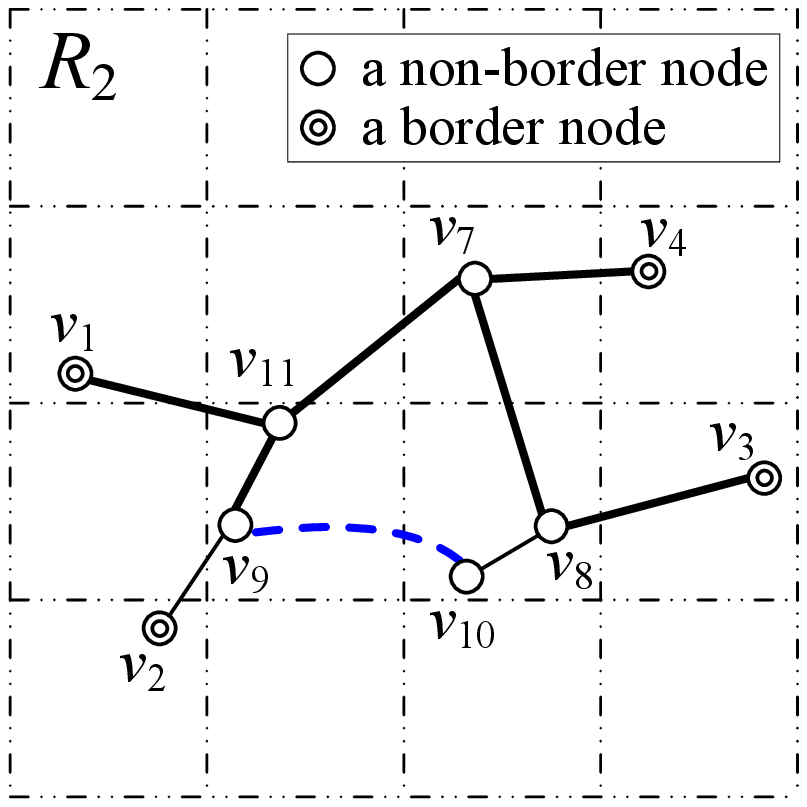}
%\vspace{-3.52mm}
\figcapup \caption{Reduce graph $G'_1$.} \figcapdown \label{fig:ah-G0reduced}
\end{minipage}
%\end{small}
\vspace{-2mm}
\end{figure}

For example, in Figure~\ref{fig:def-grid}, $v_1, v_2, v_9, v_{11}$ are all border nodes of the ($4 {\times} 4$)-cell region $B$, since each of them is an endpoint of an edge that intersects the boundary of $B$'s west strip, and none of them is contained in the $2\times2$ cells centered at $B$. Similarly, $v_3, v_4, v_7, v_8$ are also border nodes of $B$. On the other hand, $v_6$ and $v_{10}$ are not border nodes of $B$, since they are not adjacent to any edge that intersects the boundaries of $B$'s four strips.

To select level-($i{+}1$) cores from $G_{i}$, we first reduce $G_{i}$ by removing any node in $G_{i}$ that is neither a level-$i$ core nor a border node of any ($4 {\times} 4$)-cell region in $R_{i+1}$. We use $G'_{i}$ to denote the reduced graph thus obtained. For instance, given $G_1$ in Figure~\ref{fig:ah-G0}, we would remove $v_5$ and $v_6$, since none of them is a level-$1$ core or a border node in $R_2$. Figure~\ref{fig:ah-G0reduced} illustrates the reduced graph $G'_1$, with the border nodes in $R_2$ highlighted.

Given $G'_{i}$, we impose $R_{i+1}$ on $G'_{i}$ and inspect each ($4 {\times} 4$)-cell region in $R_{i+1}$ that contains at least one node. For each such region $B$, we compute every spanning path of $B$ (see Definition~\ref{def:def-ae}) that satisfies two conditions:
\begin{enumerate} [itemsep = 0.7mm, topsep = 1.5mm]
\item {\em Border Condition:} The two endpoints of the path are border nodes of $B$, while the other nodes are all level-$i$ cores.

\item {\em Coverage Condition:} Every shortcut on the path is generated from a region completely covered by $B$.
\end{enumerate}
For example, in the ($4 {\times} 4$)-cell region in Figure~\ref{fig:ah-G0reduced}, the spanning path $\langle v_2, v_9, v_{10}, v_8, v_3 \rangle$ satisfies both the border and coverage conditions, since (i) both $v_2$ and $v_3$ are border nodes, and (ii) the only shortcut on the path, $\langle v_9, v_{10} \rangle$, is generated from the region $B$ in Figure~\ref{fig:def-grid}, which is contained in the current ($4 {\times} 4$)-cell region.

%We have the following lemma:
%
%\begin{lemma} \label{lmm:ah-pseudo}
%Let $B$ be a ($4 {\times} 4$)-cell region in $R_i$ ($i \in [1, h]$), $P'$ be a path in $G'_{i-1}$, and $P$ be the path in $G$ that corresponds to $P'$. If $P'$ is a spanning path of $B$ that satisfies both the border and coverage conditions, then $P$ is a spanning path of $B$ in $G$.
%\end{lemma}

%By Lemma~\ref{lmm:ah-pseudo}, any spanning path $P$ in $G'_{i-1}$ that fulfills the border and coverage conditions can be mapped to a spanning path $P'$ in the original road network $G$. The correspondence between $P$ and $P'$ motivates us to identify important edges on $P$ in the same manner as on $P'$. In particular,

For each spanning path $P$ that fulfills the border and coverage conditions, if it connects the west and east (resp.\ north and south) strips of $B$, we identify the edge\footnote{If multiple edges or shortcuts in $P$ intersect the bisector, we choose an arbitrary one among them as the pseudo-arterial edge.}
%This does not affect the correctness or asymptotic bounds of AH.}
in $P$ that intersects $B$'s vertical bisector (resp.\ horizontal bisector) as a {\em pseudo-arterial edge} of $B$. Observe that each pseudo-arterial edge of $B$ corresponds to a path in $G$ that contains an arterial edge of $B$. Intuitively, this indicates the importance of pseudo-arterial edges in the reduced graph $G'_{i}$. Accordingly, we mark the two endpoints of every pseudo-arterial edge as level-($i{+}1$) cores, and we assign all unmarked level-$i$ cores to the $i$-th level of the node hierarchy $\H$. After that, for any local shortest path in a ($4 {\times} 4$)-cell region $B'$, if (i) the two endpoints of the path are either level-($i{+}1$) cores or border nodes of $B'$, and (ii) other than its endpoints, the path does not go through any level-($i{+}1$) core, then we insert into $G'_{i}$ a shortcut between $u$ and $v$ with the same length as the local shortest path. Once all such shortcuts are added, we define the resulting graph as $G'_{i+1}$, and use it to recursively compute higher-level nodes in $\H$.
%, as previously explained.

It remains to show that we can efficiently derive the pseudo-arterial edges and construct shortcuts in $G'_{i}$. Let $B$ be a ($4 {\times} 4$)-cell region in $R_{i+1}$, and $u$ be a border node of $B$. Suppose that we invoke Dijkstra's algorithm to start a traversal of $G'_{i}$ from $u$; for each node visited, we follow the outgoing edges of the node, ignoring any edge that violates the border condition or coverage condition. Once the traversal terminates, we can obtain the spanning paths of $B$ starting from $u$, as well as the pseudo-arterial edges on those edges. Similarly, with a traversal from $u$ that follows only the incoming edges of each node, we can compute the desired spanning paths of $B$ ending at $u$, along with the pseudo-arterial edges therein. By repeating this process on all border nodes of $B$, we can derive the set of all pseudo-arterial edges in $B$. With the same traversal algorithm, we can construct all shortcuts in $B$ using two traversals from each border node of $B$.

%The aforementioned spanning paths can be easily computed with a slightly modified version of Dijkstra's algorithm. In particular, for each border node $u$ of $B$, we first apply Dijkstra's algorithm to perform a traversal starting from $u$. For each node $v$ that we visit, if $v$ is a level-($i{-}1$) core node, then we inspect any edge $\langle v, w\rangle$ going out form $v$, and we put $w$ into the priority queue if $w$ is contained in $B$. If $v$ is not a level-($i{-}1$) node, then $v$ must be a border node, if which case we omit all edges adjacent to $v$. Once the traversal terminates, we obtain all spanning paths starting from $u$ that satisfy both the border and coverage conditions. Similarly, we can also apply Dijkstra's algorithm to derive all spanning paths ending at $u$.
%
%converts $G_{i-1}$ into a smaller graph $G_i$ and assigns nodes to the $i$-th level of $\H$.
%
%would repeat a graph {\em reduction} procedure on $G_0$ for $h-1$ times, such that the $i$-th
%
%would focus on the core nodes in $G_0$, and then identify the
%
% such that the $i$-th step ($i \in [1, h-1]$) converts $G_{i-1}$ into a smaller graph $G_i$ and assigns nodes to the $i$-th level of $\H$. Specifically, to reduce $G_{i-1}$ to $G_i$, AH first removes from $G_{i-1}$ any edge connecting two non-core nodes. After that, it imposes the grid $R_i$ on $G_{i-1}$, and it eliminates any non-core node that is not a {\em border node} of any ($4 {\times} 4$)-cell region in $R_i$. The concept of border node is defined as follows:

\header
{\bf Creation of Shortcuts.} After the level of each node is decided, AH adds shortcuts in the node hierarchy $\H$ to facilitate query processing. The construction of shortcuts requires as input a strict total order on the nodes in the same level of $\H$. We will elaborate our ordering approach in Section~\ref{sec:ah-optimize}, but in general, any strict total order can be used without affecting the space and time complexities of AH. For our discussion that follows, it suffices to know that less important nodes tend to precede more important nodes in our strict total order. For convenience, we define a rank for each node in $G$, such that a node $u$ ranks lower than another node $v$, if (i) $v$ is at a higher level than $u$, or (ii) $u$ and $v$ have the same level, but $u$ precedes $v$ in the strict total order.

AH constructs shortcuts in $\H$ in an incremental manner similar to the algorithm for deciding node levels. In particular, it first inspects $G$, and inserts into $\H$ a set of shortcuts that concern level-$0$ nodes. After that, it reduces $G$ to a smaller graph $\G_1$. Subsequently, it recursively reduces $\G_{i}$ into another graph $\G_{i+1}$ ($i \in [1, h-1]$), during which it constructs shortcuts that concern nodes at the $i$-th level of $\H$. In the following, we will elaborate the reduction from $\G_{i}$ to $\G_{i+1}$ ($i \in [0, h-1]$), assuming $\G_0 = G$. For convenience, we define the {\em level} of every edge in $\G_0$ as $-1$.

Given $\G_{i}$, AH first imposes the grid $R_{i+1}$ on $\G_{i}$. For each node $u \in \G_i$, AH inspects the ($5 {\times} 5$)-cell region $C$ centered at $u$, as well as the subgraph of $\G_i$ that consists of any level-($i{-}1$) edge overlapping with $C$. Then, AH computes two {\em shortest path trees (SPT)} of the subgraph, as defined in the following:
\begin{definition}[Shortest Path Trees (SPT)] \label{def:ah-spt}
Let $G$ be a graph, and $T$ be a directed spanning tree of $G$ rooted at a node $u$. $T$ is a {\bf forward SPT} of $G$, if $T$ contains the shortest path from $u$ to any node in $G$. On the other hand, if $T$ contains the shortest path from any node in $G$ to $u$, then $T$ is a {\bf backward SPT} of $G$.
\end{definition}

Let $T_f$ (resp.\ $T_b$) be the forward (resp.\ backward) SPT of the aforementioned subgraph that is rooted at $u$. Observe that $T_f$ (resp.\ $T_b$) can be computed by one traversal of the subgraph using Dijkstra's algorithm. Let $v$ be any node in $T_f$, such that $u$ ranks lower than $v$ but higher than any ancestor of $v$ in $T_f$. For any such $v$, AH generates a shortcut $\langle u, v\rangle$ with a length equal to the distance from $u$ to $v$ in $T_f$. In addition, AH associates $\langle u, v\rangle$ with a node $w$ on the path from $u$ to $v$, such that $w$ ranks higher than any node on the path except $u$ and $v$. This is to indicate that, when answering shortest path queries, AH can replace $\langle u, v\rangle$ with a two-hop path $\langle u, w, v\rangle$. (Our algorithm guarantees that such a two-hop path always exists.) Similarly, for any node $v$ in $T_b$, AH creates a shortcut $\langle v, u\rangle$, if $u$'s rank is lower than $v$'s but higher than those of $v$'s ancestors in $T_b$. Furthermore, the shortcut is associated with the node $w'$ that ranks the highest among $v$'s ancestors except $u$.
%For any node $v$ in the subgraph, let $P^+$ denote the path between $u$ and $v$ in $T^+$. If $v$ is the only node on $P^+$ that ranks higher than $u$, then AH adds a shortcut $\langle u, v\rangle$ into the node hierarchy $\H$, such that the length of $\langle u, v\rangle$ equals that of $P^+$.
%
%In addition, AH associates $\langle u, v\rangle$ with the node $w$ that immediately follows $u$ on $P^+$, i.e., $w$ is the neighbor of $u$ that lies on $P^+$.
We refer to the shortcuts constructed above as {\em level-$i$ edges}\footnote{If multiple shortcuts are constructed from one node $u$ to another node $v$, AH retains only the shortest one.}. Intuitively, these shortcuts connect each level-$i$ node $u$ directly to its nearby higher-rank nodes. By following these shortcuts during query processing, AH can avoid visiting less important nodes, which helps improve efficiency.

Besides the level-$i$ edges, AH creates a shortcut from $u$ and to a node $v$ in $T_f$ if (i) $u$ and $v$ are both at level $i$ or above, and (ii) all ancestors of $v$ except $u$ are below level $i$. Likewise, if $T_b$ contains a node $v$ with a level at least $i$, such that $u$ is the only ancestor of $v$ at level $i$ or higher, then AH generates a shortcut from $v$ to $u$. These shortcuts are to ensure that $\G_i$ would remain connected when we reduce $\G_i$ by removing some nodes below level $i$, as will be clarified shortly.

When $i > 0$ (i.e., $\G_{i}$ is produced from a previous reduction step), AH also generates some extra shortcuts (referred to as {\em elevating edges}), in a manner slightly different from the construction of level-$i$ edges. First, AH inspects each node $u$ in $\G_{i}$ at a level lower than $i$, and it examines the ($5 {\times} 5$)-cell region $C$ in $R_{i}$ that is centered at $u$. Then, AH constructs a subgraph of $\G_{i}$ that comprises of all level-$i$ edges covered by $C$, {\em as well as all edges that connect $u$ with any node at level $i$ or above}. After that, AH computes the subgraph's forward and backward SPTs rooted at $u$. Let $P_f$ be any path in the forward SPT that connects $u$ to a node outside of $C$, and let $v$ be the first node on $P_f$ at level $i$ or above (our algorithm ensures that such $v$ always exits).
AH constructs a shortcut from $\langle u, v\rangle$, and associates it with the node that immediately follows $u$ on $P_f$, {\em if $u$ is below level $i-1$}. On the other hand, if $u$ is at level $i-1$, then the shortcut is associated with the first node on $P_f$ that ranks higher than $u$.
%For every node $w$ preceding $v$ on $P_f$ (including $u$), we construct a shortcut $\langle w, v\rangle$, and we associate it with the node that immediately follows $w$ on $P_f$.
This shortcut is constructed to enable AH to efficiently traverse from $u$ to the $i$-th level of $\H$.
%, while the other shortcuts $\langle w, v\rangle$ are generated to ensure that AH can map $\langle u, v\rangle$ to its corresponding path in $G$
Similarly, if the backward SPT contains a path $P_b$ that links $u$ with a node located beyond $C$, AH creates a shortcut $\langle v, u\rangle$, where $v$ is the node closet to $u$ on $P_b$ among those at level $i$ or above. If $u$ is below level $i-1$, the shortcut is associated with the node that immediately precedes $u$ on $P_b$; otherwise, it is associated with the node that is closest to $u$ on $P_b$ among those with higher ranks than $u$.

%For every node $w$ that precedes $v$ on $P_b$, we create a shortcut $\langle w, u\rangle$, and we associate it with the node that immediately precedes $w$ on $P_b$.

%If the forward SPT contains a path $P_f$ from $u$ to a node $v$, such that $v$ is the only node on $P_f$ that ranks higher than $u$, then AH creates a shortcut $\langle u, v\rangle$ and associates it with the node that immediately succeeds $u$ on $P_f$. Similarly, for any path $P_b$ in the backward SPT that connects $u$ with a node $v$, if $u$ ranks higher than all other nodes on $P_b$ except $v$, then AH adds a shortcut $\langle v, u\rangle$ and associates it with the node the immediately follows $v$ on $P_b$. Intuitively, such elevating edges can help reduce query overhead, since they enable AH to traverse from lower-level nodes directly to higher-level nodes.

%Let $P_f$ be a path in the forward SPT that connects $u$ to a node outside of $C$ (This occurs when the node is adjacent to an edge intersecting $C$'s boundary), and let $v$ be the first node on $P_f$ with a rank higher than $u$. For any such path $P_f$, we construct a shortcut $\langle w, v\rangle$ from every node $w$ preceding $v$ on $P_f$, and we associate the shortcut with the node that immediately follows $w$ on $P_f$.

Once all level-$i$ edges and elevating edges are created, they are inserted into both $\G_i$ and $\H$. After that, AH reduces $\G_i$ by retaining only (i) the border nodes in $R_{i+2}$ and (ii) nodes at level $i$ or above. The resulting graph is defined as $\G_{i+1}$ and is fed into the next reduction step.

\subsection{Query Processing}\label{sec:ah-query}

The query processing algorithm of AH is similar to that of FC. In particular, for distance query from a node $s$ to another node $t$, AH also answers the query with two traversals of the node hierarchy $\H$ starting from $s$ and $t$, respectively. As with the case of FC, each traversal of AH is performed with a constrained version of Dijkstra's algorithm. However, the constraints adopted by AH are slightly different: It adopts the proximity constraint (see Section~\ref{sec:fc-query}) and a {\em rank constraint} as follows:
\begin{itemize}
\item {\em Rank Constraint:} When the traversal from $s$ (resp.\ $t$) visits a node $u$, it ignores any neighbor of $u$ that ranks lower than $u$.
\end{itemize}
Intuitively, the rank constraint is a refined version of the level constraint, in that it takes into account not only the levels of nodes but also the strict total order defined on each level of $\H$. It leads to higher query efficiency as it helps AH bypass a larger number of relative unimportant nodes during query processing.

In addition, AH also exploits the elevating edges in $\H$ (see Section~\ref{sec:ah-pre}) to reduce query cost, based on the following lemma:
\begin{lemma} \label{lmm:ah-longpath}
For any two nodes $u,v \in G$, if no ($3 {\times} 3$)-cell region in $R_i$ ($i \in [1, h]$) can cover $u$ and $v$ simultaneously, then the shortest path from $u$ to $v$ must go through a node at level $i$ or above.
\end{lemma}
Let $R_j$ ($j \in [1, h]$) be the coarsest grid where no ($3 {\times} 3$)-cell region contains both $s$ and $t$. By Lemma~\ref{lmm:ah-longpath}, the shortest path from $s$ to $t$ should pass through at least one node with a level at least $j$. This indicates that AH's traversal from $s$ would meet its traversal from $t$ at level $j$ or above. Therefore, if $s$ is a border node in $R_j$ (in which case $s$ has elevating edges to level $j$), then when we start the traversal from $s$, we can follow the elevating edges of $s$ to move directly to level $j$, ignoring any edge that connects $s$ to a node at a level lower than $j$. After that, we can continue the traversal from level $j$ under the rank and proximity constraints.

More generally, for any level-$i$ ($i < j$) node $v$ visited in the traversal from $s$, if $v$ is a border node in $R_j$, then we move along the elevating edges of $v$ to level $j$ or above, omitting any other edges of $v$. On the other hand, if $v$ is a border node in $R_{j'}$ instead of $R_j$  ($j' < j$), then we follow the elevating edges $v$ to level $j'$ or higher, i.e., we traverse as close to level $j$ as possible. Meanwhile, if $v$ does not have any elevating edges or $v$ is at a level at least $j$, then we traverse the edges of $v$ that satisfies the rank and proximity constraints. The same strategy is used when AH traverses from $t$. This traversal strategy reduces query time, since it enables AH to avoid visiting the low levels of node hierarchy $\H$.

So far we have only discussed distance queries. For any shortest path query from $s$ to $t$, AH first treats it as a distance query and computes the shortest path $P'$ from $s$ to $t$ in $\H$. After that, AH recursively replaces each shortcut in $P'$ with its corresponding two-hop path, which converts $P'$ to the actual shortest path from $s$ to $t$ in $G$, as explained in Section~\ref{sec:ah-overview}.

\subsection{Node Ranking and Selection}\label{sec:ah-optimize}

As mentioned, the shortcut construction algorithm of AH assumes that there is a strict total order on the nodes in the same level. While any strict total order can be used without affecting the asymptotic bounds of AH, we have found a heuristic ordering approach that leads to high practical performance. Specifically, for nodes in the $0$-th level of the node hierarchy $\H$, we adopt a random order; for nodes in the $i$-th level ($i \in [1, h]$) of $\H$, we derive their ordering based on information from the preprocessing procedure of AH. To explain, recall that AH decides node levels by recursively applying a reduction procedure on the road network $G$. During the $i$-th reduction step ($i \in [1, h-1]$), AH examines a graph that contains level-($i{-}1$) cores; It identifies a set $S_i$ of pseudo-arterial edges in the graph, marks the endpoints of those edges as level-$i$ cores, and then assigns all unmarked level-($i{-}1$) cores to the ($i{-}1$)-th level of $\H$.

We observe that the edges in $S_i$ are connected to some extend, and there are some level-$i$ cores that serve as hub nodes for the connections (i.e., they are adjacent to a sizable number of edges in $S_i$). Intuitively, those hub nodes are more important than the rest of the level-$i$ cores. Motivated by this, we order the level-$i$ cores using a {\em vertex cover} approach: we inspect the graph formed by the edges in $S_i$, and we compute a vertex cover of the graph using the linear-time $O(\log n)$-approximation algorithm \cite{IntroAlgo}. The output of the algorithm is a sequence $\xi$ of nodes in the graph, such that the $i$-th node $v$ in $\xi$ is adjacent to the largest number of edges that are disjoint from the first $i-1$ nodes. Based on $\xi$, we order the level-$i$ cores as follows: The $i$-th node in $\xi$ is given the $i$-th highest rank, and the level-$i$ cores not in $\xi$ are given the lowest ranks arbitrarily.

Interestingly, we find that if a level-$i$ core does not appear in $\xi$, then we can downgrade it to a level-($i{-}1$) core without affecting correctness or asymptotic performance of AH. Such downgrading reduces the number of high-level nodes in the node hierarchy $\H$, which in turn improves query efficiency, since the high levels of $\H$ are frequently traversed during query processing. Our implementation of AH adopts this downgrading approach to improve query performance.

%In our implementation of AH, we downgrade a level-$i$ core to level $i{-}1$ if it is absent from $\xi$ and is adjacent to less than $10$ edges in $S_i$.

\subsection{Space and Time Complexities}\label{sec:ah-complexity}

To establish the space and time complexities of AH, we first introduce a lemma that quantifies the densities of nodes in each level of AH's node hierarchy $\H$:
\begin{lemma} \label{lmm:ah-node-num}
Any ($\alpha {\times} \alpha$)-cell region in $R_i$ contains $O(\alpha^2 \a^2)$ nodes whose level in $\H$ are no lower than $i$, where $\a$ is the arterial dimension of $G$.
\end{lemma}
Our proof of Lemma~\ref{lmm:ah-node-num} is similar to that of Lemma~\ref{lmm:fc-node-num}. We first show that any ($\alpha {\times} \alpha$)-cell region $B$ in $R_i$ contains the endpoints of $O(\alpha^2 \a)$ arterial edges in $G$. After that, we prove that there exists a one-to-many mapping from the arterial edges in $G$ to the nodes in $B$ with levels at least $i$, such that each edge is mapped to $O(\lambda)$ nodes. Based on this, we show that any ($\alpha {\times} \alpha$)-cell region in $R_i$ contains only $O(\alpha^2 \a^2)$ nodes at level $i$ or above.

%level-$i$ nodes of $\H$.

\vspace{-0.3mm} \header
{\bf Space Overhead.} Given Lemma~\ref{lmm:ah-node-num}, we can prove that each node in $\H$ has $O(h\a^2)$ elevating edges and $O(\a^2)$ non-elevating edges. This is because, by the preprocessing algorithm of AH, there is an elevating edge from a node $u$ to a level-$i$ node $v$, only if $u$ and $v$ are contained in the same ($4 {\times} 4$)-cell region in $R_i$. By Lemma~\ref{lmm:ah-node-num}, there exist $O(\a^2)$ such level-$i$ nodes. Since $\H$ contains only $h+1$ levels, the total number of elevating edges adjacent to $u$ is $O(h\a^2)$. Similarly, we can prove that each node in $\H$ has $O(\a^2)$ non-elevating edges. Therefore, the space overhead of AH is $O(hn\a^2)$, which reduces to $O(hn)$ when $\a$ is a constant.

\begin{table*}
\centering
\tblcapup \vspace{-1mm}
\begin{small}
\caption{Asymptotic performance of the state of the art.} \tblcapdown \vspace{-1mm}
\label{tbl:related-complexity}
\setlength{\extrarowheight}{0.1mm}
%\begin{tabular}{|@{\hspace{1mm}}c@{\hspace{1mm}}|@{\hspace{1mm}}c@{\hspace{1mm}}|@{\hspace{1mm}}c@{\hspace{1mm}}|@{\hspace{1mm}}c@{\hspace{1mm}}|@{\hspace{1mm}}c@{\hspace{1mm}}|@{\hspace{1mm}}c@{\hspace{1mm}}|}
%\begin{tabular}{|@{\hspace{1mm}}c@{\hspace{1mm}}|@{\hspace{1mm}}c@{\hspace{1mm}}|c|c|c|c|} \hline
\begin{tabular}{|c|c|c|c|c|c|} \hline
{\bf Reference} & {\bf  Space } &
{\bf  Preprocessing} & {\bf Distance Query } & {\bf Shortest Path Query} & {\bf Remark}\\ \hline
\multirow{2}{*}{\cite{ms12}} &  \gape[t]{$O(n)$} & $O(n \log n)$ & $O(n^{0.5+\epsilon})$ & $O(k + n^{0.5+\epsilon})$ & \multirow{2}{5.4cm}{$S$ is a user-defined parameter in the range of $[n\log\log n, n^2]$. $D = l_{max}/l_{min}$, where $l_{max}$ (resp.\ $l_{min}$) is the largest (resp.\ smallest) road network distance between two nodes. $k$ is the number of edges in the shortest path between the source and destination of the query. $h$ is as defined in Section~\ref{sec:fc-pre}.}\\ \cline{2-5}
 & $O(S)$ & \gape[t]{$\tilde{O}(S)$} & $\tilde{O}(n/\sqrt{S})$ & $\tilde O(k + n/\sqrt{S})$  & \\ \cline{1-5}
\multirow{2}{*}{\cite{afg10}} &  \gape[t]{$O(n \log n \log D)$} & $O(n^2 \log n)$ & $O(\log^2 n \log^2 D)$ & $O(k + \log^2 n \log^2 D)$ & \\ \cline{2-5}
 & \gape[t]{$O(n \log n \log D)$} & $O(n^2 \log n)$ & $O(\log n \log D)$ & N/A & \\ \cline{1-5}
\cite{ssa08} &  $O(n\sqrt{n})$ & $O(n^2 \log n)$ & \gape[t]{$O(k\log n)$} & $O(k\log n)$ & \\ \cline{1-5}
this paper &  $O(hn)$ & $O(hn^2)$ & \gape[t]{$O(h \log h)$} & $O(k + h\log h)$ &  \\ \hline
%\cite{ssa09}  & road networks & $O(n)$ & $O(n^2 \log n)$ & $O(\log n \log D)$ & N/A \\ \cline{3-6}
\end{tabular}
\end{small}
\tbldown \vspace{-1mm}
\end{table*}

\vspace{-0.3mm} \header
{\bf Query Time.} AH answers any distance query with two traversals of $\H$, starting from the source $s$ and destination $t$ of the query, respectively. Due to the proximity constraint, in the $i$-th level of $\H$ ($i \in [0, h]$), each traversal of AH only visits the nodes in a ($5 {\times} 5$)-cell region in $R_{i+1}$. By Lemma~\ref{lmm:ah-node-num}, such a region only contains $O(\a^2)$ nodes at level $i$. Hence, the total number of nodes traversed by AH is $O(h \a^2)$. Furthermore, for each node $v$ visited during a traversal, AH either follows the elevating edges of $v$ to a certain level of $\H$, or moves along the non-elevating edges of $v$ that satisfy the rank and proximity constraints. As previously discussed, $v$ has $O(\a^2)$ elevating edges to each level of $\H$, and has $O(\a^2)$ non-elevating edges. Therefore, the total number of edges visited by AH is $O(h \a^4)$. Since each traversal is performed using Dijkstra's algorithm, its overall time complexity is $O(h\a^2 \log(h\a^2) + h\a^4)$. Consequently, when $\a$ is a constant, the time complexity of AH for a distance query is $O(h \log h)$.

To answer a shortest path query from $s$ to $t$, AH first processes its corresponding distance query to retrieve the shortest path $P'$ from $s$ to $t$ in $\H$, and then it transforms $P'$ into the actual shortest path $P$ from $s$ to $t$ in $G$. The transformation from $P'$ to $P$ takes $O(k)$ time, where $k$ is the number of edges in $P$. Therefore, AH requires $O(k + h \log h)$ time to answer a shortest path query.

\header
{\bf Preprocessing Cost.} The preprocessing algorithm of AH consists of three steps: (i) assigning nodes to each level of $\H$, (ii) deriving the strict total order on nodes at the same level, and (iii) creating shortcuts in $\H$. When assigning nodes to the $i$-th level of $\H$ ($i \in [0, h-1]$), AH inspects each non-empty ($4 {\times} 4$)-cell region in $R_{i+1}$, and construct a subgraph that consists of the level-$i$ cores and border nodes in the region. For each node in the subgraph, AH needs to apply Dijkstra's algorithm to traverse the subgraph a constant number of times. Given Lemma~\ref{lmm:ah-node-num} and the fact that each level-$i$ core or border node (i) has $O(\lambda^2)$ edges and (ii) is contained in a constant number of ($4 {\times} 4$)-cell region in $R_{i+1}$, it can be proved that AH requires $O(n^2 \lambda^2)$ time to assign nodes to level $i$ of $\H$. Meanwhile, AH takes only $O(n)$ time to derive the strict total order at level $i$ of $\H$, since the derivation is based on a linear time algorithm for vertex cover.

To construct shortcuts at the $i$-th level of $\H$, AH needs to inspect a graph $G^*_i$ reduced from $G$. For each node $u$ in $G^*_i$, AH examines the a ($5 {\times} 5$)-cell region in $R_{i+1}$ that is centered at $u$, and it creates shortcuts for $u$ by traversing the nodes in the region at level $i$ or above. Based on Lemma~\ref{lmm:ah-node-num}, it can be proved that the total cost of generating shortcuts for $u$ is $O(\a^2)$. As such, the time required to create shortcuts at level $i$ of $\H$ is $O(n \a^2)$.

Summing up the above analysis, we have the following theorem:
\begin{theorem} \label{thrm:ah-complexity}
Given a road network with a constant arterial dimension, AH takes $O(h n^2)$ time to construct an index that requires $O(h n)$ space. With the index, AH answers any distance query in $(h \log h)$ time and any shortest path query in $O(k + h \log h)$ time, where $k$ is the number of edges in the shortest path.
\end{theorem}

\section{Related Work}\label{sec:related}

%p71
Numerous techniques (e.g., \cite{ssa08,ssa09,ss10,gss08,bfm06,tsp11,rt10,gh05,gkw06,kmw10,hkm06,dss09,fr06,afg10,bfs07,jp02,jhr98,gkr04,ms12}) have been proposed for processing shortest path and distance queries on road networks. Many of these techniques focus on practical performance, and they are mostly heuristic-based. For example, ALT \cite{gh05} pre-computes the road network distances from each node to a fixed set of nodes (referred to as {\em landmarks}), and then utilizes those pre-computed distances to reduce the search space of each query. Hiti \cite{jp02} partitions the road network into vertex-disjoint subgraphs, and then pre-computes the shortest paths that connect different subgraphs to facilitate query processing. We refer the reader to \cite{wxd12} for a survey of the existing heuristic-based techniques.

In addition, there also exists a large number of worst-case efficient algorithms for shortest path and distance queries (see \cite{fr06,gkr04,kmw10,ms12,ssa08,ssa09} and the references therein). Most of these algorithms assume that the road network is a planar graph with non-negative weights, while some recent work \cite{afg10,ssa08,ssa09} adopts more subtle assumptions on the road network to derive tighter bounds on space and time complexities. Table~\ref{tbl:related-complexity} lists the performance bounds of several most recent algorithms. Compared with the state of the art, our method offers superior query efficiency while incurring moderate costs of space and pre-computation.

The work most related to ours is by Bast et al.\ \cite{bfm06}, Abraham et al.\ \cite{afg10}, and Geisberger et al.\ \cite{gss08}. Bast et al.\ \cite{bfm06} observe that, in practice, there often exist a small set $S$ of nodes in the road network (referred to as {\em transit nodes}), such that any shortest path connecting two distant locations must pass through at least one node in $S$. Based on this observation, Bast et al.\ propose a heuristic solution for answering shortest path and distance queries. However, the proposed solution is shown to be flawed in that it may return incorrect query results \cite{wxd12}. Our notion of arterial dimension is motivated by Bast et al.'s observation, but our definition of arterial edges is considerably different from Bast et al.'s formulation of transit nodes.

Abraham et al.\ \cite{afg10} introduce a theoretical abstraction of Bast et al.'s observation, based on which they propose several worst-case efficient algorithms for shortest path and distance queries. The proposed algorithms adopt an assumption that is similar in spirit to our Assumption~\ref{assu:def-ae}, but is more elegant in a theoretical sense. Nevertheless, the assumption adopted by Abraham et al.\ has not been tested on any real road networks, while our Assumption~\ref{assu:def-ae} is backed by empirical evidence from real datasets, as shown in Section~\ref{sec:def}. Furthermore, Abraham et al.'s algorithms require pre-computing the shortest path between any pair of nodes in the road network, which renders them inapplicable even for moderate-size datasets.

Geisberger et al.\ \cite{gss08} propose a road network index called the {\em Contraction Hierarchies}, which (i) heuristically imposes a total order on the road network nodes and (ii) constructs shortcuts from low-rank nodes to high-rank nodes to enable efficient query processing. Our AH method is inspired by CH, and it outperforms CH in terms of both asymptotic and practical performance, as will be shown in Section~\ref{sec:exp}.

\section{Experiments}\label{sec:exp}

This section experimentally compares our AH method with three techniques: (i) Dijkstra's algorithm \cite{d59}, (ii) {\em Spatially Induced Linkage Cognizance (SILC)} \cite{ssa08}, one of the most advanced worst-case efficient indices for shortest path and distance queries, and (iii) {\em Contraction Hierarchies (CH)} \cite{gss08}, a heuristic approach that offers the highest overall efficiency in shortest path and distance queries while incurring minimal costs of space and pre-computation, as shown in a recent experimental study \cite{wxd12} of the state of the art. We implement AH and Dijkstra's algorithm using C++, and we obtain the C++ implementations of SILC and CH from \cite{SILCcode,CHcode}. All experiments are conducted on a $64$-bit windows machine with an Intel Xeon 2.8GHz CPU and 32GB RAM.

\begin{figure*}[t]
\begin{small}
\begin{tabular}{cccc}
\multicolumn{4}{c}{\hspace{-4mm} \includegraphics[height=3.6mm]{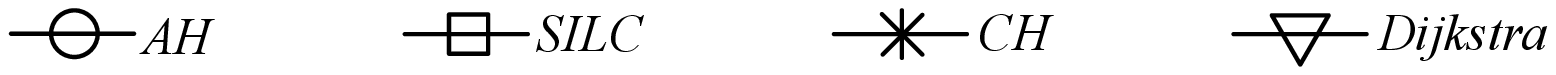}} \\
\hspace{-9.5mm}\includegraphics[height=37mm]{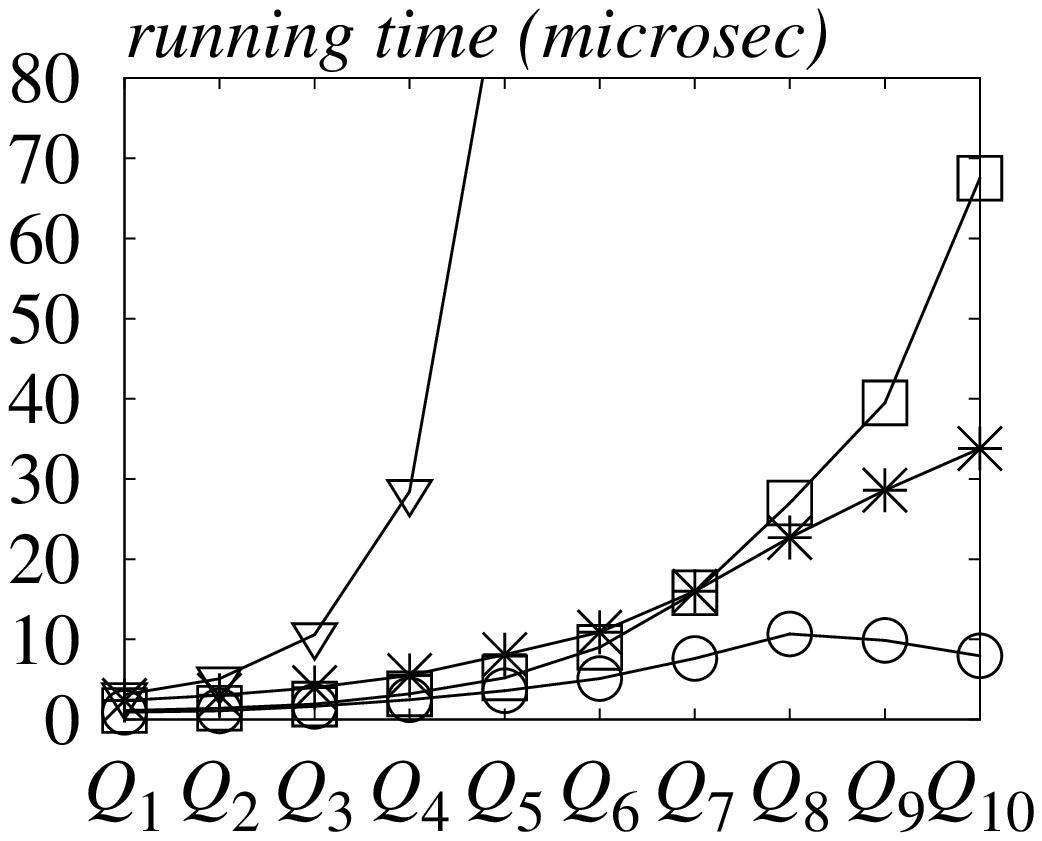}
&
\hspace{-12.5mm}\includegraphics[height=37mm]{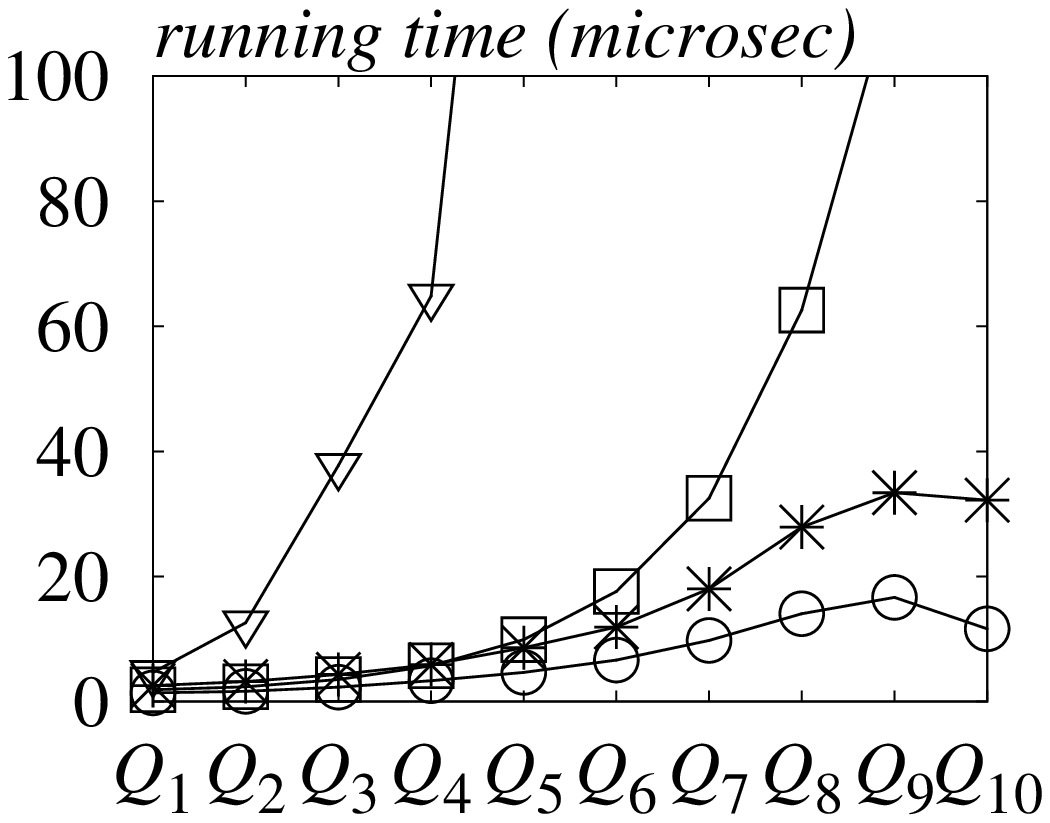}
&
\hspace{-12.5mm}\includegraphics[height=37mm]{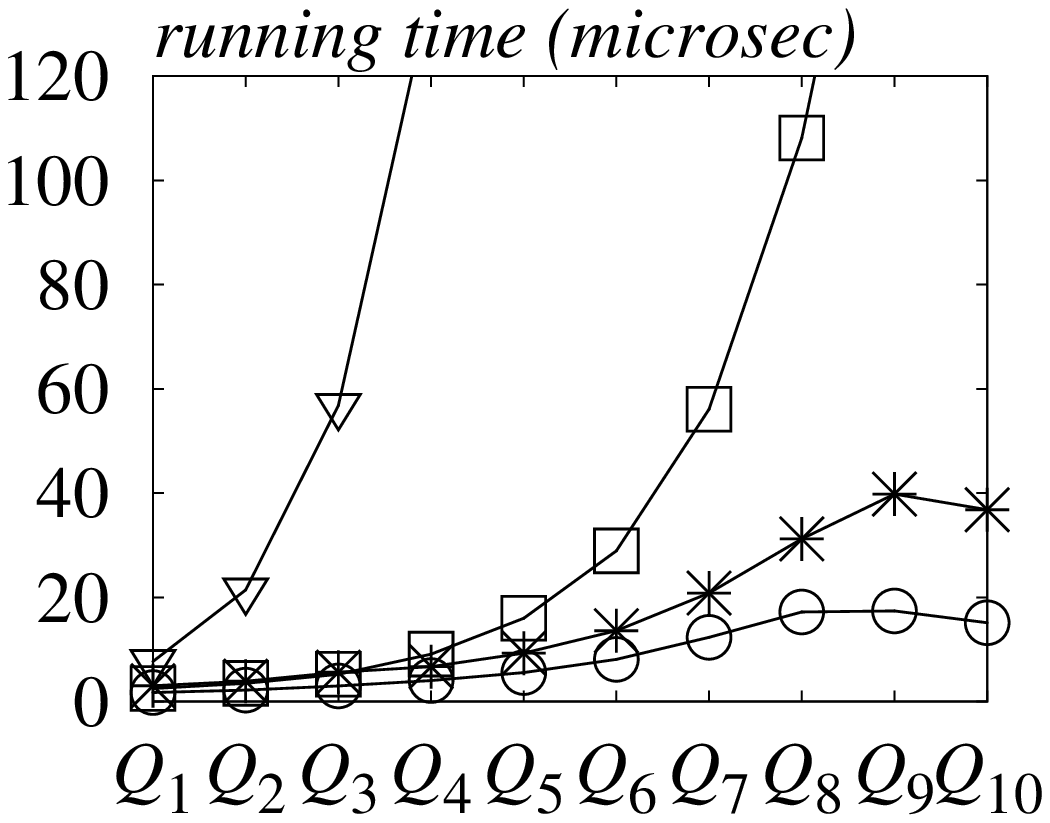}
&
\hspace{-12.5mm}\includegraphics[height=37mm]{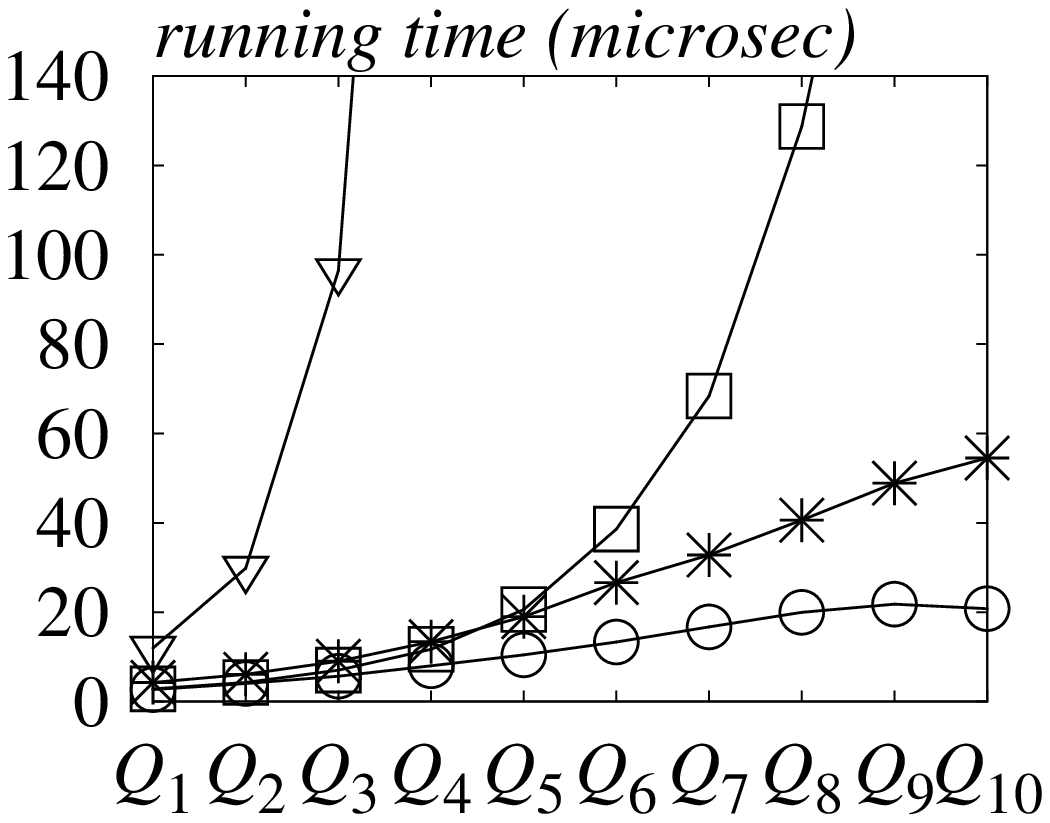} \\

\hspace{-6mm}(a) DE ($n = 48,812$) & \hspace{-6mm}(b) NH ($n = 115,055$) &
\hspace{-6mm}(c) ME ($n = 187,315$) & \hspace{-6mm} (d) CO ($n = 435,666$)\\ \vspace{-1.5mm} \\

\hspace{-9.5mm}\includegraphics[height=37mm]{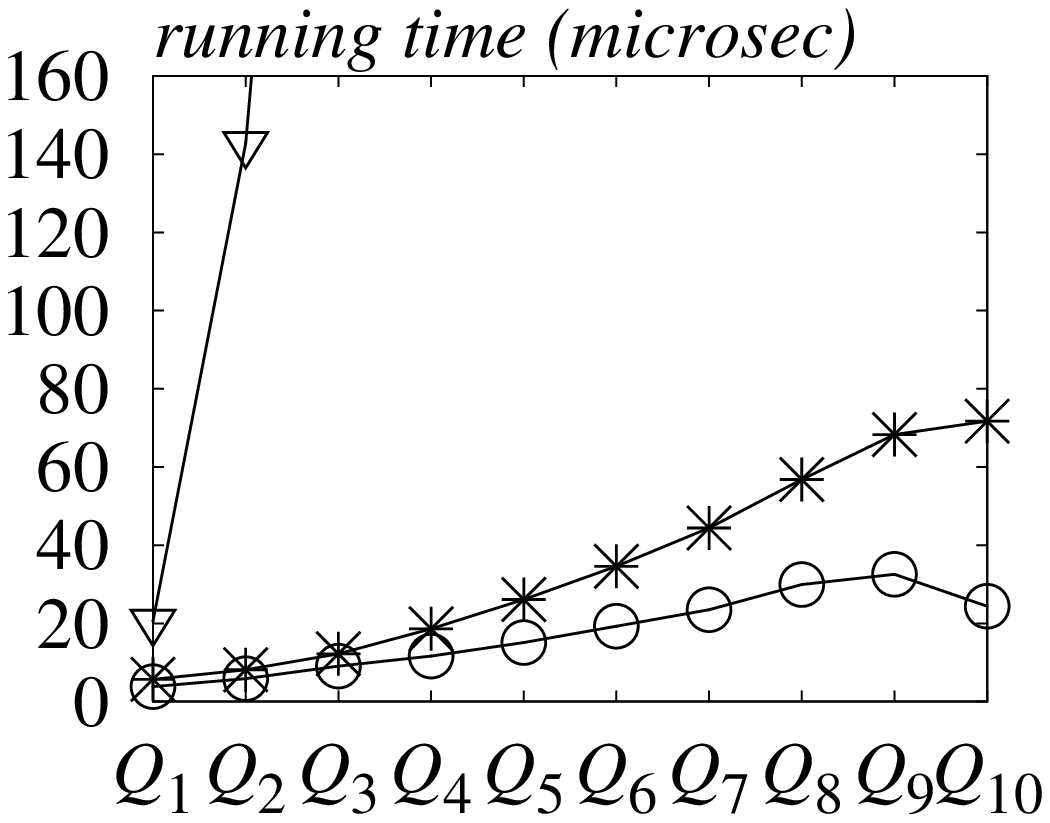}
&
\hspace{-12.5mm}\includegraphics[height=37mm]{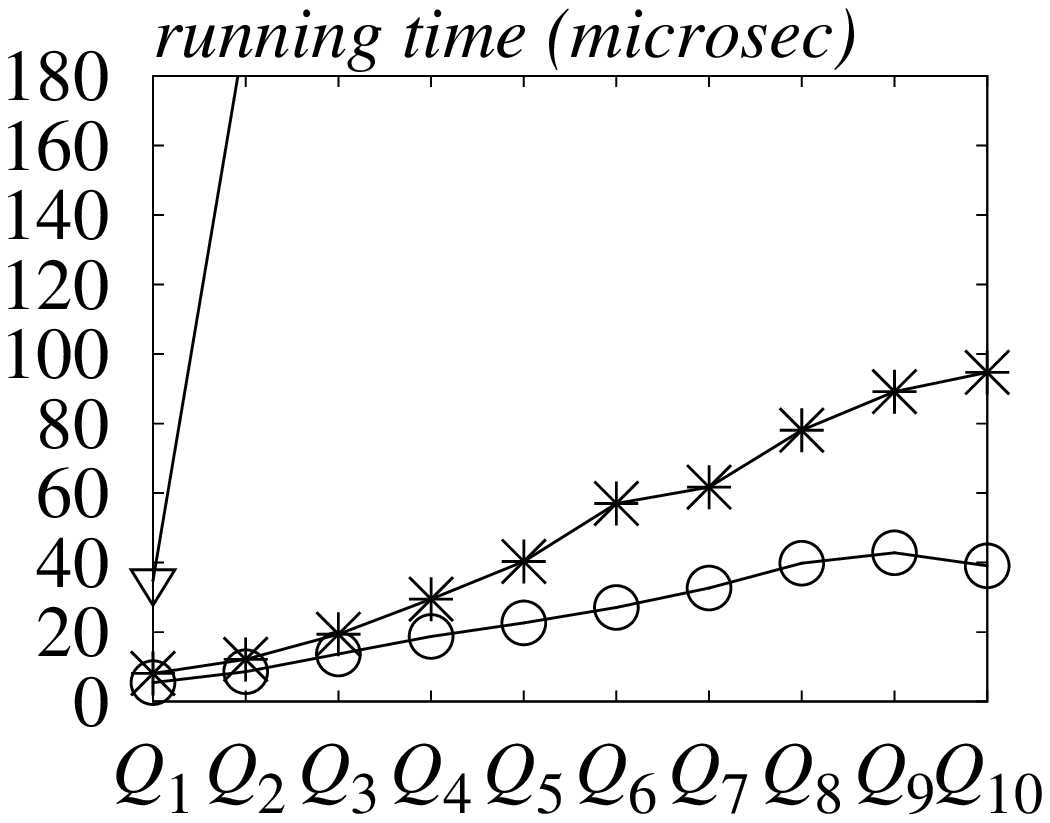}
&
\hspace{-12.5mm}\includegraphics[height=37mm]{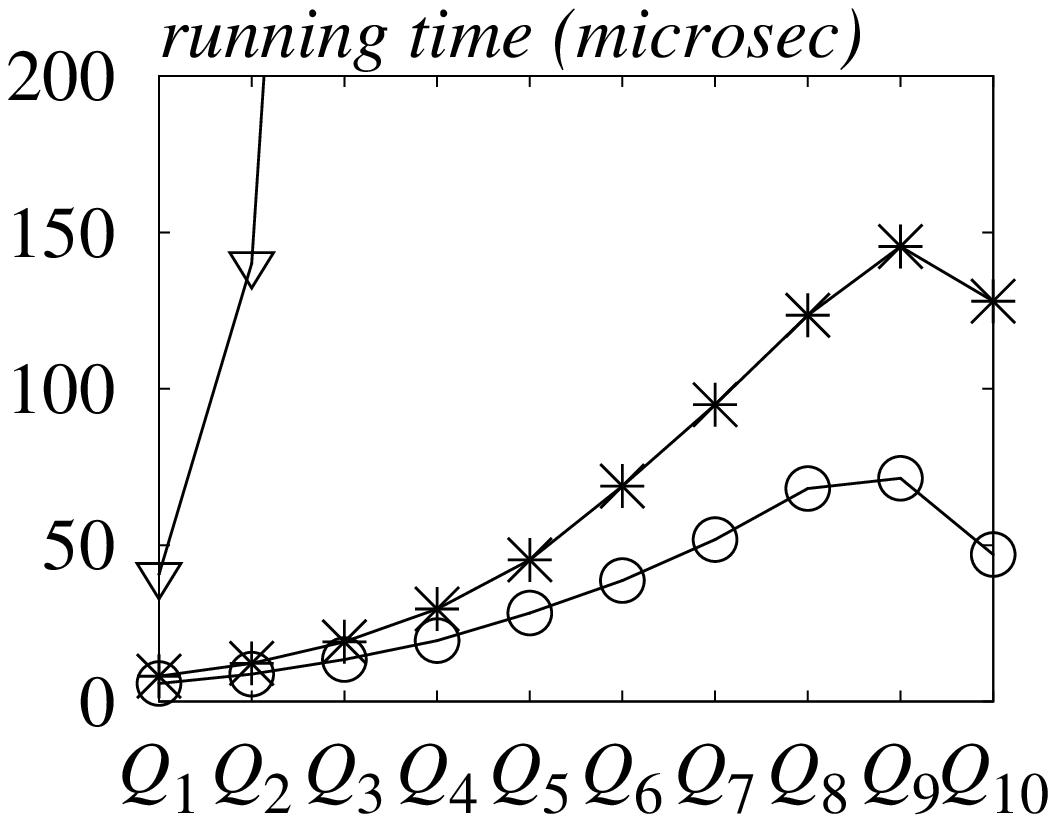}
&
\hspace{-12.5mm}\includegraphics[height=37mm]{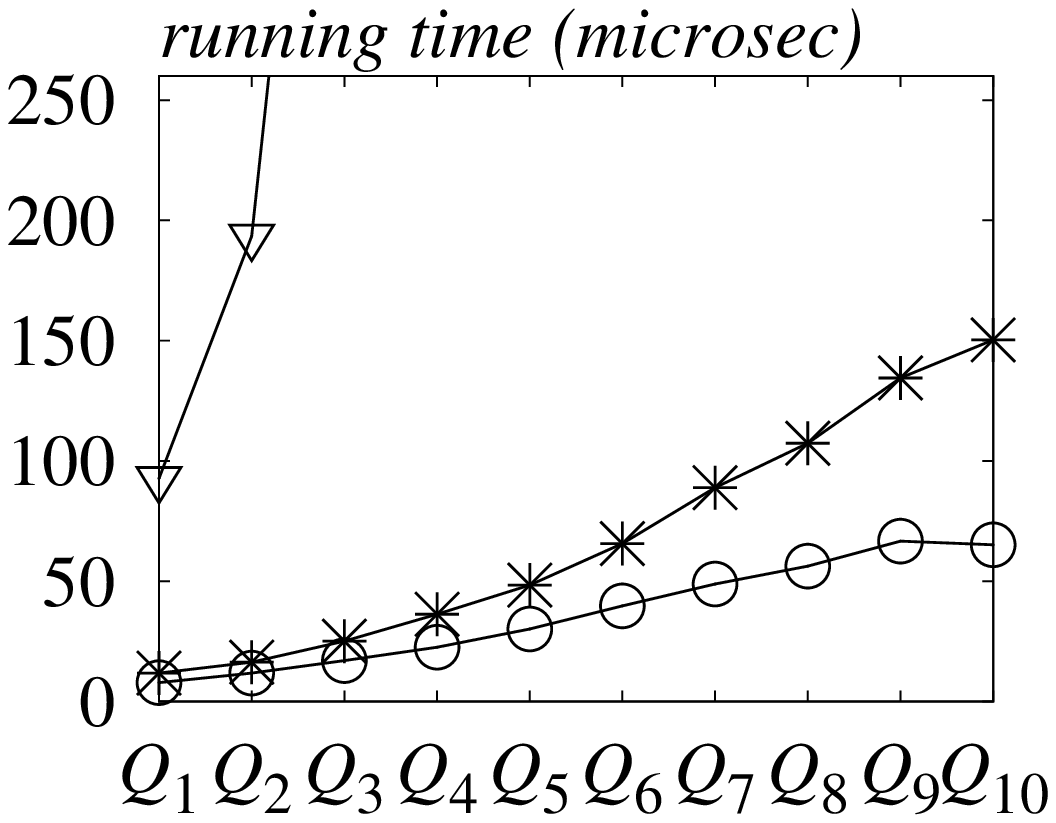} \\
\hspace{-6mm}(e) FL ($n = 1,070,376$) & \hspace{-6mm}(f) CA ($n = 1,890,815$) &
\hspace{-6mm}(g) E-US ($n = 3,598,623$) & \hspace{-6mm} (h) W-US ($n = 6,262,104$) \\ \vspace{-1.5mm} \\

\hspace{-9.5mm}
&
\hspace{-12.5mm}\includegraphics[height=37mm]{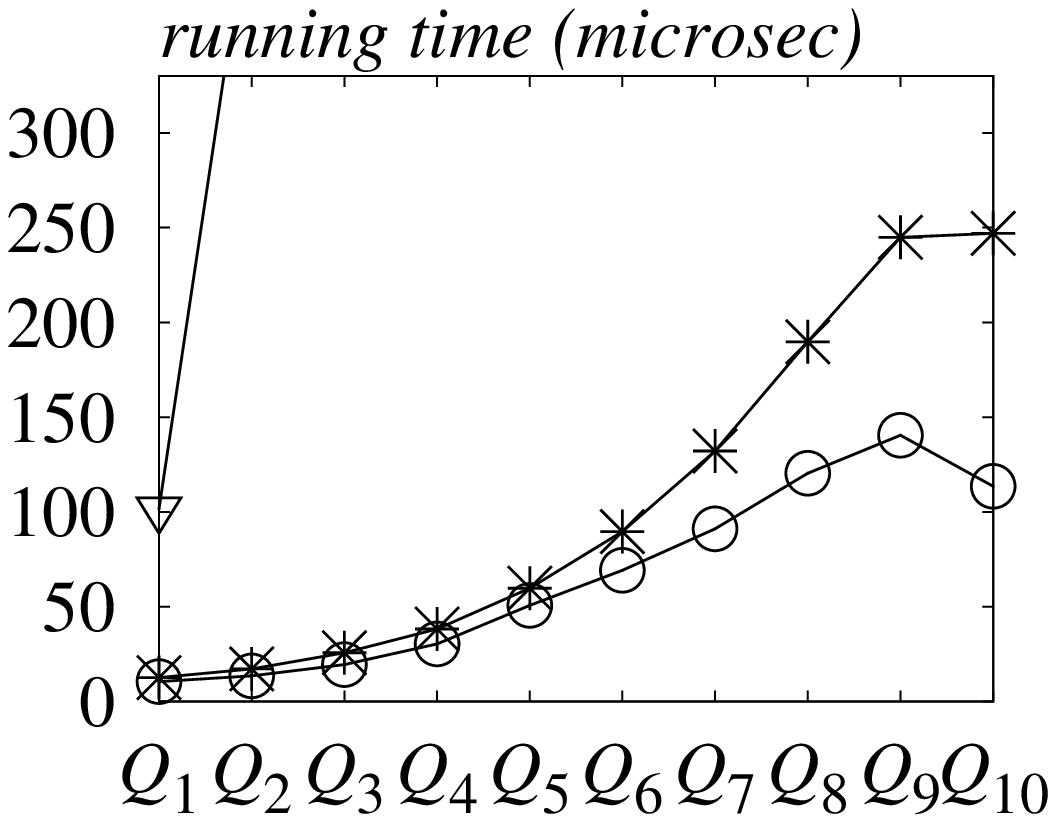}
&
\hspace{-12.5mm}\includegraphics[height=37mm]{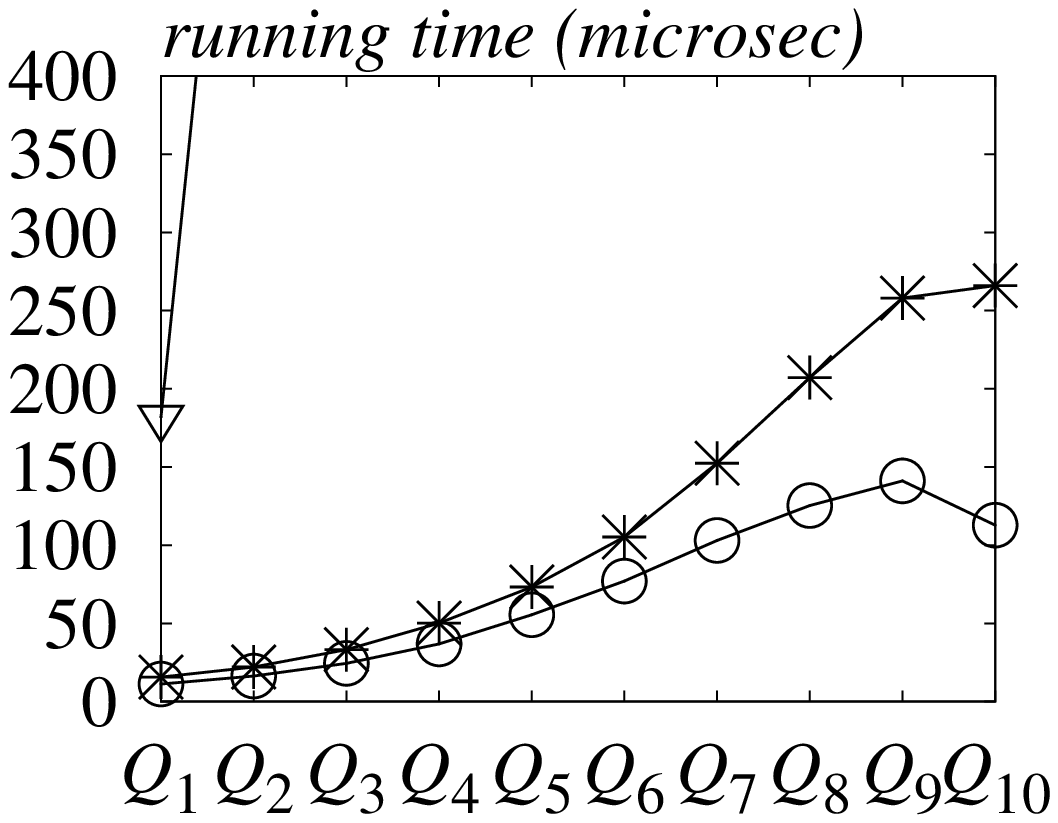}
&
\hspace{-12.5mm} \\
\hspace{-6mm} & \hspace{-6mm}(i) C-US ($n = 14,081,816$) &
\hspace{-6mm}(j) US ($n = 23,947,347$) & \hspace{-6mm}
\end{tabular}
\end{small}
\figcapup \caption{Efficiency of distance queries vs.\ query set.} %\vspace{-1mm} %\figcapdown
\label{fig:exp-dist-vary-q}
\end{figure*}

\subsection{Datasets and Queries}\label{sec:exp-settings}

We use ten publicly available datasets \cite{dimacs}, each of which corresponds to a part of the road network in the US. Table~\ref{tbl:exp-data} shows the number of nodes and edges in the data. For each edge in the datasets, its weight quantifies the time required to traverse the road segment that is represented by the edge.

Following previous work \cite{wxd12}, we generate ten sets of queries $Q_1, Q_2, \ldots, Q_{10}$ on each dataset as follows. We first estimate the maximum network distance $l_{max}$ between two nodes in the road network. After that, we insert $10000$ pairs of nodes $(s, t)$ into $Q_i$ ($i \in [1, 10]$) as queries, such that the distance between $s$ and $t$ is in $[2^{i-11}\!\cdot\!l_{max},\,2^{i-10}\!\cdot\!l_{max})$. In other words, the network distance between any pair of nodes in $Q_i$ is larger than that in $Q_{i-1}$.

\subsection{Efficiency for Distance Queries}\label{sec:exp-distance}

\begin{table}
\centering
\begin{small}
\tblcapup
\caption{Dataset Characteristics} \tblcapdown
\label{tbl:exp-data}
\begin{tabular}{|@{\hspace{1mm}}c@{\hspace{1mm}}|@{\hspace{1mm}}c@{\hspace{1mm}}|@{\hspace{1mm}}c@{\hspace{1mm}}|@{\hspace{1mm}}c@{\hspace{1mm}}|} \hline
{\bf Name} & {\bf Corresponding Region} & {\bf Number of Nodes} &
{\bf Number of Edges} \\ \hline
DE & Delaware & $48{,}812$ & $120{,}489$ \\ \hline
NH & New Hampshire & $115{,}055$ & $264{,}218$ \\ \hline
ME & Maine & $187{,}315$ & $422{,}998$ \\ \hline
CO & Colorado & $435{,}666$ & $1{,}057{,}066$ \\ \hline
FL & Florida & $1{,}070{,}376$ & $2{,}712{,}798$ \\ \hline
CA & California and Nevada & $1{,}890{,}815$ & $4{,}657{,}742$ \\ \hline
E-US & Eastern US & $3{,}598{,}623$ & $8{,}778{,}114$ \\ \hline
W-US & Western US & $6{,}262{,}104$ & $15{,}248{,}146$ \\ \hline
C-US & Central US & $14{,}081{,}816$ & $34{,}292{,}496$ \\ \hline
US & United States & $23{,}947{,}347$ & $58{,}333{,}344$ \\ \hline
\end{tabular}
\end{small}
\tbldown
\end{table}

Our first set of experiments focus on distance queries. Figure~\ref{fig:exp-dist-vary-q}a shows the average running time of each technique when answering the distance queries in $Q_i$ ($i \in [1, 10]$) on the DE road network (which contains $48{,}812$ nodes). Observe that AH consistently outperforms all competitors including CH, the state-of-the-art heuristic approach. In particular, on query sets $Q_8$, $Q_9$, and $Q_{10}$ (where each query concerns two distant locations), AH's running time is lower than that of CH and SILC by more than $50\%$. CH performs slightly worse than SILC on $Q_1, Q_2, \ldots, Q_6$, but it is evidently superior to SILC on $Q_8$, $Q_9$, and $Q_{10}$.
%This indicates that SILC is relatively inefficient in handling the queries where the source and destination of the query are far apart.
Dijkstra's algorithm incurs the highest computation overhead on all query sets.

Figure~\ref{fig:exp-dist-vary-q}b shows the query processing time of each method on NH, which is about $2$ times the size of DE. Again, AH is consistently more efficient than the other three techniques, especially on query sets $Q_8$, $Q_9$, and $Q_{10}$. CH suppresses SILC in most query sets, which contrasts the case on DE where CH only dominates SILC on $Q_8$, $Q_9$, and $Q_{10}$. This indicates that SILC does not scale as well as CH. Dijkstra's algorithm is still the least efficient one among the four techniques. Similar results are shown in Figure \ref{fig:exp-dist-vary-q}c and \ref{fig:exp-dist-vary-q}d.

Figures \ref{fig:exp-dist-vary-q}e - \ref{fig:exp-dist-vary-q}j show the running time of AH, CH, and Dijkstra's algorithm on the largest six datasets. (SILC is omitted since its preprocessing and space overheads on those these datasets are prohibitive, as will be shown in Section~\ref{sec:exp-pre}) The relative performance of AH, CH, and Dijkstra's algorithm remain the same as in Figures \ref{fig:exp-dist-vary-q}a - \ref{fig:exp-dist-vary-q}d, with AH (resp.\ Dijkstra's algorithm) being the most (resp.\ least) efficient method by far.

\subsection{Efficiency for Shortest Queries}\label{sec:exp-path}

\begin{figure*}[t]
\begin{small}
\begin{tabular}{cccc}
\multicolumn{4}{c}{\hspace{-4mm} \includegraphics[height=3.6mm]{./figures/Exp-Legends-2.eps}} \\
\hspace{-9.5mm}\includegraphics[height=37mm]{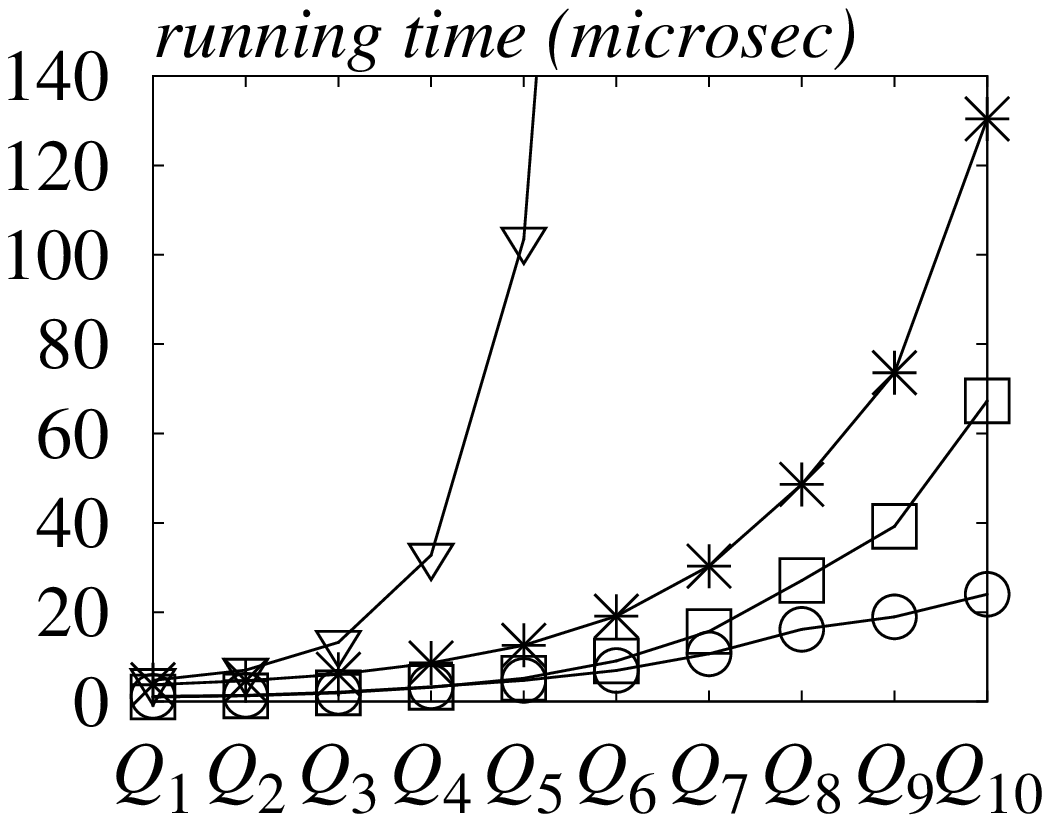}
&
\hspace{-12.5mm}\includegraphics[height=37mm]{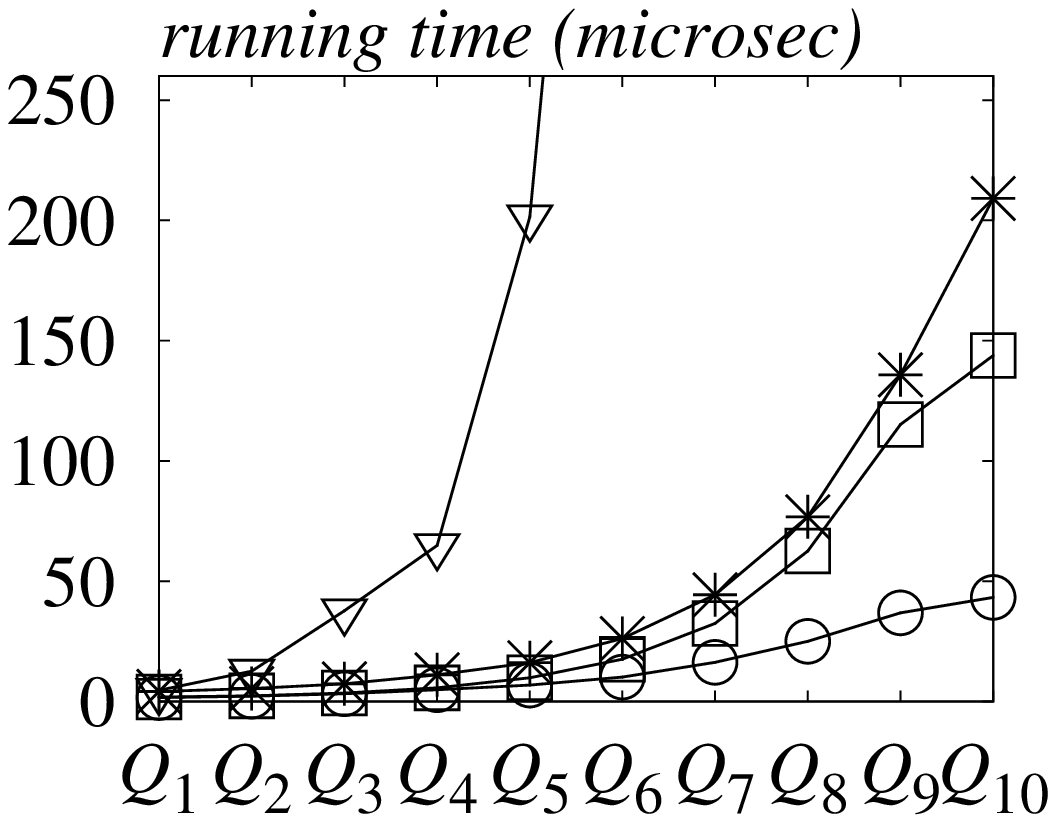}
&
\hspace{-12.5mm}\includegraphics[height=37mm]{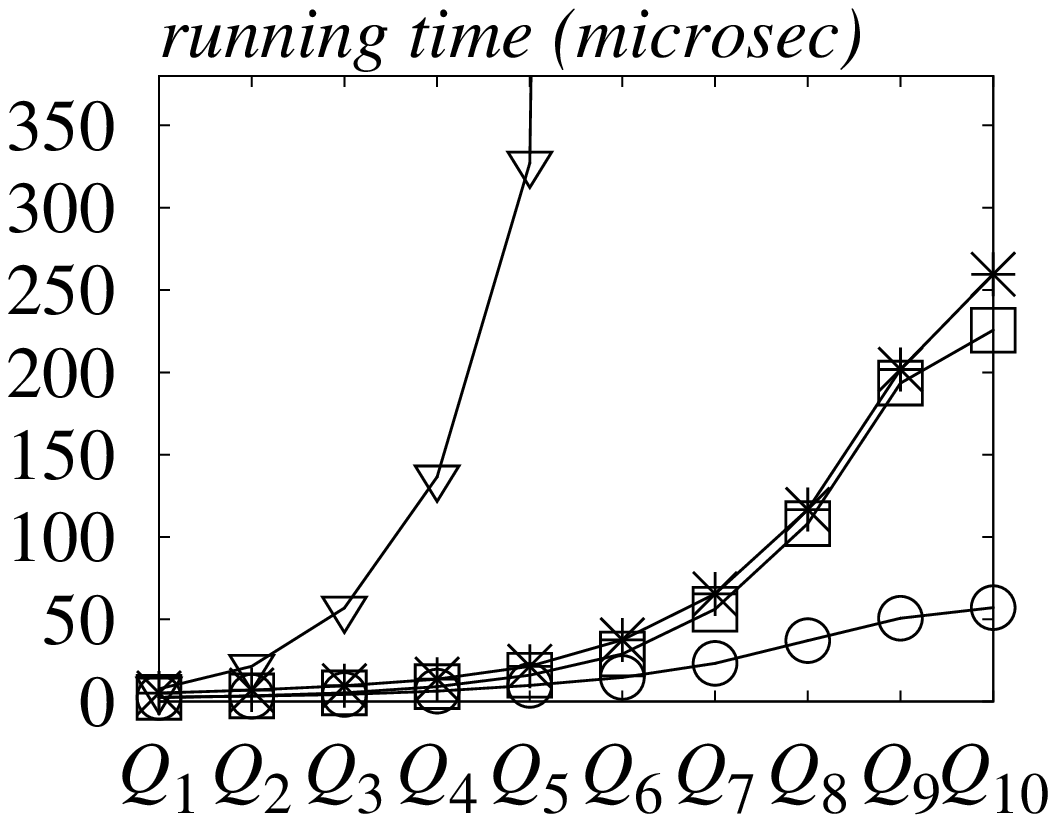}
&
\hspace{-12.5mm}\includegraphics[height=37mm]{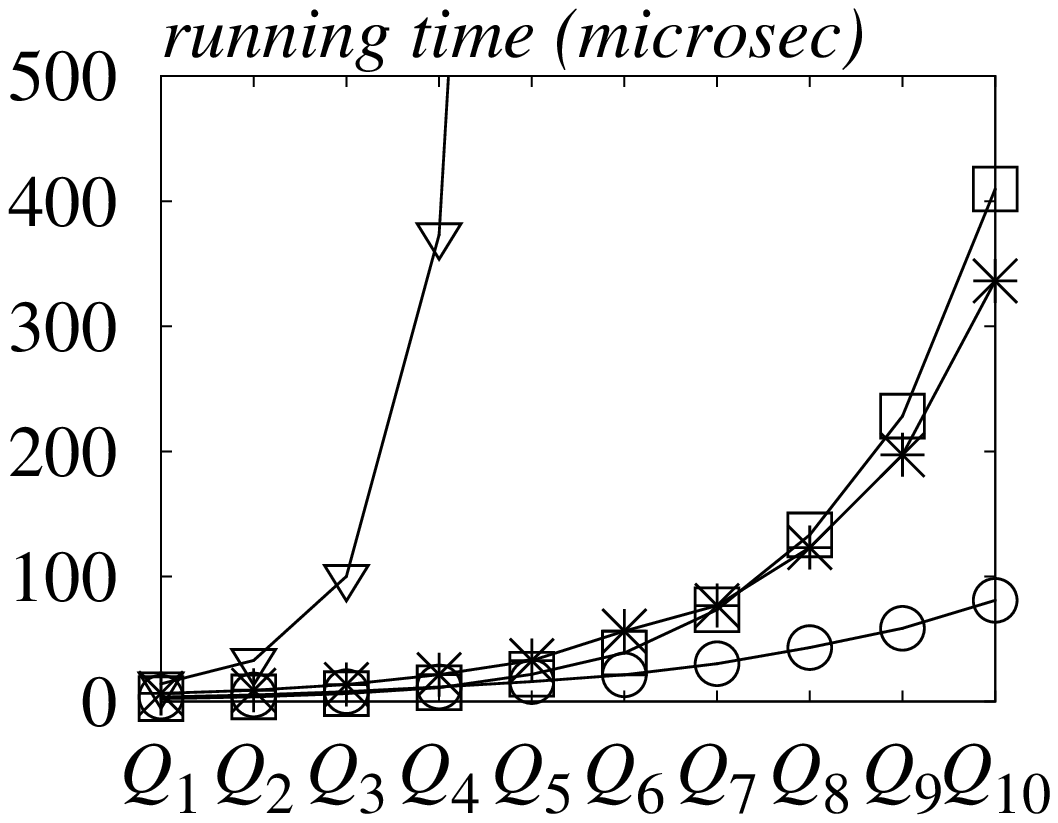} \\

\hspace{-6mm}(a) DE ($n = 48,812$) & \hspace{-6mm}(b) NH ($n = 115,055$) &
\hspace{-6mm}(c) ME ($n = 187,315$) & \hspace{-6mm} (d) CO ($n = 435,666$)\\ \vspace{-1.5mm} \\

\hspace{-9.5mm}\includegraphics[height=37mm]{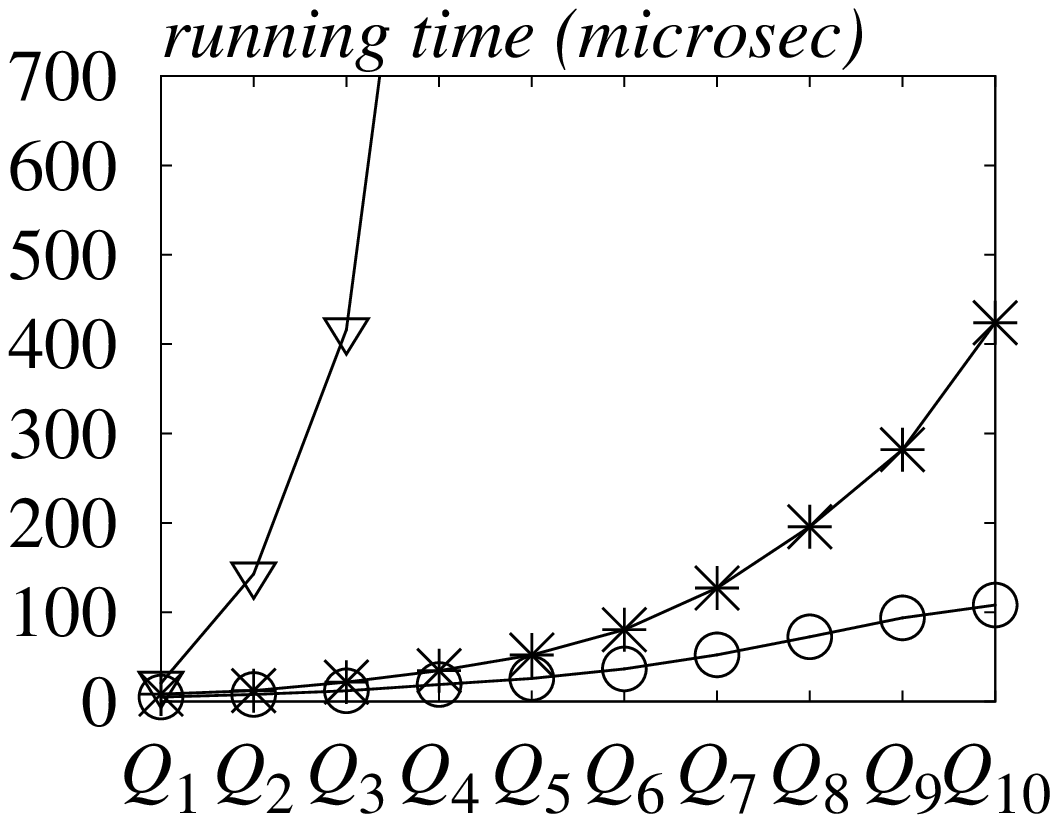}
&
\hspace{-12.5mm}\includegraphics[height=37mm]{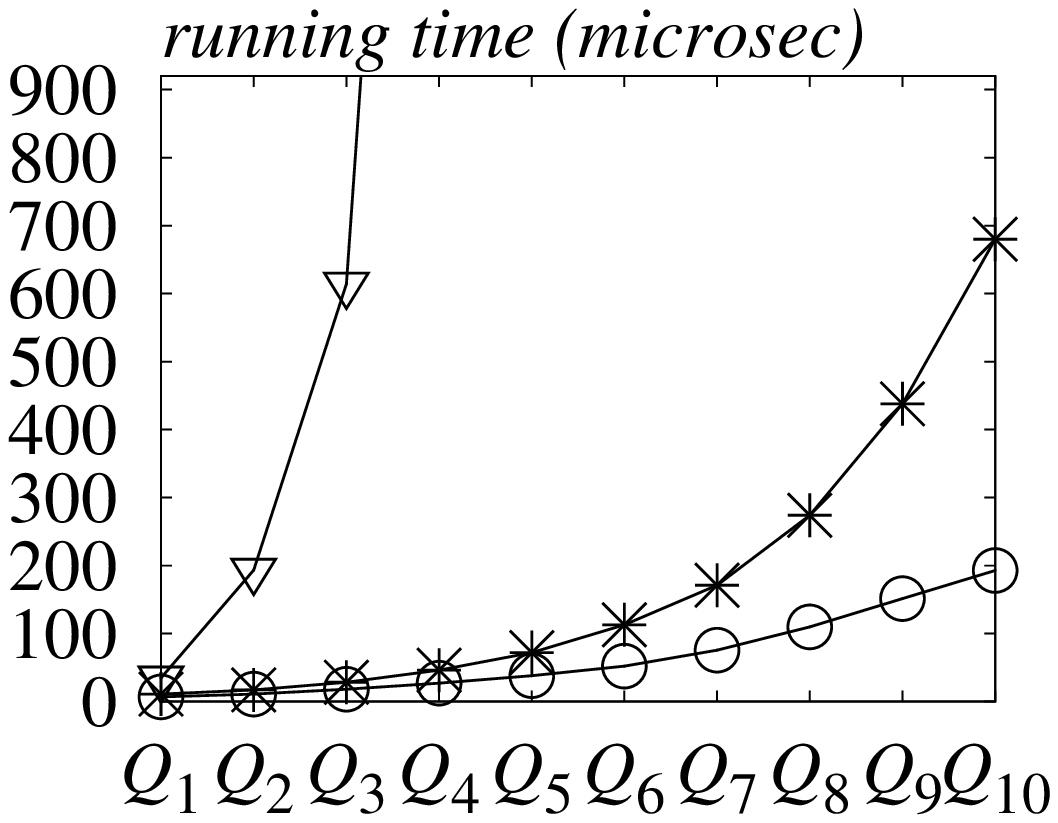}
&
\hspace{-12.5mm}\includegraphics[height=37mm]{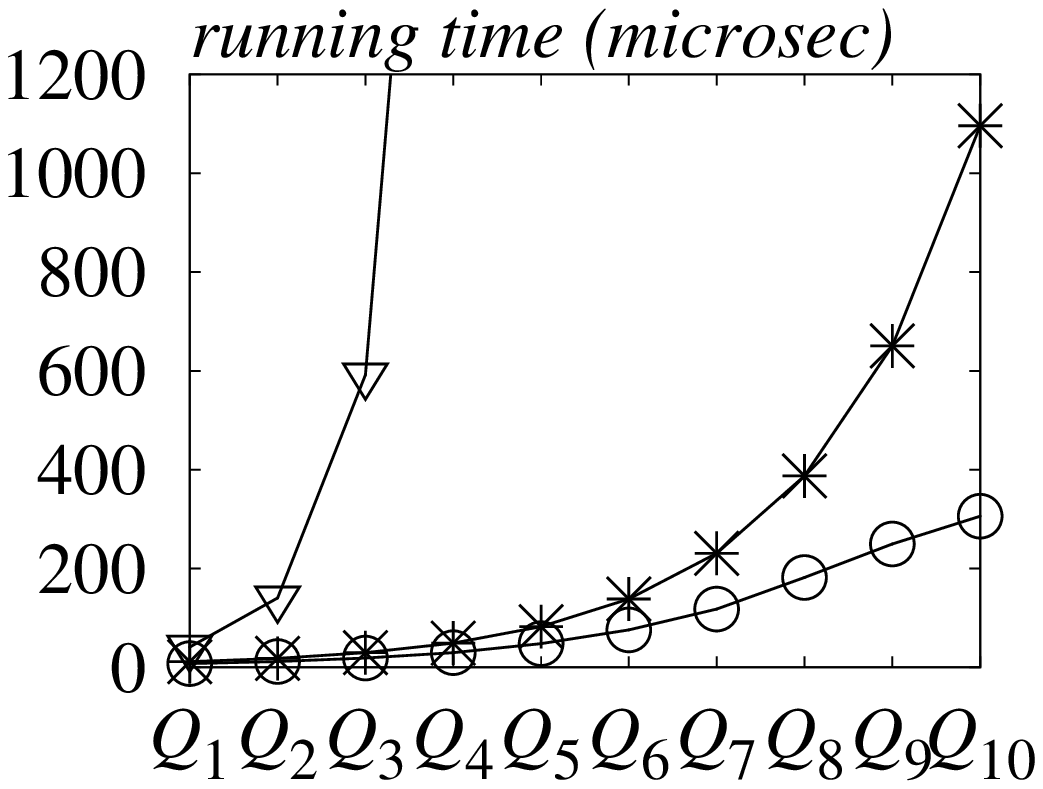}
&
\hspace{-12.5mm}\includegraphics[height=37mm]{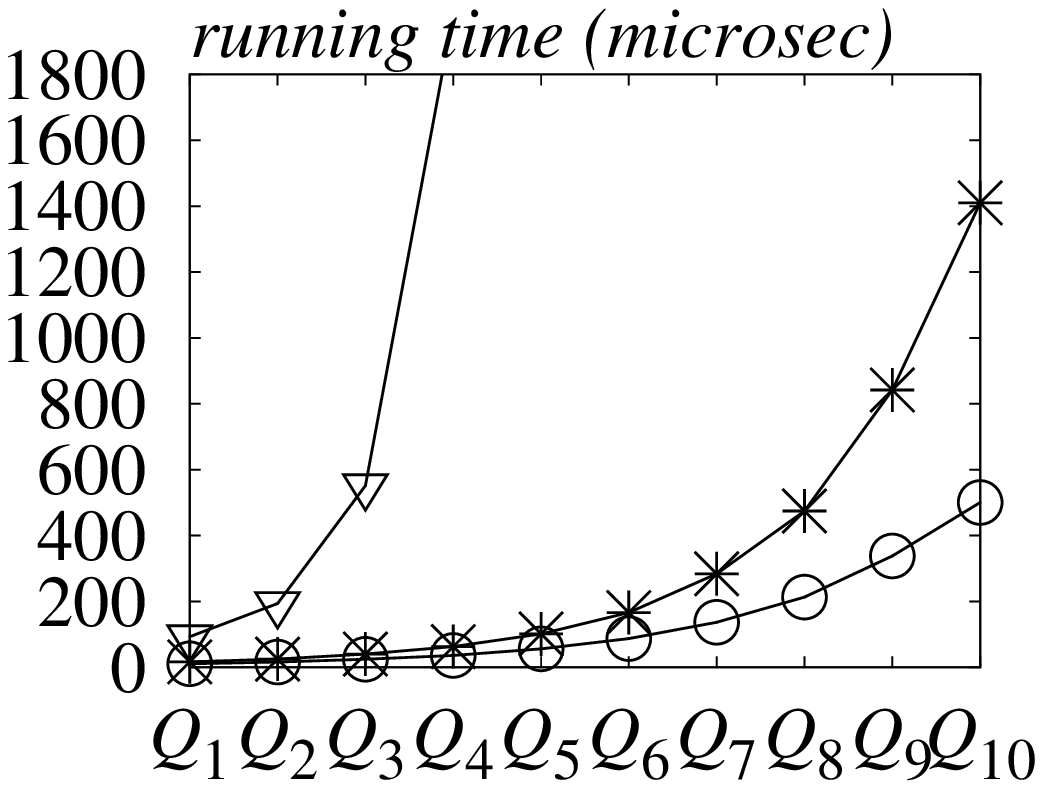} \\
\hspace{-6mm}(e) FL ($n = 1,070,376$) & \hspace{-6mm}(f) CA ($n = 1,890,815$) &
\hspace{-6mm}(g) E-US ($n = 3,598,623$) & \hspace{-6mm} (h) W-US ($n = 6,262,104$) \\ \vspace{-1.5mm} \\

\hspace{-9.5mm}
&
\hspace{-12.5mm}\includegraphics[height=37mm]{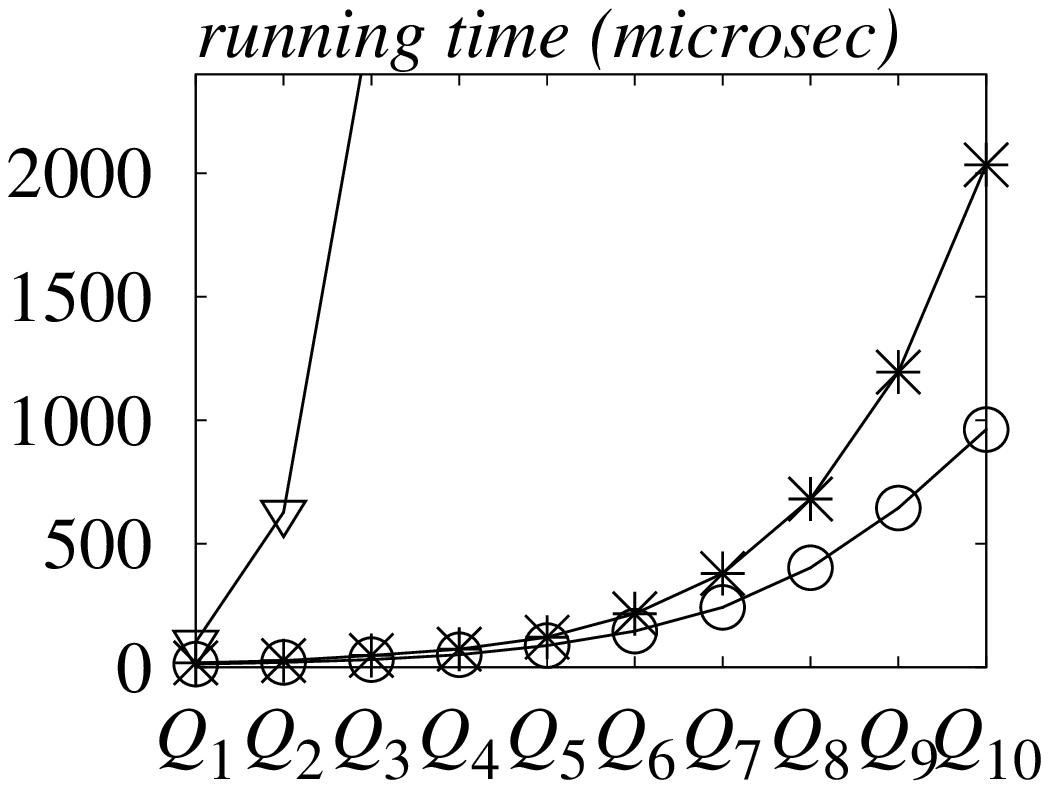}
&
\hspace{-12.5mm}\includegraphics[height=37mm]{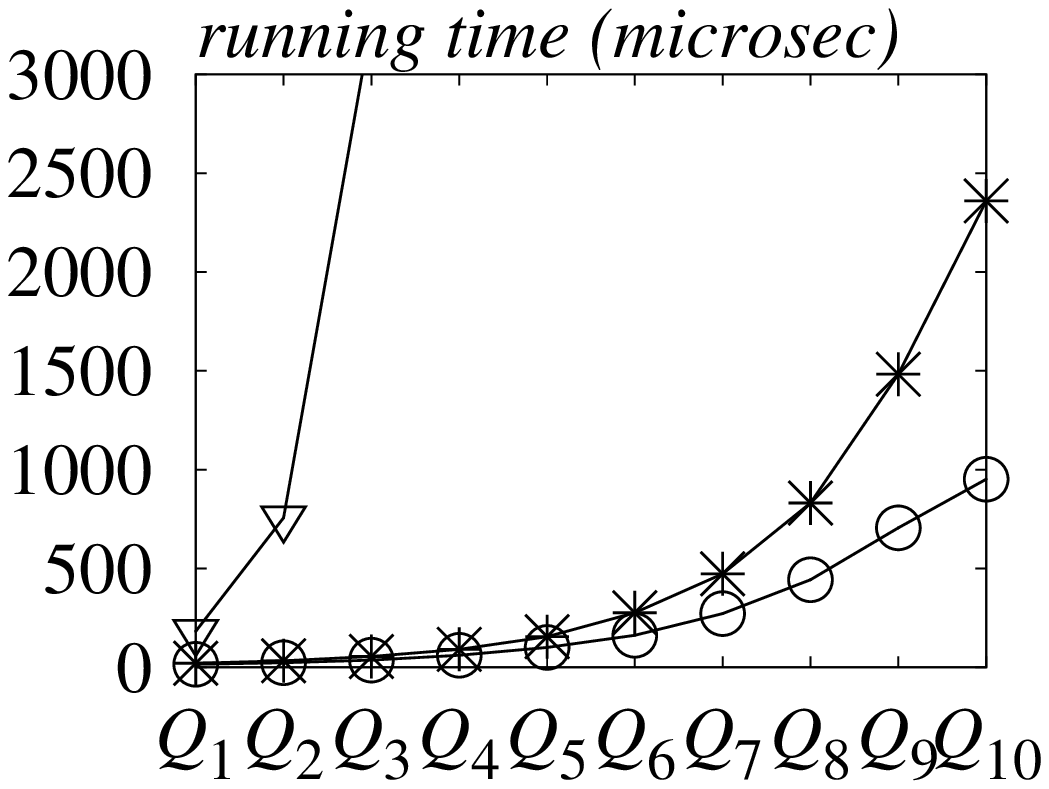}
&
\hspace{-12.5mm} \\
\hspace{-6mm} & \hspace{-6mm}(i) C-US ($n = 14,081,816$) &
\hspace{-6mm}(j) US ($n = 23,947,347$) & \hspace{-6mm}
\end{tabular}
\end{small}
\figcapup \caption{Efficiency of shortest path queries vs.\ query set.} %\vspace{-1mm} %\figcapdown
\label{fig:exp-path-vary-q}
\end{figure*}

Figure~\ref{fig:exp-dist-vary-q} shows the average computation time of each technique when answering the shortest path queries in $Q_i$ ($i \in [1, 10]$) on all ten datasets. Regardlsss of the dataset, AH significantly outperforms the other three techniques. SILC is superior to CH on DE, but the performance of the two methods becomes comparable on the larger datasets. Dijkstra's algorithm is the least efficient one in all cases.

The running time of AH is higher for shortest path queries than distance queries. This is because, when answering a shortest path query from a source $s$ to a destination $t$, AH first (i) computes the distance from $s$ to $t$, and then (ii) derives the shortest path based on the result of the distance query. As a consequence, any shortest path query incurs a strictly higher overhead than a distance query with the same source and destination. Similarly, CH also incurs a higher cost for shortest path queries than distance queries.

In contrast, the running time of SILC (resp.\ Dijkstra's algorithm) is identical in Figures \ref{fig:exp-dist-vary-q} and \ref{fig:exp-path-vary-q}. The reason is that, SILC (resp.\ Dijkstra's algorithm) answers any distance query by first deriving the shortest path $P$ from the source to the destination, and then returning the length of $P$. Computing the length of $P$ incurs only negligible overhead, which explains why the costs of distance queries are the same as that of the shortest path queries.

\subsection{Space and Preprocessing Costs}\label{sec:exp-pre}

In the last sets of experiments, we evaluate the space and pre-computation overheads of AH, SILC, and CH. (We omit Dijkstra's algorithm as it does not require building an index on the road network.) Figure~\ref{fig:exp-pre}a illustrates the index space required by AH, SILC, and CH on each dataset. Although SILC is worst-case efficient, its space overhead is extremely high, and it increases super-linearly with the number of nodes $n$ in the road network. In particular, for all datasets with more than $500,000$ nodes, the index of SILC is more than $32$GB in size, i.e., it cannot fit in the main memory of our machine. For this reason, we omit SILC from the experimental on those datasets. Meanwhile, the space consumption of AH is moderate, and it increases linearly with $n$. This is consistent with our analysis in Section~\ref{sec:ah-complexity} that AH incurs a linear space complexity. Lastly, CH is the most space-economic method: it requires no more than $2$GB of space even for the largest dataset.

Figure~\ref{fig:exp-pre}b shows the time required by AH, SILC, and CH to construct indices on our datasets. Observe that SILC has a pre-computation cost super-linear to $n$, and it requires more than one week to preprocess even the relatively small dataset CO (which contains $435,666$ nodes). In contrast, the preprocessing time of AH exhibits a linear increase with $n$, even though AH's index construction algorithm has an $O(hn^2)$ time complexity. Furthermore, the pre-computation cost of AH is fairly small: it only requires around three hours to preprocess the US road network with $23$ million nodes. On the other hand, the pre-computation time of CH is minimum and is below $40$ minutes for all datasets.
%\subsection{Summary}\label{sec:exp-summary}

\section{Conclusion} \label{sec:conclusion}

This paper presents Arterial Hierarchy (AH), a worst-case efficient index structure for shortest path and distance queries on road networks. Under a practical assumption about the road network, AH offers superior query time complexities in both shortest path and distance queries, and its space and preprocessing time complexities are comparable to the best existing worst-case efficient methods. With extensive experiments on real datasets, we show that AH also provides excellent query efficiency in practice, and it even outperforms CH (i.e., the state-of-the-art heuristic method) in terms of query time. Furthermore, the space consumption and preprocessing cost of AH are fairly small: It takes only around three hours to preprocess a continent-scale road network with $23$ million nodes, and the resulting index structure is no more than $32$GB in size. For future work, we plan to extend AH for the scenarios when (i) the weight of each road network edge may change with time (e.g., due to traffic conditions) and (ii) the memory footprint of the index structure is a significant concern (as is the case for mobile devices).

\begin{figure}[h]
\begin{small}
\begin{tabular}{cc}
\multicolumn{2}{c}{\hspace{-4mm} \includegraphics[height=3.3mm]{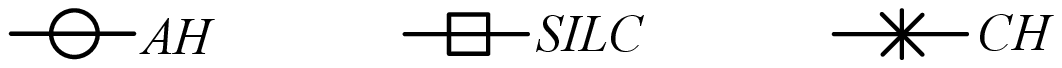}} \\
\hspace{-9mm} \includegraphics[height=39mm]{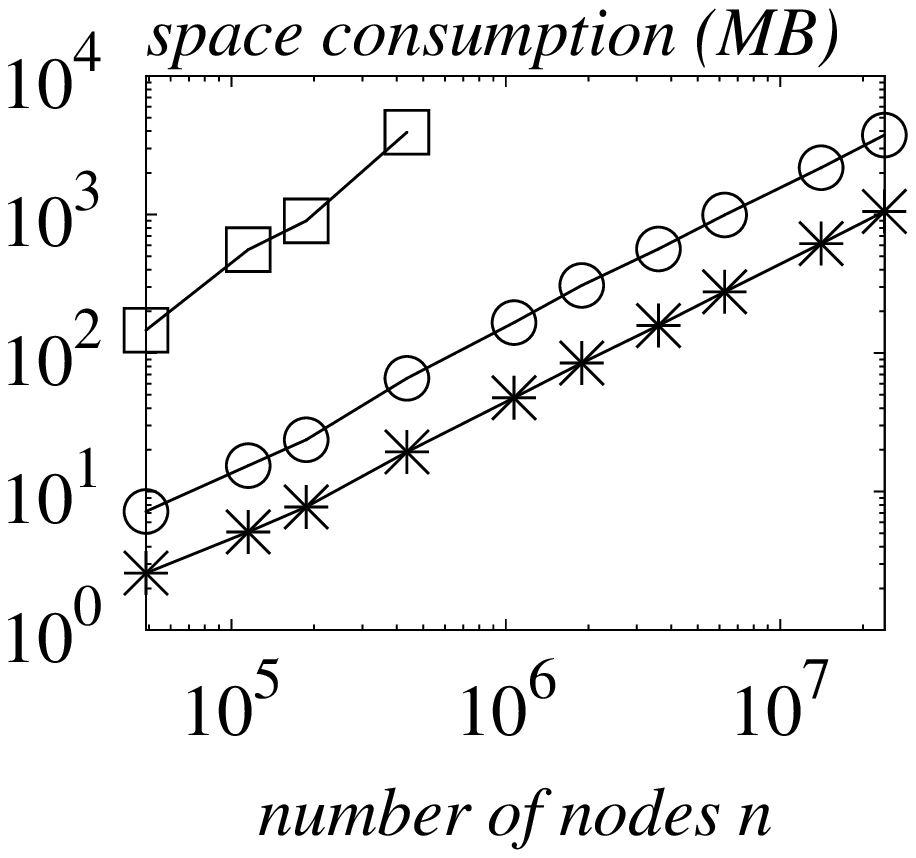} &
\hspace{-14mm} \includegraphics[height=39mm]{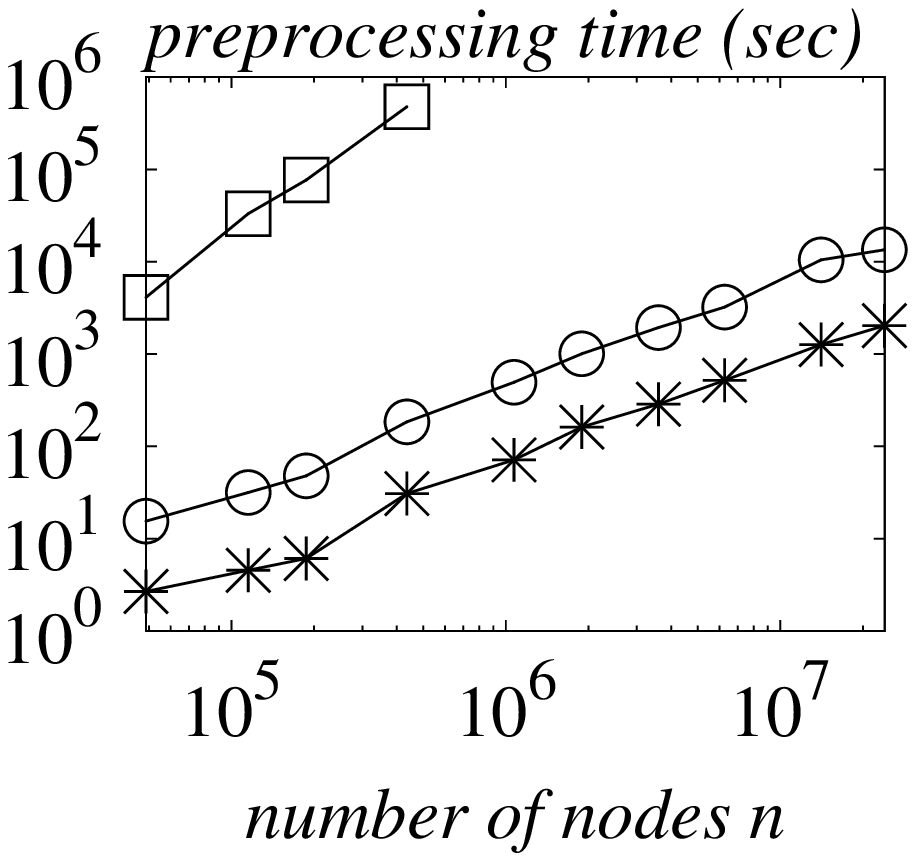} \\
\hspace{-6mm}(a) Space Overhead vs. $n$ & \hspace{-12mm}(b) Preprocessing Time vs. $n$
\end{tabular}
\end{small}
\figcapup  \caption{Space Overhead and Preprocessing Time vs.\ $\boldsymbol{n}$} \figcapdown
\label{fig:exp-pre}
\end{figure}

\bibliographystyle{abbrv}
%\begin{small}
\bibliography{ref}

\begin{thebibliography}{10}

\bibitem{SILCcode}
\url{https://sourceforge.net/projects/ntu-sp-exp/}.

\bibitem{CHcode}
\url{http://algo2.iti.kit.edu/routeplanning.php}.

\bibitem{dimacs}
\url{http://www.dis.uniroma1.it/~challenge9/}.

\bibitem{afg10}
I.~Abraham, A.~Fiat, A.~V. Goldberg, and R.~F.~F. Werneck.
\newblock Highway dimension, shortest paths, and provably efficient algorithms.
\newblock In {\em SODA}, pages 782--793, 2010.

\bibitem{bfm06}
H.~Bast, S.~Funke, and D.~Matijevic.
\newblock Transit: ultrafast shortest-path queries with linear-time
  preprocessing.
\newblock In {\em Proc. of the 9th DIMACS Implementation Challenge}, pages
  175--192, 2006.

\bibitem{bfs07}
H.~Bast, S.~Funke, P.~Sanders, and D.~Schultes.
\newblock Fast routing in road networks with transit nodes.
\newblock {\em Science}, 316(5824):566, 2007.

\bibitem{IntroAlgo}
T.~H. Cormen, C.~Stein, R.~L. Rivest, and C.~E. Leiserson.
\newblock {\em Introduction to Algorithms}.
\newblock McGraw-Hill Higher Education, 2nd edition, 2001.

\bibitem{dss09}
D.~Delling, P.~Sanders, D.~Schultes, and D.~Wagner.
\newblock Engineering route planning algorithms.
\newblock In {\em Algorithmics of Large and Complex Networks}, pages 117--139,
  2009.

\bibitem{d59}
E.~W. Dijkstra.
\newblock A note on two problems in connection with graphs.
\newblock {\em Numerical Mathematics}, 1:269--271, 1959.

\bibitem{fr06}
J.~Fakcharoenphol and S.~Rao.
\newblock Planar graphs, negative weight edges, shortest paths, and near linear
  time.
\newblock {\em J. Comput. Syst. Sci.}, 72(5):868--889, 2006.

\bibitem{gss08}
R.~Geisberger, P.~Sanders, D.~Schultes, and D.~Delling.
\newblock Contraction hierarchies: Faster and simpler hierarchical routing in
  road networks.
\newblock In {\em WEA}, pages 319--333, 2008.

\bibitem{gh05}
A.~V. Goldberg and C.~Harrelson.
\newblock Computing the shortest path: {A$^*$} search meets graph theory.
\newblock In {\em SODA}, pages 156--165, 2005.

\bibitem{gkw06}
A.~V. Goldberg, H.~Kaplan, and R.~F. Werneck.
\newblock Reach for {A$^{*}$}: Efficient point-to-point shortest path
  algorithms.
\newblock In {\em ALENEX}, pages 129--143, 2006.

\bibitem{gkr04}
S.~Gupta, S.~Kopparty, and C.~V. Ravishankar.
\newblock Roads, codes and spatiotemporal queries.
\newblock In {\em PODS}, pages 115--124, 2004.

\bibitem{hkm06}
M.~Hilger, E.~K\"{o}hler, R.~H. M\"{o}hring, and H.~Schilling.
\newblock Fast point-to-point shortest path computations with arc-flags.
\newblock In {\em Proc. of the 9th DIMACS Implementation Challenge}, pages
  73--92, 2006.

\bibitem{jhr98}
N.~Jing, Y.-W. Huang, and E.~A. Rundensteiner.
\newblock Hierarchical encoded path views for path query processing: An optimal
  model and its performance evaluation.
\newblock {\em TKDE}, 10(3):409--432, 1998.

\bibitem{jp02}
S.~Jung and S.~Pramanik.
\newblock An efficient path computation model for hierarchically structured
  topographical road maps.
\newblock {\em TKDE}, 14(5):1029--1046, 2002.

\bibitem{kmw10}
P.~N. Klein, S.~Mozes, and O.~Weimann.
\newblock Shortest paths in directed planar graphs with negative lengths: A
  linear-space ${O}(n log^2n)$-time algorithm.
\newblock {\em ACM Transactions on Algorithms}, 6(2), 2010.

\bibitem{ms12}
S.~Mozes and C.~Sommer.
\newblock Exact distance oracles for planar graphs.
\newblock In {\em SODA}, pages 209--222, 2012.

\bibitem{rt10}
M.~Rice and V.~J. Tsotras.
\newblock Graph indexing of road networks for shortest pth queries with label
  restrictions.
\newblock {\em PVLDB}, 4:69--80, 2010.

\bibitem{ssa08}
H.~Samet, J.~Sankaranarayanan, and H.~Alborzi.
\newblock Scalable network distance browsing in spatial databases.
\newblock In {\em SIGMOD}, pages 43--54, 2008.

\bibitem{ss10}
J.~Sankaranarayanan and H.~Samet.
\newblock Query processing using distance oracles for spatial networks.
\newblock {\em IEEE Trans. Knowl. Data Eng.}, 22(8):1158--1175, 2010.

\bibitem{ssa09}
J.~Sankaranarayanan, H.~Samet, and H.~Alborzi.
\newblock Path oracles for spatial networks.
\newblock {\em PVLDB}, 2:1210--1221, 2009.

\bibitem{tsp11}
Y.~Tao, C.~Sheng, and J.~Pei.
\newblock On k-skip shortest paths.
\newblock In {\em SIGMOD}, pages 421--432, 2011.

\bibitem{wxd12}
L.~Wu, X.~Xiao, D.~Deng, G.~Cong, A.~D. Zhu, and S.~Zhou.
\newblock Shortest path and distance queries on road networks: An experimental
  evaluation.
\newblock {\em PVLDB}, 5(5):406--417, 2012.

\end{thebibliography}
%\end{small}

\appendix
%\Appendix
%\section{Appendix} \label{sec:appendix}

\section{Unique Shortest Paths via Weight Perturbation} \label{sec:perturb}

Let $h$ be as defined in Section~\ref{sec:fc-pre}. The solutions in the paper rely on the following assumption:
\begin{assumption} \label{assu:perturb}
For any ($4{\times}4$)-cell region $B$ in the square grid $R_i$ ($i \in [0, h]$), there do not exist two local shortest paths in $B$ that share the same endpoints and have the same length.
\end{assumption}
In the section, we show that Assumption~\ref{assu:perturb} can be enforced by adding a small perturbation to the weight of each edge in the road network $G$. Specifically, we associate each edge $e$ in $G$ with an integer $\rho(e)$ that is randomly selected in the range $[0, \tau-1]$, where $\tau$ is a parameter to be specified shortly. We refer to $\rho(e)$ as the {\em nuance} of $e$, and we define the nuance of a path $P$ as the sum of the nuance of the edges on the path, denoted as $\rho(P)$. For any two path $P_1$ and $P_2$ such that $l(P_1) = l(P_2)$, we consider $P_1$ shorter than $P_2$ if $\rho(P_1) < \rho(P_2)$. We will establish the following theorem.
\begin{theorem} \label{thrm:perturb}
Let $\Delta$ be the largest degree of any node in $G$. If $\tau \ge 32hn^3{\Delta\choose 2}$, then Assumption~\ref{assu:perturb} holds with a probability at least $1 - \frac{1}{n}$.
\end{theorem}
In other words, by setting $\tau$ to a sufficiently large value, we can ensure that Assumption~\ref{assu:perturb} holds with an overwhelming probability.

Our proof of Theorem~\ref{thrm:perturb} is based on a few lemmas as follows.

\begin{lemma} \label{lemma:path_equal_pr}
Let $P$ and $P'$ be two paths in $G$. Then, $\rho(P) = \rho(P')$ occurs with at most $1/\tau$ probability.
\end{lemma}
\begin{proof}
Assume that $P=\langle e_1, \ldots, e_i \rangle$ and $P'=\langle e'_1, \ldots, e'_j\rangle$. Then,
\begin{eqnarray*}
\lefteqn{Pr\left\{\rho(P) = \rho(P')\right\}} \\
& = & Pr\left\{\sum_{1 \le k \le i}\rho(e_k) = \sum_{1 \le k \le j}\rho(e'_k)\right\} \\
& = & Pr\left\{\rho(e_1) = \sum_{1 \le k \le j}\rho(e'_k) - \sum_{2 \le k \le i}\rho(e_k) \right\} \\
& = & \sum_{0\le x \le \tau-1} \Bigg(Pr\left\{\rho(e_1) = x\right\} \\
&  & \qquad \cdot \left.Pr\left\{\sum_{1 \le k \le j}\rho(e'_j) - \sum_{2 \le k \le i}\rho(e_i) = x \right\}\right) \\
& = & \sum_{0\le x \le \tau-1} \left(\frac{1}{\tau} \cdot Pr\left\{\sum_{1 \le k \le j}\rho(e'_j) - \sum_{2 \le k \le i}\rho(e_i) = x \right\}\right) \\
& = & \frac{1}{\tau} \cdot \sum_{0\le x \le \tau-1} Pr\left\{\sum_{1 \le k \le j}\rho(e'_j) - \sum_{2 \le k \le i}\rho(e_i) = x \right\} \\
& = & \frac{1}{\tau}
\end{eqnarray*}
\end{proof}
%Based on Lemma~\ref{lemma:path_equal_pr}, we can have the following result:
%{\bf Define what we mean by "Unique" shortest path.}
Let $\Delta$ be the maximum degree of any node in $G$. Based on Lemma~\ref{lemma:path_equal_pr}, we have the following result:
\begin{lemma} \label{lemma:fixnode_equal_pr}
Let $B$ be a ($4{\times}4$)-cell region in $R_i$ ($i \in [0, h]$). For a node $s$ in $B$, let $\zeta$ (resp.\ $\zeta'$) be the event that there exists a another node $v$, such that the local shortest path from $s$ to $v$ (resp.\ from $v$ to $s$) in $B$ is not unique. Then, $Pr\{\zeta \vee \zeta'\} \le {\Delta \choose 2}\cdot 2n / \tau$.
\end{lemma}
\begin{proof}
We will prove that $Pr\{\zeta\} \le {\Delta \choose 2}\cdot n / \tau$. By symmetry, it can also be shown that $Pr\{\zeta'\} \le {\Delta \choose 2}\cdot n / \tau$, leading to $\left.Pr\{\zeta \vee \zeta'\} \le {\Delta \choose 2}\cdot 2n / \tau\right.$.

Let $d_s(v)$ be length of the the local shortest path distance from $s$ to $v$ in $B$. Let $\langle v_0, v_1, \ldots, v_{k}\rangle$ be a permutation of all the nodes that can be reached from $s$ via local paths in $B$, such that $d_s(v_i)\leq d_s(v_j)$ for any $0\leq i<j\leq k$. That is, $v_0,v_1,\ldots,v_{k}$ are sorted in a non-decreasing order of their distances from $s$. Note that $v_0 = s$.
Let $\zeta_{i}$ ($i \in [1, k]$) be the event that (i) the local shortest path from $s$ to any $v_j$ ($j \in [0, i-1]$) in $B$ is unique, but (ii) the local shortest path from $s$ to $v_i$ is not unique. We have $Pr\{\zeta_1\} = 0$; otherwise, there must exist another node $u$ such that $d_s(u) \le d_s(v_1)$, contradicts the definition of $v_i$ ($i \in [0, k]$). In addition,
\begin{eqnarray*}
Pr\left\{\zeta\right\} &=& Pr\bigg\{\bigcup_{1\leq i\leq k}\zeta_i\bigg\}\\
&\le& Pr\{\zeta_1\} + Pr\bigg\{\bigcup_{2\leq i\leq k}\zeta_i\bigg\}\\
&=& Pr\bigg\{\bigcup_{2\leq i\leq k}\zeta^i\bigg\}.
\end{eqnarray*}

Now let us consider $Pr\{\zeta^i\}$ for $i \in [2, n-1]$. Let $\mathcal{P}_{s,v_i}=\{P_1,\ldots,P_q\}$ be the set of local shortest paths from $s$ to $v_i$ in $B$. For each $P_j$ ($j \in [1, q]$), let $\langle u_j,v_i \rangle$ be the last edge on $P$. Then, $u_j$ should be in $\{v_0,v_1,\ldots, v_{i-1}\}$. By the definition of $\zeta_i$, the local shortest path from $s$ to $u_j$ in $B$ is unique. Furthermore, $q \le \Delta$, since the degree of $v_i$ is at most $\Delta$. By Lemma~\ref{lemma:path_equal_pr}, we have
\begin{eqnarray*} \nonumber
Pr(\zeta^i) &\leq & \sum_{1\leq j<k\leq q} Pr\left\{\rho(P_j)=\rho(P_k)\right\} \\
&\leq &\sum_{1\leq j<k\leq q} \frac{1}{m} \quad = \; \frac{1}{m}{q \choose 2} \quad \leq \;  \frac{1}{m}{\Delta \choose 2}.
\end{eqnarray*}
Therefore,
\begin{eqnarray*}
Pr\left\{\bigcup_{2\leq i\leq n-1}\zeta_i\right\}
& \leq & \sum_{2\leq i\leq n-1}Pr\left\{\zeta_i\right\} \quad \leq \quad \frac{n}{\tau}{\Delta \choose 2},
\end{eqnarray*}
which completes the proof.
\end{proof}

Given Lemma~\ref{lemma:fixnode_equal_pr}, we prove Theorem~\ref{thrm:perturb} as follows:
\begin{proof}[Of Theorem~\ref{thrm:perturb}]
Let $s$ be an arbitrary node in $G$. For any $R_i$ ($i \in [0, h]$), there exist at most $16$ ($4{\times}4$)-cell regions in $R_i$ that contains $s$. By Lemma~\ref{lemma:fixnode_equal_pr}, for each $B$ of those $16$ regions, there is at most ${\Delta \choose 2}\cdot n / \tau$ probability that $B$ contains non-unique local shortest paths between $s$ and and another node. Taking in account all possible choices of $B$ in all $R_i$ and all possible choices of $s$, Assumption~\ref{assu:perturb} fails with a probability at most
\begin{eqnarray*}
{\Delta \choose 2}\cdot 32 n / \tau \cdot h \cdot n &=& {\Delta \choose 2}\cdot 32 n^2 h / \tau.
\end{eqnarray*}
By setting $\tau \ge 32 n^3 h {\Delta \choose 2} $, we can guarantee that the above probability is at most $1/n$. Therefore, Theorem~\ref{thrm:perturb} is proved.
\end{proof}

\begin{figure*}[t]
%\begin{small}
\centering
\begin{minipage}[t]{3 in}
\centering
\includegraphics[width = 2in]{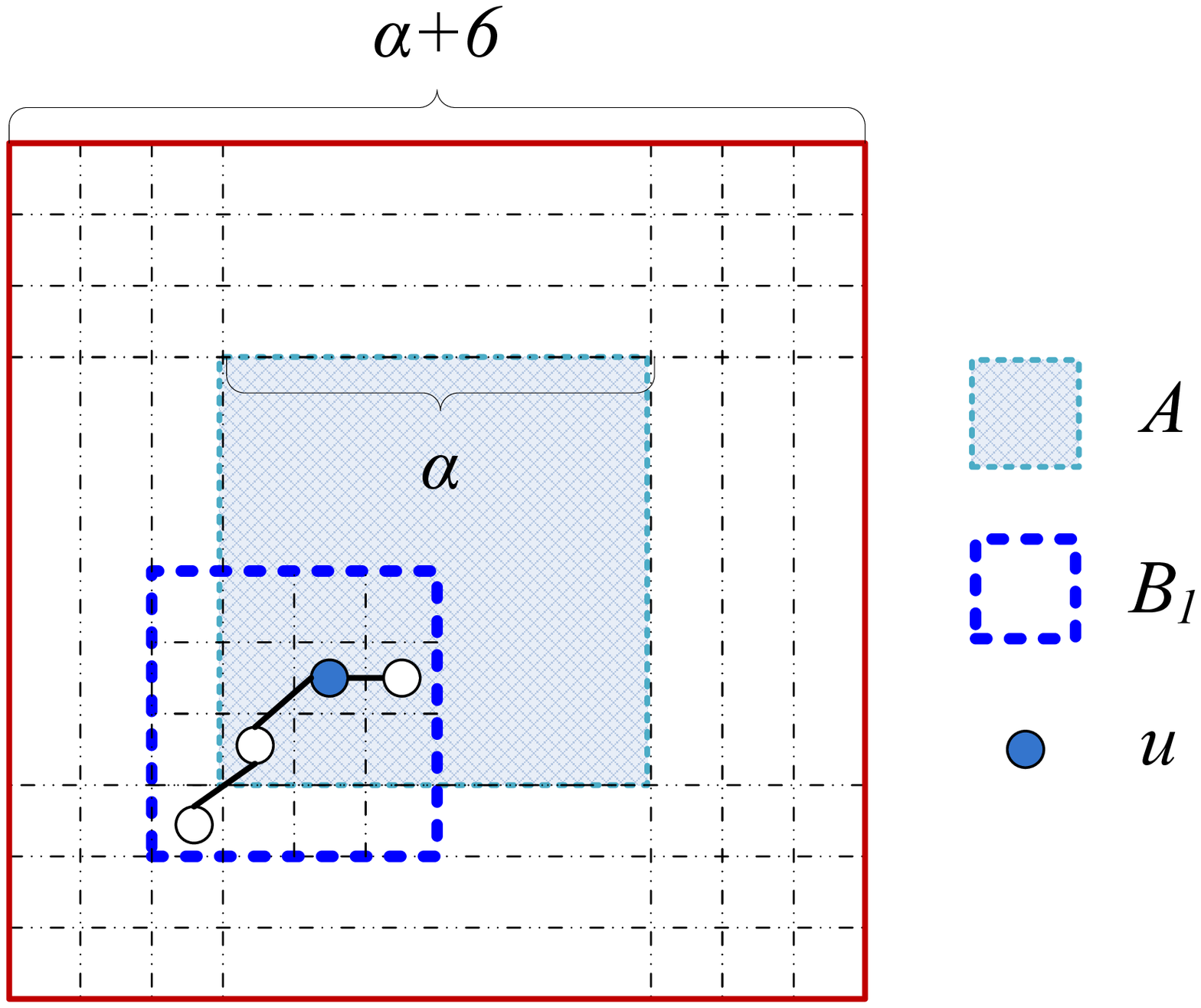}
%\vspace{-3.52mm}
\figcapup \caption{$B_1$ is completely contained in the ($(\alpha{+}6){\times} (\alpha{+}6)$)-cell region.} %\figcapdown
\label{fig:density_1}
\end{minipage}
\hspace{10mm}
\begin{minipage}[t]{3 in}
\centering
\includegraphics[width = 2in]{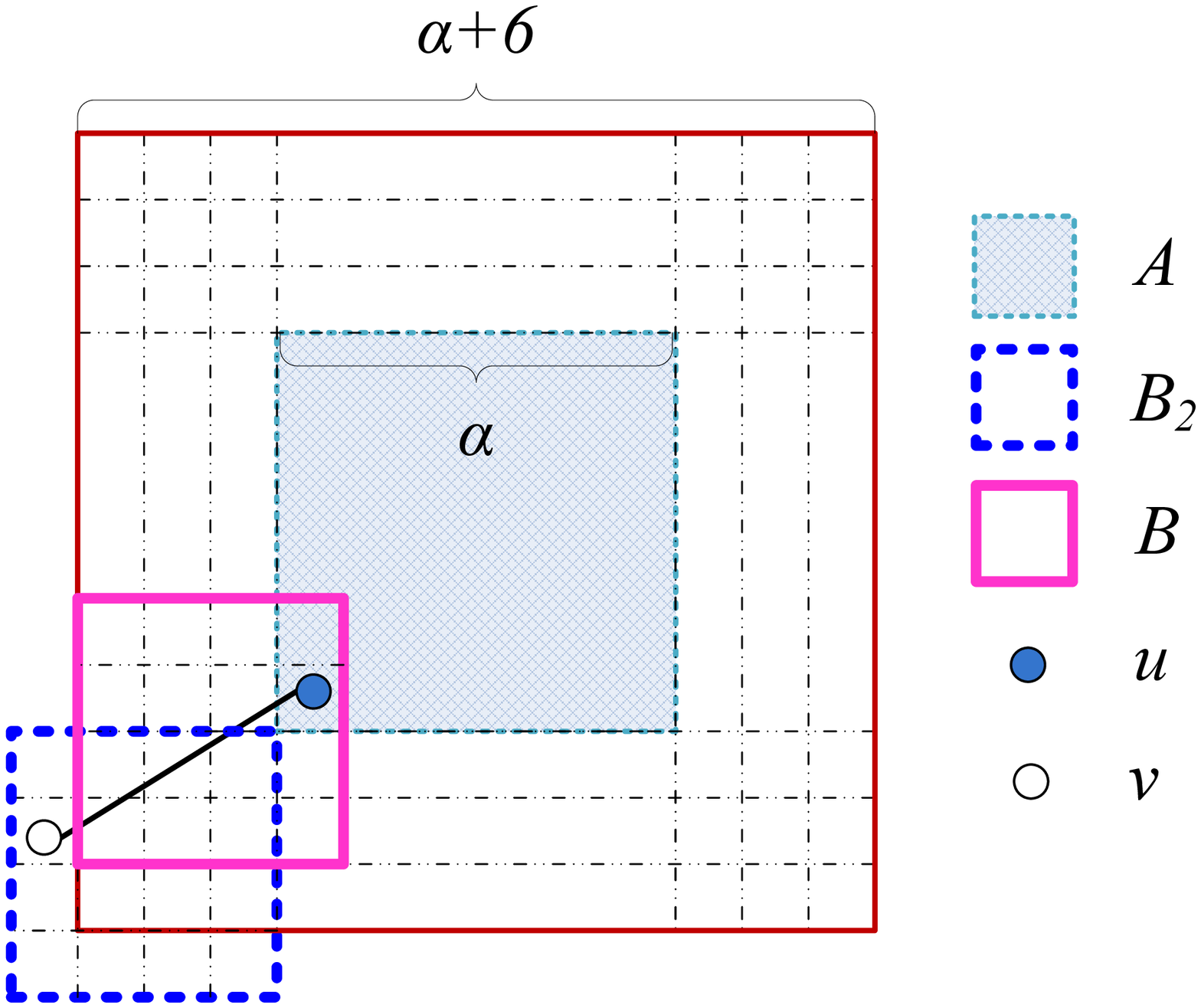}
\figcapup \caption{$B_2$ is not completely contained in the ($(\alpha{+}6){\times} (\alpha{+}6)$)-cell region.} %\figcapdown
\label{fig:density_2}
\end{minipage}
%\end{small}
%\vspace{-1mm}
\end{figure*}

\header
{\bf Remark.}
The above perturbation method requires generating random numbers in the integer range of $[0, \tau-1]$, which causes practical concerns since $\tau-1$ can be too large to be represented with a normal integer. We address this issue by using multiple random integers in a relatively narrow range to represent $\tau$. In particular, to generate the nuance for an edge, we can use $k$ random integers in the range of $[0,\tau'-1]$, where $\tau'=\tau^\frac{1}{k}$. Accordingly, the nuance on each edge would be a $k$-dimensional vector. It can be verified that, under such edge perturbation, the results in this section still hold.

\section{The SlidingWindow Algorithm} \label{sec:slidingwindow}

This section presents an algorithm called {\em SlidingWindow}, which will be used in proving the key lemmas in the following sections. Let $P$ be a shortest path in $G$, such that no ($3 {\times} 3$)-cell region in $R_i$ ($i \in [1, h]$) can cover all nodes in $P$ simultaneously. Given $P$ and $R_i$, the {\em SlidingWindow} algorithm identifies a ($4 {\times} 4$)-cell region $B$ in $R_i$, such that $B$ has a spanning path $P'$ that is a sub-path of $P$. Figure~\ref{code:slidingwindow} shows the pseudo-code of {\em SlidingWindow}.

\begin{figure}
\begin{small}
\begin{tabbing}
\line(1,0){240}\\
\hspace{4mm} \= \hspace{4mm} \= \hspace{4mm} \= \hspace{4mm} \=
\hspace{4mm} \= \hspace{4mm} \= \hspace{4mm} \= \hspace{4mm} \=
\hspace{4mm} \kill
{\bf Algorithm {\em SlidingWindow}} ( $P=\langle v_1,v_2,\ldots,v_k\rangle$, $R_i$ ) \\[1mm]
1.\> Let $\mathcal{B}$ be a set that contains any region $B$ consisting of a consecutive \\
\> block of $x\times y$ cells in $R_i$ ($x, y>0$).\\
2.\> Initialize $\theta=0$.\\
3.\> For $i=1$ to $k$\\
4.\>\> $\theta = j$. \\
5.\>\> Let $B_j$ be the smallest region in $\mathcal{B}$ that covers $v_1,v_2,\ldots,v_j$ \\
\>\> simultaneously. \\
6.\>\> if $B_j$ is at least $4$ cells in width or height, then break. \\
7.\> If $B_\theta$ is at least $4$ cells in width \\
8.\>\> Let $v_\alpha$ and $v_\beta$ be the nodes in $\{v_1, v_2, \ldots, v_\theta\}$ with the smallest \\
\>\> and largest x-coordinates, respectively.\\
9.\>\> Let $B$ be any ($4\times 4$)-cell region in $R_i$ such that (i) $B$ covers $v_1$\\
\>\> $v_2, \ldots, v_{\theta-1}$ simultaneously, and (ii) $v_\alpha$ is in the west strip of $B$.\\
10.\> Else \\
11.\>\> Let $v_\alpha$ and $v_\beta$ be the nodes in $\{v_1, v_2, \ldots, v_\theta\}$ with the smallest \\
\>\> and largest y-coordinates, respectively.\\
12.\>\> Let $B$ be any ($4\times 4$)-cell region in $R_i$ such that (i) $B$ covers $v_1$\\
\>\> $v_2, \ldots, v_{\theta-1}$ simultaneously, and (ii) $v_\alpha$ is in the south strip of $B$.\\
13.\> Let $a=\min\{\alpha,\beta\}$, and $b=\max\{\alpha,\beta\}$. \\
14.\> Return $B$ and $P'=\langle v_a, v_{a+1}, \ldots, v_b \rangle$.\\
\line(1,0){240}
\end{tabbing}
\end{small}
\algocapup
\caption{The {\em SlidingWindow} Algorithm} \label{code:slidingwindow}
\algocapdown
\end{figure}

Given the grid $R_i$ and a path $P=\langle v_1,v_2,\ldots,v_k\rangle$, the algorithm scans the nodes in $P$ one by one (Line 3-6 in Figure~\ref{fig:slidingwindow}). Each time it scans a node $v_j$ in $P$, it computes a minimal rectangular region $B_\theta$ (in $R_i$) that contains all nodes $v_1, v_2, \ldots, v_j$ that have been visited (Line 5). If $B_\theta$ is at least $4$ cells in width or height, then the path $\langle v_1, v_2, \ldots, v_j \rangle$ must contain a sub-path $P'$ that is a spanning path of some ($4\times4$)-cell region $B$. To derive such a region $B$, the algorithm inspects the nodes in $\{v_1, v_2, \ldots, v_j\}$, and then identifies two nodes $v_{\alpha}$ and $v_{\beta}$ as the endpoints of the sub-path $P'$ (Lines 7-8). Based on $v_{\alpha}$ and $v_{\beta}$, the algorithm identifies $B$ (Line 12), and then returns $B$ and $P'$ as the result (Lines 13-14).

For example, let us consider the path $P=\langle v_1, v_2, \ldots, v_6 \rangle$ shown in Figure~\ref{fig:slidingwindow}. Given $P$ and $R_i$, the {\em SlidingWindow} algorithm examines the nodes in $P$ one by one, and monitors the minimal rectangular region $B_\theta$ (in $R_i$) that covers all nodes visited. Figure~\ref{fig:slidingwindow} illustrates the region $B_\theta$ right after $v_5$ is visited. As $B_\theta$ is $5$ cells in width, the algorithm stops the examination procedure, and identifies the nodes with the smallest (resp.\ largest) x-coordinate in $\{v_1, v_2, \ldots, v_5\}$, i.e., $v_2$ (resp.\ $v_5$). After that, the algorithm derives the ($4{\times}4$)-cell region $B$ in Figure~\ref{fig:slidingwindow}, such that (i) $B$ covers $v_1, v_2, v_3, v_4$ simultaneously, and (ii) $v_2$ is in the west strip of $B$. Finally, the algorithm returns $B$ and the path $P' = \langle v_2, v_3, v_4, v_5\rangle$.

\begin{figure}[t]
%\begin{small}
\centering

\includegraphics[width = 2.3in]{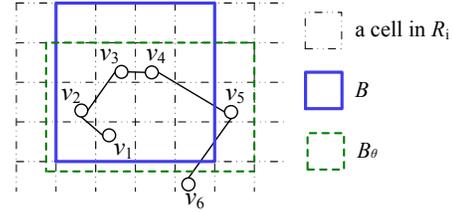}
%\vspace{-3.52mm}\
\vspace{-2mm}
\figcapup \caption{Illustration of the {\em SlidingWindow} algorithm.} \figcapdown \label{fig:slidingwindow}
%\end{minipage}
%\end{small}
%\vspace{-1mm}
\end{figure}

\section{Proofs of Lemmas 1 and 2} \label{sec:FC-pre}

\setcounter{lemma}{0}
\begin{lemma}
Any ($\alpha {\times} \alpha$)-cell region in $R_i$ contains $O(\alpha^2 \a)$ level-$i$ nodes in $H$, where $\a$ is the arterial dimension of $G$.
\end{lemma}
\begin{proof}
Consider the ($\alpha {\times} \alpha$)-cell region $A$ in Figure~\ref{fig:density_1}, as well as the ($(\alpha{+}6)\times (\alpha{+}6)$)-cell region that is centered at $A$. The level-$i$ nodes that fall in $A$ can be categorized into two overlapping groups. The first group contains the endpoints of the arterial edges for a region $B_1$ completely covered by the ($(\alpha{+}6)\times (\alpha{+}6)$)-cell region, and the second group consists of the endpoints of the arterial edges for a region ($4{\times}4$)-cell region $B_2$ that is disjoint from $A$.

The number of nodes in the first group is at most $2\cdot (\alpha+3)^2 \cdot \a$ nodes. This is because (i) the number of ($4{\times}4$)-cell regions contained in the ($(\alpha{+}6)\times (\alpha{+}6)$)-cell area is $(\alpha+3)^2$, and (ii) each ($4{\times}4$)-cell region has $\lambda$ arterial edges. Meanwhile, all nodes in the second group also appear in the first group. To explain, observe that for any node $u$ in $A$ and any ($4{\times}4$)-cell region $B_2$ that is disjoint from $A$, if $u$ is the endpoint of an arterial edge for $B_2$, then the edge must (i) connects $u$ to a node $v$ in $B_2$ and (ii) lies on a spanning path $P$ of $B_2$. It can be verified that there should exist a ($4{\times}4$)-cell region $B$ in $A$, such that $B$ contains $u$, and $P$ is a spanning path of $B$, as exemplified in Figure~\ref{fig:density_2}. In that case, the edge between $u$ and $v$ would also be an arterial edge for $B$. In other words, the node $u$ is also contained in the first group mentioned before. As a consequence, the total number of level-$i$ nodes in $A$ is $2\cdot (\alpha+3)^2 \cdot \a$, which proves the lemma.
\end{proof}

\begin{lemma} \label{lemma:fc-sliding}
Let $P$ be a shortest path in $G$, such that no ($3 {\times} 3$)-cell region in $R_i$ ($i \in [1, h]$) can cover all nodes in $P$ simultaneously. Then, $P$ must contain an arterial edge of some ($4 {\times} 4$)-cell region in $R_i$.
\end{lemma}
\begin{proof}
By the \emph{SlidingWindow} algorithm, we can always find a ($4 {\times} 4$)-cell region (denoted as $B$) in $R_i$, such that a sub-path of $P$ (denoted as $P'$) is a spanning path of $B$. Whenever such a region $B$ exists, $P$ must contain an arterial edge for $B$, and hence, the lemma is proved.
\end{proof}

\setcounter{lemma}{6}

\section{Proof of Lemma 3} \label{sec:proof-3}
%Our proof of Lemmas \ref{lmm:ah-longpath} and \ref{lmm:ah-node-num} are based on a series of lemma on the property of AH's preprocessing algorithm. In what follows, we will first revisit AH's preprocessing algorithm, and then elaborate the proofs. Unless otherwise specified, we use the term {\em region} to refer to a ($4{\times}4$)-cell region.

In this section, we first revisit the preprocessing algorithm of AH, based on which we present the proof of Lemma~\ref{lmm:ah-longpath}.

\subsection{Preprocessing Algorithm Revisited} \label{sec:ah-revision}
Given a road network $G$, AH selects \emph{level $i$ cores} (the nodes that are at least at level $i$) in an incremental manner. At the $i$-th iteration, AH performs two steps: (i) it computes the spanning paths so as to select the level-$i$ cores, and (ii) it adds shortcuts concerning the boarder nodes in $R_{i+1}$ and the level-$i$ cores to obtain a reduced graph $G'_i$ for the next iteration.

More specifically, in the first step, each original edge $\langle u,v\rangle\in G$ is considered an edge generated from a region $B$, if $u$ is in $B$. Then, as for each iteration, AH selects level-($i{+}1$) cores in the following manner. First, AH imposes the grid $R_{i+1}$ on $G'_i$ (note that $G_0'=G$). Then, for each region $B$ in $R_{i+1}$, AH inspects each sub-graph of $G'_i$ that overlaps with $B$, denoted as $G'_{i,B}$. For each boarder node $s$ of $B$, AH invokes Dijkstra's algorithm to start a traversal on $G'_{i,B}$ from $s$; for each node $u$ visited, AH follows its outgoing edges $\langle u,v \rangle$ such that (i) $v$ is a level-$i$ core, and (ii) $\langle u,v \rangle$ satisfies the coverage condition. This results in a spanning tree $T_s$. Subsequently, for each node $u$ on $T_s$, AH inspects its outgoing edges $\langle u,t \rangle$ in $G'_{i,B}$, such that (i) $\langle u,t \rangle$ satisfies the coverage condition, and (ii) $t$ is a boarder node in $B$, or $t$ is not in $B$ and $t$ is a level-$i$ core. As such, a path from $s$ to $t$ is obtained. Similarly, AH invokes Dijkstra's algorithm to start a traversal from $s$ again, but with the difference that for each node visited, AH follows its incoming edges. We use $\mathcal{P}_{i+1,B}$ to denote the paths thus obtained. We will prove that each path $P\in \mathcal{P}_{i+1,B}$ is a spanning path of $B$ in Lemma~\ref{lemma:spanning_path}, based on Lemma~\ref{lemma:sliding_window} below.

After the level-($i{+}1$) cores are selected, AH adds shortcuts concerning the boarder nodes of $R_{i+2}$ and the level-($i{+}1$) cores to form $G_{i+1}'$ in this manner: let $B$ be a ($4\times 4$)-cell region in $R_{i+1}$, and $u$ a boarder node of $R_{i+2}$ or a level-($i{+}1$) core in $B$. Then, similar to the process of finding spanning paths, AH invokes a similar constrained version of Dijkstra's algorithm to start a traversal on $G'_{i,B}$ from $u$. This results in a spanning tree $T_u$. Subsequently, AH examines each branch on $T_u$: let $v$ be the first level-($i{+}1$) core on the branch. Then, AH adds $\langle u, v\rangle$ as a shortcut. Similarly, AH also invokes a constrained version of Dijkstra's algorithm to  start a traversal from $s$ again, but with the difference that for each node visited, AH follows its incoming edges. After that, AH adds shortcuts from level-($i{+}1$) cores to $u$.

For convenience, we define a few terms that will be frequently used in our proofs. Let $B$ be a region in $R_i$ and $B'$ be a region in $R_j$ ($i<j$). If $B$ is completely contained in $B'$, we say that $B$ is a {\em sub-region} of $B'$ and $B'$ is a {\em super-region} of $B$. Let $P$ be a path from a node $s$ to another node $t$. We say $P$ is {\em contained} in a region $B$, if all the nodes on $P$ are in $B$. Under the grid $R_i$, we say two nodes $s$ and $t$ are {\em far-apart} if they are not covered in the same $3\times 3$ cell region.

\subsection{Supporting Lemmas and Proofs} \label{sec:proof-3-proof}

Our proof of Lemma~\ref{lmm:ah-longpath} is based on Lemma~\ref{lemma:sliding_window} and Lemma~\ref{lemma:spanning_path}.

\begin{lemma} \label{lemma:sliding_window}
Let $B_s$ be a region in $R_j$, $s$ and $t$ be two nodes in $B_s$, and $P$ be the local shortest path from $s$ to $t$ in $B_s$. If no ($4\times 4$)-cell region in $R_i$ ($i<j$) can cover all nodes on $P$, then:
\begin{enumerate}
\item [1.] There is a ($4\times 4$)-cell sub-region of $B_s$ in $R_i$ (denoted as $B$), that makes a sub-path of $P$ (denoted as $P'$) a spanning path of $B$.
\item [2.] Either (i) $P'$ is contained in $B$, or (ii) $P'$ is not fully contained in $B$, and the last two nodes on $P'$ are not in two adjacent cells in $R_i$.
\end{enumerate}
\end{lemma}
\begin{proof}
First, we can obtain $B$ and $P'$ using the \emph{SlidingWindow} algorithm in Figure~\ref{code:slidingwindow}, with $P$ and $R_i$ as the input. Second, we show that using the \emph{SlidingWindow} algorithm, we can find $B$ which is a sub-region of $B_s$.  Without loss of generality, suppose that $B_\theta$ is at least $4$ cells in width (Line 7 in Figure~\ref{fig:slidingwindow}), and $P'=\langle v_a, v_{a+1}, \ldots, v_b\rangle$ is a horizontal spanning path of $B$, i.e., $v_a$ is in the west strip of $B$ and $v_b$ is in the east strip (or to the east strip) of $B$.  Suppose that $v_c$ has the smallest y-coordinate among the nodes $v_a, v_{a+1}, \ldots, v_b$. Then we can derive a ($4\times 4$)-cell region $B$, where $v_c$ is in the south strip of $B$. It could be verified that $B$ is a sub-region of $B_s$ because: (i)  the side length of a cell in $B$ is smaller than the side length of a cell in $B_s$, (ii) the largest difference of the y-coordinates among the nodes $v_a, v_{a+1}, \ldots, v_{b-1}$ is less than $4$, and (iii) all the nodes on $P'$ is contained in $B_s$. Hence, Statement 1 is proved.

Next, we show that the $P'$ and $B$ obtained satisfy the two conditions stated in Statement $2$.  Without loss of generality, suppose that after the iteration (Line 3-6 in Figure~\ref{fig:slidingwindow}) terminates, $B_\theta$ is at least four cells in width (Line 7 in Figure~\ref{fig:slidingwindow}). We consider two cases: (i) $B_\theta$'s width is exactly $4$, and (ii) $B_\theta$'s width is larger than $4$. In case (i), it can be verified that $P'$ is contained in $B$. In case (ii), apparently, the width of $B_\theta$ is less than $4$ before $v_\theta$ is visited, and is larger than $4$ after $v_\theta$ is visited. Therefore, $v_{\theta-1}$ and $v_\theta$ cannot be in two adjacent cells.

Hence, the lemma is proved.
\end{proof}

\begin{lemma} \label{lemma:spanning_path}
Let $B$ be a ($4\times 4$)-cell region in $R_i$. Then, the following statements are true:
\begin{enumerate}
\item[1.] Each $P\in \mathcal{P}_{i,B}$ is a spanning path of $B$.

\item[2.] For any $P\in \mathcal{P}_{i,B}$, a path from $s$ to $t$, either (a) $P$ contains only one edge, and $s,t$ are level-$i$ cores, or (b) $P$ contains more than one edge, then a node $w$ on $P$ with $w\not=s,t$ is a level-$i$ core, and $w$ is in $B$.

\item[3.] Let $B_s$ be a region in $R_j$ ($j>i$), $s,t$ be two nodes in $B_s$, and $P$ the local shortest path from $s$ to $t$ in $B_s$. If $s,t$ are far-apart in $R_i$, then there exists a ($4\times 4$)-cell region $B$ in $R_i$, where $B$ is a sub-region of $B_s$, such that $P$ covers a path in $\mathcal{P}_{i,B}$.

\item[4.] Let $B_s$ be a region in $R_j$ ($j>i$). Let $s,t$ be two nodes in $B_s$ and $P$ the shortest path from $s$ to $t$ in $B_s$. If $s,t$ are far-apart in $R_i$, then $P$ is covered by a level-$i$ core. Further, if $P$ contains multiple edges, then there is a level-$i$ core $u$ on $P$ where $u\not=s,t$.
\end{enumerate}
\end{lemma}
\begin{proof}
This lemma could be proved by mathematical induction.

As for Statement 1, for simplicity, we only consider the paths from the west strip to the east strip of $B$. The lemma could be proved true for the paths in the other directions in a similar way. Within a ($4\times 4$)-cell region $B$, the algorithm finds out two types of paths. Suppose that $P$ is a path from $s$ to $t$ found by the algorithm. Then, either (a) $P$ ends at a boarder node $t$ of $B$, and $t$ lies in the east strip of $B$, or (b) $t$ is to the east of $B$, while $t$, and the predecessor of $t$ on $P$, are both level-($i{-}1$) cores. We prove Statement 1 respectively concerning these two types of paths. For ease exposition, we make a slight difference from the algorithm described in Section~\ref{sec:ah-revision}: after the level-$i$ cores are selected, we add shortcuts concern all the nodes in $G$ and the level-$i$ cores to obtain $G'_i$. At the end of the proof, we will show that if we only add shortcuts concerning the boarder nodes of $R_{i+1}$ and the level-$i$ cores, this proof also works. Furthermore, based on Statement 1, we prove Statements 2, 3, and 4 also hold.

To facilitate our proof, we make the following six claims. Claims 1.1 to 1.3 together prove Statement 1 true. Claims 1.4 to 1.6 are used for induction.
\begin{enumerate}
\item[1.1.] Let $u,v$ be two level-($i{-}1$) cores in $B$. Then the local shortest path from $u$ to $v$ contained in $B$ could be found by invoking a constrained version of Dijkstra's algorithm to start a traversal on $G'_{i-1}$ within $B$ from $u$: for each node $w$ visited, the traversal only relaxes the edges $\langle w,x\rangle$, where $x$ is a level-($i{-}1$) core and $\langle w,x \rangle$ satisfies the coverage condition.

\item[1.2.] Let $u$ be a node in $B$ and $v$ a level-($i{-}1$) core in $B$. Then the shortest path from $u$ to $v$ contained in $B$ could be found by a Dijkstra traversal as described in Claim 1.1.

\item[1.3.] Let $u$ be a level-($i{-}1$) core in $B$ and $v$ a node in $B$. Then the shortest path from $u$ to $v$ contained in $B$ could be found in two steps: (i) performs a Dijkstra traversal as described in Claim 1.1, and (ii) after all the level-($i{-}1$) cores in $B$ reachable from $u$ are visited, a spanning tree from $u$ is created; inspects the edges $\langle w,v\rangle$ where $w$ is on the spanning tree and $\langle w,v\rangle$ satisfies the coverage condition to obtain the shortest path from $u$ to $v$.

\item[1.4.] Let $u,v$ be two level-$i$ cores in $B$. Then the shortest path from $u$ to $v$ contained in $B$ could be found by by invoking a constrained version of Dijkstra's algorithm to start a traversal within $B$ from $u$: for each node $w$ visited, the traversal only relaxes the edges $\langle w,x\rangle$ where $x$ is a level-$i$ core and $\langle w,x\rangle$ is generated from $B$.

\item[1.5.] Let $u$ be a node in $B$ and $v$ be a level-$i$ core in $B$. Then the shortest path from $u$ to $v$ contained in $B$ could be found by a Dijkstra traversal as described in Claim 1.4.

\item[1.6.] Let $u$ be a level-$i$ core in $B$ and $v$ a node in $B$. Then the shortest path from $u$ to $v$ contained in $B$ could be found in two steps: (i) performs a Dijkstra traversal as described in Claim 1.4, and (ii) similar to that of Claim 1.3, inspects the edges $\langle w,v\rangle$ where $w$ is on the spanning tree obtained in (i), and $\langle w,v\rangle$ is generated from $B$ to get the shortest path from $u$ to $v$.
\end{enumerate}

Our proof is organized as shown in Figure~\ref{fig:proof-structure}. At the beginning, we show all statements and claims are true when $i=1$. Subsequently, assuming all statements and claims are true when $i=k$, we prove that they also hold when $i=k+1$ in an order denoted by the circle numbers shown in Figure~\ref{fig:proof-structure}.
%In particular, we first prove that the six claims hold. Given Claim 1.4 and Statements 2 and 4 when $i = k$, we prove that Claim 1.1 holds when $i=k+1$. Then, given Claim 1.5 (resp.\ Claim 1.6) when $i=k$ and Claim 1.1 when $i=k+1$, we show that Claim 1.2 (resp.\ Claim 1.3) is true when $i=k+1$. After that, given Claim 1.2 (resp.\ Claim 1.3) when $i=k$ and Claim 1.4 when $i=k+1$, we prove Claim 1.5 (resp.\ Claim 1.6). Finally, given Claims 1.1 to 1.3 when $i = k+1$, we prove that Statements 1 to 4 hold when $i = k+1$.

\begin{figure}[t]
%\begin{small}
\centering
\includegraphics[width = 3.3in]{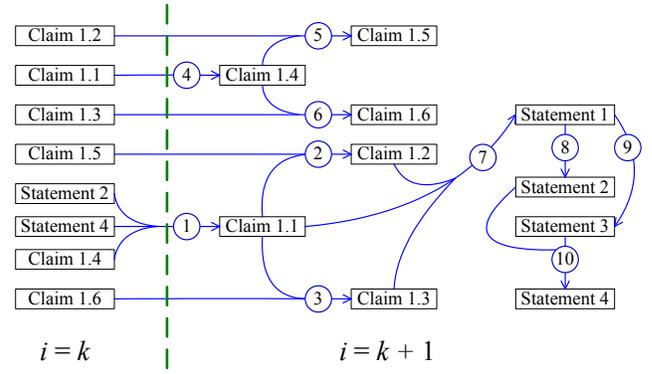}
\vspace{-2mm}
\figcapup \caption{Proof structure.} \figcapdown \label{fig:proof-structure}
\vspace{-2mm}
%\end{small}
\end{figure}

It could be verified that Claims 1.1 to 1.3 are true when $i=1$ given all the nodes in $G$ are level-$0$ cores, and $G'_0$ is composed of the edges in $G$. Claim 1.4 is true when $i=1$. Suppose that $P$ is a local shortest path from $u$ to $v$, and $t_1,\ldots,t_j$ are in turns the level-$1$ cores on $P$. Then the algorithm would add a shortcut $\langle u,t_1\rangle$ (or $\langle u,t_1\rangle$ is an original edge in $G$) whose length equals to the distance in $G'_{0}$ from $u$ to $t_1$, and $\langle u,t_1\rangle$ is generated from $B$. So are the edges $\langle t_1,t_2\rangle,\ldots,\langle t_j,v\rangle$. Hence, Claim 1.4 holds when $i=1$. Claim 1.5 is similar to Claim 1.4, except that the source node $u$ is not necessarily a level-$1$ core. Since the algorithm also adds an edge from $u$ to a level-$1$ core, similar to the case of Claim 1.4, Claim 1.5 also holds when $i=1$. As for Claim 1.6, Suppose $P$ is a local shortest path from $u$ to $v$, and $t_1,\ldots,t_j$ are in turns the level-$1$ cores on $P$. Then $\langle t_j,v \rangle$, is a shortcut (or original edge) generated from $B$, and the weight of $\langle t_j,v\rangle$ equals that of the shortest path from $t_j$ to $v$ contained in $B$. Besides, by Claim 1.4, $u$ could equally reach $t_j$ by the edges generated from $B$. Hence, Claim 1.6 holds when $i=1$. Given Claims 1.1 to 1.3, Statement 1 is true when $i=1$.

  Statement 2 could be proved by the algorithm when $i=1$. If $P$ contains one edge, both the endpoints of $P$ would be selected as a level-$1$ core. If $P$ contains more than one edge, then an internal node of $P$ (denoted as $w$), $w\not=s,t$ is selected as a level-$1$ core. By Statement 1, $P$ is a spanning path of $B$, and $w$ is not an endpoint of $P$. Hence, $w$ is in $B$.

  As for Statement 3, when $i=1$, first we show that there exists a ($4\times 4$)-cell region $B$ in $R_1$, such that a sub-path of $P$ (including $P$ itself) is a spanning path of $B$. If $P$ is contained in a ($4\times 4$)-cell region $B$ in $R_1$, since $s,t$ are far-apart in $R_1$, $P$ is a spanning path of $B$. Otherwise, if $P$ is not contained in any ($4\times 4$)-cell region in $R_1$, then, by Lemma~\ref{lemma:sliding_window}, there exists a region $B$ in $R_1$, which makes a sub-path of $P$ a spanning path of $B$. Second, note that at the first iteration, given $B$ all its spanning paths could be found because the paths are found on the original graph $G$. Hence, there should exist a region $B$ in $R_1$, such that $P$ covers a path in $P_{1,B}$.

  Statement 4 could be proved by Statement 2 and 3. By Statement 3, there is a ($4\times 4$)-cell region $B$ in $R_i$, such that a sub-path $P'$ of $P$ covers a path in $P_{i,B}$. If $P'$ is exactly $P$, since $P$ contains more than one edge, then by Statement 2, there should be a node $u$ on $P$ other than $s,t$, such that $u$ is a level-$i$ core. If $P'$ is not $P$, nevertheless where the positions of the two endpoints of $P'$ on $P$ are, by Statement 2, there should exist a node $u$ ($u\not=s,t$) on $P$ and $u$ is a level-$i$ core. Hence, in either case, Statement 4 is true.

    The above all show that the lemma is true when $i=1$. We now turn to the induction phase. Suppose the statements and claims above are true when $i=k$, we show them true when $i=k+1$.

    Step \circled{1} in Figure~\ref{fig:proof-structure}: Claim 1.1 is true when $i=k+1$ given Claim 1.4, Statement 2 and 4 true when $i=k$. We prove the claim in two cases: (a) $P$ is also contained in a sub-region $B'$ of $B$, where $B'$ is a ($4\times 4$)-cell region in $R_k$, and (b) $P$ is not contained in any ($4\times 4$)-cell region in $R_k$. As for case (a), given Claim 1.4 true when $i=k$,  Claim 1.1 is true because a level-$k{+}1$ core is also a level-$k$ core. As for case (b), if $P$ contains only one original edge $\langle u,v\rangle$, then $\langle u,v\rangle$ is an edge that satisfies the coverage condition, which follows that Claim 1.1 is true. Subsequently we consider the case when $P$ contains multiple edges. If $u,v$ are far-apart in $R_k$, then by Statement 4 when $i=k$, there should be a node $w$ ($w\not=u,v$) on $P$ and $w$ is a level-$k$ core. If $u,v$ are not far-apart in $R_k$, given the hypothesis that $P$ is not contained in any ($4\times 4$)-cell region in $R_k$, there should exist two nodes $u'$ and $v'$ such that $u'$ and $v'$ are far-apart in $R_k$. Again, by Statement 4, there also exists a node $w$ on $P$ ($w\not=u,v$) and $w$ is a level-$k$ core. Since $u,w,v$ are all level-$k$ cores, similarly we can consider the sub-path from $u$ to $w$ (and the sub-path from $w$ to $v$ as well) in the two cases above. Because $P$ contains a finite number of nodes, $P$ could not be infinitely decomposed. i.e. we could always find a sub-path $P'$ of $P$ between two level-$k$ cores, such that, either (i) $P'$ contains only one edge; or (ii)  $P'$ is contained in a ($4\times 4$)-cell region in $R_k$, and by Claim 1.4 true when $i=k$, $P'$ could be discovered by a constrained version of Dijkstra traversal which relaxes the edges concerning level-$k$ cores. Therefore, Claim 1.1 is true when $i=k+1$.

Step \circled{2} in Figure~\ref{fig:proof-structure}: Claim 1.2 is true when $i=k+1$ given Claim 1.1 true when $i=k+1$, and Claim 1.5 true when $i=k$.  The proof of Claim 1.2 is similar to that of Claim 1.1, except that the source $u$ is a level-$0$ core instead of a level-$k$ core. We can consider $P$ in two cases: whether it is contained in a ($4\times 4$)-cell region in $R_k$ or not. We focus on the sub-path from $u$. In this way we can always find a level-$k$ core $w$, so that, either the sub-path from $u$ to $w$ is contained in a ($4\times 4$)-cell sub-region $B'$ in $R_k$, or $\langle u,w\rangle$ is an original edge. In the former case, given Claim 1.5 true when $i=k$, Claim 1.2 is true when $i=k+1$. In the latter case, $\langle u,w\rangle$ is an edge that satisfies the coverage condition, hence, Claim 1.2 is also true.

Step \circled{3} in Figure~\ref{fig:proof-structure}: Claim 1.3 could be similarly proved like Claim 1.2 given Claim 1.6 true when $i=k$, and Claim 1.1 true when $i=k+1$.

Step \circled{4} in Figure~\ref{fig:proof-structure}: Claim 1.4 could be proved true when $i=k+1$ given Claim 1.1 true when $i=k+1$. Consider a shortest path $P$ between two level-($k{+}1$) cores $u,v$. contained in a ($4\times 4$)-cell region $B$ in $R_{k+1}$. By Claim 1.1 when $i=k+1$, $P$ should be equally found by invoking a Dijkstra algorithm to start a traverse which only vista the level-$k$ cores. Along $P$, let $t_1, t_2, \ldots, t_j, v$ be the level-($k{+}1$) cores (note that a level-($k{+}1$) core is also a level-$k$ core). Then AH would add $\langle u,t_1 \rangle$ (if not existed in $G$) as a shortcut generated from $B$, and its weight equals the weight of the shortest path from $u$ to $t_1$ in $B$. So are the shortcuts $\langle t_1,t_2\rangle, \ldots, \langle t_j,v\rangle$. Hence, Claim 1.4 is true when $i=k+1$.

Step \circled{5} (resp. \circled{6}) in Figure~\ref{fig:proof-structure}: Claim 1.5 (resp. Claim 1.6) could also be proved given Claim 1.1 true when $i=k+1$, and Claim 1.2 (resp. Claim 1.3) true when $i=k$.

Step \circled{7} in Figure~\ref{fig:proof-structure}: Statement 1 is true when $i=k+1$ given Claims 1.1 to 1.3 and Statement 4 true when $i=k$. First, type (a) path could be correctly found. Let $s$ (resp. $t$) be a boarder node in the west (resp. east) strip of $B$. Then, for any pair of such $s$ and $t$, the shortest path $P$ from $s$ to $t$ contained in $B$ is a spanning path of $B$. Besides, followed by Statement 4, $P$ is covered by a level-$k$ core since $s$ and $t$ are far-apart in $R_{k+1}$ (which follows that $s$ and $t$ are far-apart in $R_k$). Subsequently, given Claims 1.1 to 1.3, $P$ could be correctly found. Second, type (b) path could also be correctly found. Let $P$ be a local shortest path of $B$ from $s$ to $t$ where $t$ is beyond $B$. Let $\langle u,t\rangle$ be the last edge on $P$, and $u$ a level-$k$ core. Then, given Claim 1.1 and Claim 1.2, the shortest path from $s$ to $u$ contained in $B$ could be correctly found. On the other hand, $\langle u,t\rangle$ satisfies the coverage condition. As a result, $P$ could also be correctly found. The above all shows that: (i) for a spanning path $P$ of type (a) or type (b), $P$ could be found by the algorithm supported by Claims 1.1 to 1.3. And (ii) every $P\in \mathcal{P}_{k+1,B}$ is a spanning path of $B$. Hence, Statement 1 is true when $i=k+1$.

Step \circled{8} in Figure~\ref{fig:proof-structure}: Statement 2 is true when $k+1$. The proof is similar to the case when $i=1$. By the algorithm, each $P\in \mathcal{P}_{k+1,B}$, satisfies either condition stated in Statement 2.

Step \circled{9} in Figure~\ref{fig:proof-structure}: Statement 3 is true given Statement 1 true when $i=k+1$. If $P$ is not contained in any ($4\times 4$)-cell region in $R_{k+1}$, by Lemma~\ref{lemma:sliding_window}, there should exist a ($4\times 4$)-cell region $B$ in $R_{k+1}$, where $B$ is a sub-region of $B_s$, such that a sub-path of $P$ (denoted as $P'$) is a spanning path of $B$. If $P$ is contained in a ($4\times 4$)-cell region in $R_{k+1}$, we put $P'=P$. In what follows, we show that $P'$ covers a path in $\mathcal{P}_{k+1,B}$. We consider $P'$ in two cases: (a) $P'$ is contained in $B$, and (b) $P'$ is not contained in $B$. Suppose that $P'$ is from $u$ to $v$. Without loss of generality, suppose that $u$ is in the west strip of $B$ and $v$ is in the west strip. In case (a), we show that there should exist two nodes $u'$ and $v'$ on $P'$, such that the sub-path of $P'$ from $u'$ to $v'$ is a type (a) spanning path in $\mathcal{P}_{k+1,B}$: starting with $u$, we scan each node on $P'$ one by one. Stop until the first time a node $v'$ in the east strip of $B$ is met. Such $v'$ exists because $P'$ ends at a node $v$ in the east strip ($v'$ might equal to $v$). Similarly, there should exist a node $u'$ such that $u'$ is in the west strip, and the successors of $u'$ on $P'$ is to the east of the west strip ($u'$ might equal to $u$ too). On the other hand, the sub-path of $P'$ from $u'$ to $v'$ is a type (a) spanning path that could be found by the algorithm, hence, $P$ covers a path in $\mathcal{P}_{k+1,B}$. In case (b), we show that $P'$ covers a type (b) spanning path in $\mathcal{P}_{k+1,B}$. First, similar to case (a), there should be a node $u'$ on $P'$, such that $u'$ is in the west strip, and the successors of $u'$ on $P'$ is in to the east of the west strip. On the other hand, let $v'$ be the predecessor of $v$ on $P'$. By Lemma~\ref{lemma:sliding_window}, since $P'$ is not contained in $B$, the $v'$ and $v$ are not in two adjacent cells in $R_{k+1}$, which follows that $v'$ and $v$ are far-apart in $R_k$. Hence, the sub-path of $P'$ from $u'$ to $v$ is a type (b) spanning path in $\mathcal{P}_{k+1,B}$.

 Step \circled{10} in Figure~\ref{fig:proof-structure}: Statement 4 could be proved when $i=k+1$ given Statements 2 and 3 are true when $i=k+1$. The proof is similar to the case when $i=1$.

  Finally, we show that the Statements and Claims also hold if at each iteration, AH only adds shortcuts concern the boarder nodes of $R_{i+1}$ and the level-$i$ cores. Let $B$ be a ($4\times 4$)-cell region in $R_{i+1}$ ($i\geq 1$). In the computation of spanning paths,  AH invokes a Dijkstra algorithm to start a traversal from a boarder node of $R_{i+1}$. During the traversal, AH only visits the level-$i$ cores.  After the traversal is completed, AH only visits the edges from a level-$i$ core to a boarder node to obtain a spanning path. Hence, AH only uses the edges concerning the boarder nodes of $R_{i+1}$ and the level-$i$ cores. As such, it suffices to only add shortcuts concerning the boarder nodes and level-$i$ cores.
%  Above all, the lemma is proved.
\end{proof}

\setcounter{lemma}{2}
\begin{lemma} %\label{lmm:ah-longpath}
For any two nodes $u,v \in G$, if no ($3 {\times} 3$)-cell region in $R_i$ ($i \in [1, h]$) can cover $u$ and $v$ simultaneously, then the shortest path from $u$ to $v$ must go through a node at level $i$ or above.
\end{lemma}
\begin{proof}
This lemma follows from Statement 4 of Lemma~\ref{lemma:spanning_path}, when $B_{s}$ is the ($4\times 4$)-cell region in $R_h$ where $B_s$ covers the entire road network $G$.
\end{proof}

\setcounter{lemma}{8}

\section{Proof of Lemma 4} \label{sec:proof-4}

Our proof of Lemma~\ref{lmm:ah-node-num} is based on the following lemma.

\begin{lemma} \label{lemma:density}
Let $B$ be any ($4\times 4$)-cell region in $R_i$ ($i \in [1, h]$), $E_B$ be the set of pseudo-arterial edges for $B$, and $V_B$ be a set containing the endpoints of the edges in $E_B$. Then, the number of nodes in $V_B$ is at most $50\lambda^2$.
\end{lemma}

In what follows, we will first prove Lemma~\ref{lmm:ah-node-num} based on Lemma~\ref{lemma:density}, and then establish the validity of Lemma~\ref{lemma:density}.

\setcounter{lemma}{3}
\begin{lemma} %\label{lmm:ah-node-num}
Any ($\alpha {\times} \alpha$)-cell region in $R_i$ contains $O(\alpha^2 \a^2)$ nodes whose level in $\H$ are no lower than $i$, where $\a$ is the arterial dimension of $G$.
\end{lemma}
\begin{proof}
Consider the ($\alpha {\times} \alpha$)-cell region $A$ in Figure~\ref{fig:density_1}, as well as the ($(\alpha{+}6)\times (\alpha{+}6)$)-cell region that is centered at $A$. Contained in $A$, the nodes whose level in $\H$ are no lower than $i$ (i.e., the level-$i$ cores), can be categorized into two overlapping groups. The first group contains the level-$i$ cores selected due to a spanning path of a ($4\times 4$)-cell region $B_1$ completely covered by the ($(\alpha{+}6)\times (\alpha{+}6)$)-cell region, and the second group consists of the level-$i$ cores selected from a region ($4{\times}4$)-cell region $B_2$ that is disjoint from $A$.

The number of nodes in the first group is $50\cdot (\alpha+3)^2 \cdot \a^2$. This is because (i) the number of ($4{\times}4$)-cell regions contained in the ($(\alpha{+}6)\times (\alpha{+}6)$)-cell area is $(\alpha+3)^2$, and (ii) each ($4{\times}4$)-cell region generates at most $50\lambda^2$ level-$i$ cores followed by Lemma~\ref{lemma:density} (note that the endpoints of the pseudo-arterial edges in $R_i$ are level-($i{-}1$) cores). Meanwhile, all nodes in the second group also appear in the first group. To explain, note that for any level-$i$ core $u$ in $A$, and any ($4\times 4$)-cell region $B_2$ that is disjoint from $A$, if $u$ is selected as a level-$i$ core due to a spanning path $P$ of $B_2$, then, by Statement 2 of Lemma~\ref{lemma:spanning_path}, $P$ contains only one edge since $u$ is not in $B_2$. It could be verified that there should exist a ($4\times 4$)-cell region $B$ in $A$, such that $B$ contains $u$, and $P$ is also a spanning path of $B$, as exemplified in Figure~\ref{fig:density_2}. In other words, the node $u$ is also contained in the first group mentioned before, because $u$ is selected due to $P$, a spanning path of $B$. As a consequence, the total number of level-$i$ cores in $A$ is at most $50\cdot (\alpha+3)^2 \cdot \a^2$, which is $O(\alpha^2\a^2)$.
\end{proof}

It remains to prove Lemma~\ref{lemma:density}. The key idea of our proof is to show that, for any ($4\times 4$)-cell region $B$ in $R_i$, the spanning paths of $B$ contain $O(\lambda^2)$ level-($i{-}1$) cores, which results in $O(\lambda^2)$ level-$i$ cores selected for any $i\in [1,h]$. To facilitate our proof, in the following, we will first establish some properties of the shortcuts in $\H$ (in Lemma~\ref{lemma:shortcut_scope}). Next, based on Lemma~\ref{lemma:shortcut_scope}, we demonstrate the characteristics of certain spanning paths of $B$ (in Lemma~\ref{lemma:type_b_long_path}). Subsequently, we will employ Lemma~\ref{lemma:type_b_long_path} to show a general property of every spanning path in a ($4\times 4$)-cell region $B$ in $R_i$ (in Lemma~\ref{lemma:fixnode_path_set}). Finally, we will prove lemma~\ref{lemma:density} based on Lemma~\ref{lemma:fixnode_path_set}.

\setcounter{lemma}{9}
\begin{lemma} \label{lemma:shortcut_scope}
  Let $B$ be a ($4\times 4$)-cell region in $R_i$, and $\langle u, v \rangle$ a shortcut created in $B$. Then the path contracted by $\langle u, v \rangle$ is contained in $B$.
\end{lemma}
\begin{proof}
  This lemma could be proved by mathematical induction. Let $P$ be the path contracted by $\langle u, v \rangle$. First, we show that the lemma holds when $i=1$. Assume to the contrary that the path contracted by $\langle u, v \rangle$ is not contained in $B$. Then, there is a node $x$ on $P$, such that $x$ is beyond $B$. Let $\langle x, y \rangle$ be the edge on $P$. Such $y$ exists because $x\not= v$ given $\langle u,v \rangle$ is a shortcut created in $B$, which follows that $v$ should be in $B$. In that case, $\langle x, y \rangle$ is not an edge generated from $B$ because $x$ is not in $B$.  Second, suppose that the lemma holds when $i=k$, we show that it also holds when $i=k+1$. By contradiction, let $x$ be the node on $P$ that is not in $B$. If $x$ is not a level-$k$ core, $x$ is not visited during the creation of  $\langle u,v\rangle$. It follows that $x$ is contracted by a shortcut $e$ where $e$ is generated from a sub-region of $B$, and given the lemma true when $i=k$, $x$ should be in $B$. If $x$ is a level-$k$ core, let $\langle x, y \rangle$ be the edge visited during the creation of $\langle u,v\rangle$. Then $\langle x, y \rangle$ violates the coverage condition since $x$ is beyond $B$. Hence, the lemma is proved.
\end{proof}

Given Lemma~\ref{lemma:shortcut_scope}, we have the following lemma:
\begin{lemma} \label{lemma:type_b_long_path}
 Let $B$ be a ($4\times 4$)-cell region in $R_i$, $P\in \mathcal{P}_{i,B}$ be a path from $s$ to $t$, and $\langle u,t\rangle$ be the last edge on $P$. If $t$ is beyond $B$, then $u,t$ are level-($i{-}1$) cores.
\end{lemma}
\begin{proof}
  Apparently the lemma holds when $i=1$ given all the nodes in $G$ are level-$0$ cores. Suppose that the lemma holds when $i=k$, in the following we show that the lemma true when $i=k+1$.
  First, $t$ is a level-$k$ core, since except for the source node $s$, the Dijkstra traversal only visits the level-$k$ cores to find the spanning paths. Second, we show that $\langle u,t \rangle$ is visited by the Dijkstra traversal. Suppose that $\langle x, t \rangle$ is the edge on $P$ relaxed by the Dijkstra traversal. Then, $\langle x,t\rangle$ should be an original edge because by Lemma ~\ref{lemma:shortcut_scope}, if $\langle x,t\rangle$ is a shortcut, $t$ should be in $B$, which violates the hypothesis. Hence, $x=u$. In what follows, $u,t$ are both level-($i{-}1$) cores because the Dijkstra traversal would only visit the level-($i{-}1$) cores. Therefore, the lemma holds.
\end{proof}

Given Lemma~\ref{lemma:type_b_long_path}, the following lemma is proved.
\begin{lemma} \label{lemma:fixnode_path_set}
  Let $B$ be a ($4\times 4$)-cell region in $R_i$ ($i\geq 2$), $P\in \mathcal{P}_{i,B}$ a path from a node $s$ to another node $t$. Let $\langle a,b \rangle$ be an arterial edge on $P$. Then, there exists a  ($4\times 4$)-cell sub-region of $B$ in $R_{i-1}$ (denoted as $B'$), such that the sub-path of $P$ from $a$ to $t$ covers a path in $\mathcal{P}_{i-1,B'}$.
\end{lemma}
\begin{proof}
Suppose that $P_{a,t}$ is the sub-path of $P$ from $a$ to $t$. Without loss of generality, suppose that $s$ is in the west strip of $B$, and $\langle a,b \rangle$ goes across the vertical bisector of $B$ (denoted as $l_b$) where $a$ (resp. $b$) lies at the west (resp. east) of $l_b$. According to the definition of \emph{Spanning Path}, in that case, $t$ should be at the east of $l_b$, and $t$ is not at the cell adjacent to $l_b$.  It could be verified that $a$ and $t$ are far-apart in $R_{i-1}$ since the side length of a cell in $R_i$ is two times of that in $R_{i-1}$. If $t$ is contained in $B$, then, followed by Statement 3 of Lemma~\ref{lemma:spanning_path}, this lemma holds. In the following we consider the case where $t$ is beyond $B$.

Let $\langle u,t\rangle$ be the last edge on $P_{a,t}$. We denote the sub-path of $P_{a,t}$ from $a$ to $u$ as $P_{a,u}$. If on $P_{a,u}$ there exist two node $x,y$, such that $x$ and $y$ are far-apart in $R_{i-1}$, then, followed by Statement 3 of Lemma~\ref{lemma:spanning_path}, this lemma holds. Next, we consider the case when there is a ($3\times 3$)-cell sub-region of $B$ in $R_{i-1}$, such that the sub-region contains all the nodes on $P_{a,u}$. Let $a_l$ be the node on $P_{a,u}$ which has the smallest x-coordinate. It could be verified that there is a ($4\times 4$)-cell sub-region of $B$ in $R_{i-1}$, denoted as $B'$, such that: (i) $B'$ contains all the nodes on $P_{a,u}$, and (ii) $a_l$ is in the west strip of $B'$.  Let $v$ be the node on $P_{a,u}$ such that (i) $v$ is in the west strip of $B'$, and (ii) $v$ is a boarder node of $B'$. Figure~\ref{fig:fixnode_path} illustrates an example of the positions of $a$, $b$, $a_l$, $v$, $u$, $t$, $B$ and $B'$ as described above.Then, the path $P_{v,t}$ from $v$ to $t$ is in $\mathcal{P}_{i-1, B'}$ because: (i) $t$ is at the east of $l_b$, but is not in the cell adjacent to $l_b$, which indicates that $t$ is at the east of the vertical bisector of $B'$ (denoted as $l'_b$), and $t$ is not in the adjacent cell of $l'_b$, and (ii) followed by Lemma~\ref{lemma:type_b_long_path}, both $u$ and $t$ should be level-($i{-}1$) cores because $\langle u,t \rangle$ is the last edge on $P$, which indicates that $u,t$ are level-($i{-}2$) cores. Hence, $P_{v,t}$ is in $\mathcal{P}_{i-1,B'}$, and therefore the lemma holds.
\end{proof}
\begin{figure}[t]
%\begin{small}
\centering
\includegraphics[width = 3.3in]{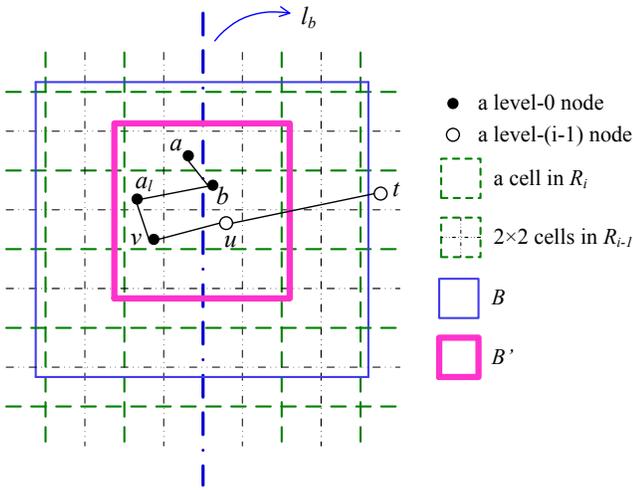}
\vspace{-2mm}
\figcapup \caption{A sub-path from $a$ to $t$.} \figcapdown \label{fig:fixnode_path}
\vspace{-2mm}
%\end{small}
\end{figure}

Given Lemma~\ref{lemma:fixnode_path_set}, we prove Lemma~\ref{lemma:density} as follows:
\begin{proof} [Of Lemma~\ref{lemma:density}]
When $i=1$, the edges in $E_B$ are arterial edges of $B$. Then, there are at most $2\lambda$ nodes in $V_B$. In the following, we show that the lemma also holds when $i\geq 2$.
 % When $i=1$, the spanning paths are found by a Dijkstra traversal which relaxes only the original edges in $G$. By assumption there are at most $\lambda$ arterial edges in $B$. For a lose estimation, both endpoints of those arterial edges are selected as level-$1$ cores. In that case, the number of level-$1$ cores is still less than $50\lambda^2$. In the following, we show that the lemma also holds when $i\geq 2$.

Each pseudo-arterial edge $e\in E_B$ is either an original edge in $G$ or a shortcut. We consider two disjoint subsets $E_1$ and $E_2$ of $E_B$, where $E_1$ contains all the original arterial edges, and $E_2$ the shortcuts. Suppose that $|E_1|=\lambda_1$ and $|E_2|=\lambda_2$. Then, there are at most $2 \lambda_1$ nodes in $V_B$ that are endpoints of the edges in $E_1$. In what follows, we consider $E_2$.

We divide $E_2$ into several disjoint subsets according to the arterial edges $e\in E_2$ contracts, i.e., the shortcuts that contract the same arterial edge are in the same subset. For each subset $E_{sub}$, we show that there are at most $50\lambda$ nodes in $V_B$ that are endpoints of the edges in $E_{sub}$. Suppose that the shortcuts in $E_{sub}$ are $\langle X_1,Y_1 \rangle,\ldots, \langle X_k,Y_k\rangle$ (note that by the AH algorithm, $X_1,\ldots,X_k,Y_1,\ldots,Y_k$ are all level-($i{-}1$) cores), and they are respectively on the paths $P_1,\ldots,P_k$ to the nodes $t_1,\ldots,t_k$. We use $\langle a,b \rangle$ to denote the arterial edge that those shortcuts $\langle X_1,Y_1 \rangle,\ldots, \langle X_k,Y_k \rangle$ contract. Besides, we use $P'_j$ to denote the sub-path of $P_j$ from $a$ to $t_j$ for all $1\leq j \leq k$. Let $G_a$ be a graph which is composed of the edges $e$ where $e$ is on $P'_j$. Then, $G_a$ is a tree, otherwise, it violates the hypothesis that the local shortest paths in $B$ are unique. Besides, $Y_1,\ldots, Y_k$ is in $G_a$, the tree. We show that there are at most $25\lambda$ distinct nodes among $Y_1,\ldots,Y_k$. By Lemma~\ref{lemma:fixnode_path_set}, for each $P'_j$, ($1\leq j \leq k$), there exists a ($4\times 4$)-cell region in $R_{i-1}$ (denoted as $B'$), where $B'$ is a sub-region of $B$, such that $P'_j$ covers a path in $\mathcal{P}_{i-1,B'}$. And followed by Statement 2 of Lemma~\ref{lemma:spanning_path}, $P'_j$ is covered by a level-($i{-}1$) core. Since $B'$ contains at most $\lambda$ arterial edges, and there are at most twenty-five sub-regions of $B$, hence, there are at most $25\lambda$ level-($i{-}1$) cores that cover the paths $P'_j$ for $1\leq j \leq k$. On the other hand, on $P'_j$, among the level-($i{-}1$) cores, $Y_j$ is the closet one to $a$, otherwise, if a level-($i{-}1$) core $u$ is on the path from $a$ to $Y_j$, then $\langle X_j,Y_j \rangle$ contracts a level-($i{-}1$) core, which violates the algorithm. Hence, there are at most $25\lambda$ distinct nodes of $Y_1,\ldots, Y_k$. Symmetrically, the number of distinct nodes of $X_1,\ldots,X_k$ is at most $25\lambda$ as well. Hence, in $V_B$, there are at most $50 \lambda$ nodes that are endpoints of the edges in $E_{sub}$. In addition, given $|E_2|=\lambda_2$, there are at most $\lambda_2$ disjoint subsets, which follows that, in $V_B$, there are at most $50\lambda\cdot\lambda_2$ nodes that are endpoints of the edges in $E_2$.

On the other hand, if $e\in E_B$ is an original edge, $e$ cannot be contracted by other shortcuts in $AE$ at the same time because both of the two endpoints of $e$ are level-($i{-}1$) cores. Hence, totally there are at most $50\lambda^2$ nodes in $V_B$.
\end{proof}

\section{Proof of Theorem 1} \label{sec:complexities}
%\subsubsection{Space Complexity} \label{subsec:space}

We prove Theorem~\ref{thrm:ah-complexity} by presenting a series of lemmas and theorems that establish the space and time complexities of AH, as well as the correctness of AH's query processing algorithms.

\begin{lemma} \label{lemma:num_query}
Each node in $\H$ has $O(\a^2)$ non-elevating edges.
\end{lemma}
\begin{proof}
  Let $u$ be a node at level $i$, and $\langle u, v \rangle$ a non-elevating edge. Under the grid $R_{i+1}$, consider the ($5\times 5$)-cell region centered at $u$ (denoted as $C$). Then, $v$ is in $C$, because the algorithm firstly builds a \emph{SPT} (see definition \ref{def:ah-spt}) from $u$ in $C$, and then adds non-elevating edges from $u$ to the level-$i$ cores on the SPT. Besides, by Lemma~\ref{lmm:ah-node-num}, there are $O(\lambda^2)$ level-$i$ cores in $C$, which is a ($10\times 10$)-cell region in $R_i$. Hence, the number of non-elevating edges is $O(\lambda^2)$.
\end{proof}

\begin{lemma} \label{lemma:num_jumping}
Each node in $\H$ has $O(h\a^2)$ elevating edges.
\end{lemma}
\begin{proof}
Let $\langle u, v \rangle$ be an elevating edge. Then $v$ has a higher rank than that of $u$. Suppose $v$ is at level $i$. It suffices to show that for a fixed $i$, the number of such elevating edges is $O(\a^2)$. This is because $\langle u, v \rangle$ is obtained from a \emph{SPT} rooted at $u$ generated in a ($5\times 5$)-cell region centered at $u$ in $R_i$, and by Lemma~\ref{lmm:ah-node-num} the number of such node $v$ is $O(\a^2)$. In the worst case, $u$ is at level $0$, and AH adds elevating edges from $u$ to the level-$i$ nodes where $i\in [1,h]$. Hence, $u$ has $O(h\a^2)$ elevating edges.
\end{proof}

\begin{theorem}
The space overhead of AH is $O(hn)$.
\end{theorem}
\begin{proof}
By Lemmas \ref{lemma:num_query} and \ref{lemma:num_jumping}, each node has $O(h\a^2)$ edges, and there are $n$ nodes in $G$. Hence, the space overhead is $O(hn)$ when $\a$ is a constant.
\end{proof}

%\subsubsection{Time Complexity} \label{subsec:time}
\begin{theorem}
AH answers a distance query in $O(h \log h)$ time. Besides, it answers a shortest path query in $O(h \log h + k)$ time, where $k$ is the number of edges in the shortest path.
\end{theorem}
\begin{proof}
 AH answers any distance query with two traversals of $\H$, starting from the source $s$ and destination $t$ of the query, respectively. Due to the proximity constraint, in the $i$-th level of $\H$ ($i \in [0, h]$), each traversal of AH only visits the nodes in a ($5 {\times} 5$)-cell region in $R_{i+1}$. By Lemma~\ref{lmm:ah-node-num}, such a region only contains $O(\a^2)$ level-$i$ cores because a ($5\times 5$)-cell region in $R_{i+1}$ is a ($10\times 10$)-cell region in $R_i$. Hence, the total number of nodes traversed by AH is $O(h \a^2)$. Furthermore, for each node $v$ visited during a traversal, AH either follows the elevating edges of $v$ to a certain level of $\H$, or moves along the non-elevating edges of $v$ that satisfy the rank and proximity constraints. As previously discussed in Lemma~\ref{lemma:num_jumping}, $v$ has $O(\a^2)$ elevating edges to each level of $\H$, and has $O(\a^2)$ non-elevating edges. Therefore, the total number of edges visited by AH is $O(h \a^4)$. Since each traversal is performed using Dijkstra's algorithm, its overall time complexity is $O(h\a^2 \log(h\a^2) + h\a^4)$. Consequently, when $\a$ is a constant, the time complexity of AH for a distance query is $O(h \log h)$.

To answer a shortest path query from $s$ to $t$, AH first processes its corresponding distance query to retrieve the shortest path $P'$ from $s$ to $t$ in $\H$, and then its transforms $P'$ into the actual shortest path $P$ from $s$ to $t$ in $G$. For each shortcut $e$, it requires $O(1)$ time to decompose $e$ into two edges, and an original edge cannot be further decomposed. Hence, the transformation from $P'$ to $P$ takes $O(k)$ time, where $k$ is the number of edges in $P$. Therefore, AH requires $O(h \log h + k)$ time to answer a shortest path query.
\end{proof}

\begin{lemma}
AH requires $O(h n^2)$ time to construct $\H$.
\end{lemma}
\begin{proof}
The preprocessing algorithm of AH consists of three steps: (i) assigning nodes to each level of $\H$, (ii) deriving the strict total order on nodes at the same level, and (iii) creating shortcuts in $\H$. As for (i), When assigning nodes to the $i$-th level of $\H$ ($i \in [0, h-1]$), AH inspects each non-empty ($4 {\times} 4$)-cell region in $R_{i+1}$, and constructs a subgraph that consists of the level-$i$ cores and border nodes in the region. For each boarder node $u$, AH needs to apply Dijkstra's algorithm to traverse the subgraph. By Lemma~\ref{lmm:ah-node-num}, $O(\a^2)$ level-$i$ cores are visited during the traversal, and each node visited has $O(\a^2)$ edges. Hence building a Dijkstra tree from $u$ require $O(\a^4)$ time. After the Dijkstra tree is constructed, AH needs to inspect each node in the tree and the boarder nodes to find out a spanning path which requires $O(n \a^2)$, because there are $O(\a^2)$ nodes in the tree and $n$ boarder nodes for a loose estimation. On the other hand, $u$ is contained in a constant number of ($4 {\times} 4$)-cell region in $R_i$, hence, it requires $O(n \a^2)$ time to find out the spanning paths from $u$. As such, it requires $O(n^2 \a^2)$ time to find out all the spanning paths. As for (ii), AH takes only $O(n)$ time to derive the strict total order at level $i$ of $\H$, since the derivation is based on a linear time algorithm for vertex cover. As for (iii), to construct shortcuts at the $i$-th level of $\H$, AH needs to inspect a graph $G^*_i$ reduced from $G$. For each node $u$ in $G^*_i$, AH invokes a Dijkstra's algorithm to start a  traversal in the ($5 {\times} 5$)-cell region in $R_{i+1}$ that is centered at $u$. After that, it creates shortcuts from $u$ by traversing the nodes in the Dijkstra tree obtained. By Lemma~\ref{lmm:ah-node-num}, there are $O(\a^2)$ nodes in the tree. Hence, it costs $O(\a^2)$ time to generate level-$i$ shortcuts from $u$. As such, the time required to create shortcuts at level $i$ of $\H$ is $O(n \a^2)$. There are $h$ levels, and therefore, AH costs $O(h n\a^2)$ to create shortcuts. Since $\a$ is a constant, AH totally requires $O(h n^2)$ time to construct $\H$.
\end{proof}

%\subsubsection{Correctness} \label{subsec:correctness}
%In this section we prove that the distance query algorithm and the shortest path query algorithm are correct.
To prove Theorems~\ref{thrm:distance_correct} and~\ref{thrm:path_correct}, we need the following lemma:
\begin{lemma} \label{lemma:unimodal_path}
  For any two nodes $s$ and $t$, there exits a path $P$ from $s$ to $t$ such that (i) the weight of $P$ equals the weight of the shortest path from $s$ to $t$ in $G$, and (ii) the rank sequence of $P$ is unimodal.
\end{lemma}
\begin{proof}
  Among the nodes on the shortest path from $s$ to $t$ in $G$, let $u$ be the one with the highest rank. Let $P_{s,u}$ be the shortest path from $s$ to $u$ in $G$. It suffices to show that there is a path $P_f$ from $s$ to $u$, such that (i) the weight of $P_f$ equals the weight of $P_{s,u}$, and (ii) the rank sequence of $P_f$ is increasing. We use $rank(u)$ to denote the rank of a node $u$.

  Among the nodes on $P_{s,u}$ let $v$ be such a node that: (i) $rank(s)<rank(v)$, and (ii) if $v'$ is another node on $P_{s,u}$ that has a higher rank value than that of $s$, $v$ is closer to $s$ than $v'$. Let $P_{s,v}$ be the sub-path of $P_{s,u}$ from $s$ to $v$. Then $P_{s,v}$ is also a shortest path, and among the nodes on $P_{s,v}$ $v$ has the highest rank and $s$ comes the second. Otherwise it violates the second property of $v$. If $P_{s,v}$ contains multiple edges, by the algorithm $\langle s,v \rangle$ is a level-$rank(s)$ edge, and the weight of $\langle s,v \rangle$ equals the weight of $P_{s,v}$. It means that $\langle s,v\rangle$ is an edge on $P_f$. Then in a similar way we continue to consider the path from $v$ to $u$. Therefore, the lemma is proved.
\end{proof}

	\begin{theorem} \label{thrm:distance_correct}
AH correctly answers any distance query.
\end{theorem}
\begin{proof}
  Suppose $P_{s,t}$ is a shortest path from $s$ to $t$ in $G$. Lemma~\ref{lemma:unimodal_path} shows there exists a unimodal rank sequence path $P$ from $s$ to $t$ where the weight of $P$ equals the weight of $P_{s,t}$. i.e., the \emph{Rank Constraint} would not affect the query correctness. Now we show that the \emph{Proximity constraint} would not affect the correctness either.

  Let $u$ be the node with the highest rank among all the nodes on $P_{s,t}$. Let $P_f$ be the shortest path from $s$ to $u$ and the rank sequence of $P_f$ is increasing. It suffices to show that the proximity constraint in the forward search would not affect the discovery of $P_f$. By contradiction, suppose that $v$ is a level-$i$ node on $P_f$, but $v$ is beyond the ($5\times 5$)-cell region in $R_{i+1}$ centered at the cell that contains $s$. Let $P_{s,v}$ be the shortest path from $s$ to $v$. Then, in this case, $P_{s,v}$ contains multiple edges, otherwise, since $s$ and $v$ are far-apart in $R_{i+1}$, and followed by Statement 2 of Lemma~\ref{lemma:spanning_path}, $s$ and $v$ are both level-($i{+}1$) cores, which violates the assumption $v$ as at level $i$. Besides, by Statement 4 of Lemma~\ref{lemma:spanning_path}, there is a node level-($i{+}1$) core $x$ on $P_{s,v}$ and $x\not=v$. Then, it violates the rank-increasing property of $P_f$ because $x$ comes before $v$ on $P_f$.

  Finally, we show that the elevating edges would not affect the correctness. Let $R_j$ ($j \in [1, h]$) be the coarsest grid where no ($4 {\times} 4$)-cell region contains both $s$ and $t$. By Lemma~\ref{lmm:ah-longpath}, the shortest path from $s$ to $t$ should passes through at least one node whose level is at least $j$. This indicates that AH's traversal from $s$ would meet its traversal from $t$ at level $i$ or above of $H$. As such, omitting visiting the nodes that are lower than level $j$ would not affect the correctness. Therefore, the lemma is proved.
\end{proof}

\begin{theorem} \label{thrm:path_correct}
AH correctly answers any shortest path query.
\end{theorem}
\begin{proof}
 By theorem~\ref{thrm:distance_correct}, given two nodes $s$ and $t$, the query algorithm can discover a unimodal node rank sequence path $P$. It remains to show that every shortcut $e$ on $P$ can reconstruct the path $e$ contracts. There are two types of shortcuts: elevating edges and non-elevating edges (i.e. the level-$i$ edges).

 As for the non-elevating edges, it suffices to show that the non-elevating edges on the rank-increasing path discovered by the forward search from $s$ can be reconstructed. Suppose that $\langle u, v \rangle$ is a non-elevating edge on the path discovered by the forward search. Then, $\langle u,v \rangle$ contracts the shortest path $P_{u,v}$ from $u$ to $v$. In addition, the rank of $v$ is higher than that of $u$. Suppose that $\langle u, v \rangle$ is marked with a node $w$. Then the sub-path $P_{w,v}$ (resp. $P_{u,w}$) from $w$ to $v$ (resp. from $u$ to $w$) of $P_{u,v}$ is also a shortest path. It suffices to show that $\langle w,v \rangle$ (resp.\ $\langle u,w \rangle$) is an original edge or a non-elevating edge. If $P_{w,v}$ contains multiple edges, since by the meaning of $w$, among the nodes on $P_{w,v}$, $v$ has the highest rank and $w$ comes the second, hence, AH would add a level-$i$ edge $\langle w,v \rangle$ where $w$ is at level-$i$. On the other hand, we can also have a similar conclusion of $\langle u,w \rangle$ if we build the non-elevating edge to $w$ from a \emph{Backward SPT} rooted at $w$.

 As for the elevating edges, we use mathematical induction to prove that every elevating edge $e$ on the shortest path $P$ could be reconstructed. First, let $\langle u, v \rangle$ be an elevating edge where (i) $u$ is a node at level-$0$, (ii) $v$ at level-$1$, and (iii) $\langle u,v \rangle$ is marked with a node $w$. Then $w$ is the first node that has higher rank than that of $u$ on the shortest path from $u$ to $v$. In that case, $\langle u, w \rangle$ is a level-$0$ edge (non-elevating edge), and as discussed before, $\langle u, w \rangle$ could be reconstructed. If $w\not= v$, we continue to consider $\langle w, v \rangle$. Since $w$ is also at level-$0$, and $v$ is the first node at level-$1$ on the shortest path from $w$ to $v$. Then $\langle w, v \rangle$ is also an elevating edge. Similar to the discussion aforementioned, $\langle w, v \rangle$ could also be reconstructed. As such, all the elevating edges to a node at level $1$ on the shortest path could be reconstructed.
 Suppose that the elevating edges $e$ to a level-$i$ node on a shortest path $P$ could be reconstructed. We show that the elevating edges $e$ to a level-$(i+1)$ node could also be reconstructed. Let $\langle u, v \rangle$ be an elevating edge where $v$ is at level-($i{+}1$). There are two types of elevating edges: (a) $u$ is a boarder node of $R_{i+1}$, and (b) $u$ is a node at level $i$. We firstly consider (a). Let $T$ be an SPT from $u$. Let $w$ be the node immediately follows $u$ on the branch from $u$ to $v$ in $T$. Then $w$ is a level-$i$ core. In that case, $\langle u, w \rangle$ is also an elevating edge because: (i) $u$ is a boarder node of $R_{i+1}$ implies that $u$ is also a boarder node of $R_i$, and (ii) $w$ is a node at level-$i$ that is closet to $u$ on the path from $u$ to $v$. By induction, $\langle u, w \rangle$ could be reconstructed. If $w\not v$, we continue to consider $\langle w, v \rangle$. Since the shortest path from $w$ to $v$ does not go through another node at level $i{+}1$, $\langle w, v \rangle$ is an elevating edge of type (b). It remains to show that the type (b) elevating edges could be reconstructed. Let $\langle u, v \rangle$ be a type (b) elevating edge, i.e., $u$ is at level $i-1$ and $v$ at level $i$. By the algorithm, $\langle u, v \rangle$ is marked with a node $w$, the first node that has a higher rank than that of $u$ on the shortest path from $u$ to $v$. Then, $\langle u, w \rangle$ is a level-$i$ edge and could be reconstructed. If $w\not= v$, we continue to consider $w$, and similarly, $\langle w, v \rangle$ is also an elevating edge of type (b) marked with $w'$. As such, there is an edge $\langle w',v \rangle$ such that $w'$ is at level $i$, and the nodes on the shortest path from $w'$ to $v$ (excluding $w'$ and $v$) have lower ranks than that of $w'$. In what follows, $\langle w',v \rangle$ is also a level-$i$ edge, and could be reconstructed. Similarly, all the elevating edges from a level-($i{+}1$) node to the boarder nodes of $R_{i+1}$ and the level-$i$ nodes could also be reconstructed.

 Therefore, the lemma is proved.
\end{proof}

%============================================

\end{sloppy}
\end{document}